\pdfoutput=1
\documentclass{article}
\usepackage{style}
\usepackage{shortcuts}

\title{Universe Reduction for APSP:\\Equivalence of Three Fine-Grained Hypotheses}
\author{Nick Fischer\thanks{Max Planck Institute for Informatics. Email: \texttt{nfischer@mpi-inf.mpg.de}.}}
\date{}

\begin{document}

\maketitle
\begin{abstract}
\noindent
The \emph{APSP Hypothesis} states that the All-Pairs Shortest Paths (APSP) problem requires time $n^{3-o(1)}$ on graphs with polynomially bounded integer edge weights. Two increasingly stronger assumptions are the \emph{Strong APSP Hypothesis} and the \emph{Directed Unweighted APSP Hypothesis}, which state that the fastest-known APSP algorithms on graphs with small weights and unweighted graphs, respectively, are best-possible. In this paper, we design an efficient \emph{universe reduction} for APSP, which proves that these three hypotheses are, in fact, \emph{equivalent}, conditioned on $\omega = 2$ and a plausible additive combinatorics assumption. 

Along the way, we resolve the fine-grained complexity of many long-standing graph and matrix problems with ``intermediate'' complexity such as Node-Weighted APSP, All-Pairs Bottleneck Paths, Monotone Min-Plus Product in certain settings, and many others, by designing matching APSP-based lower bounds.
\end{abstract}

\thispagestyle{empty}
\clearpage

\tableofcontents
\thispagestyle{empty}
\clearpage

\setcounter{page}{1}
\section{Introduction} \label{sec:intro}
The overarching goal of fine-grained complexity is to pinpoint the exact complexity of polynomial-time problems conditioned on a small set of believable core assumptions. Over the years, three such hypotheses have crystallized: SETH~\cite{ImpagliazzoP01,ImpagliazzoPZ01}, the 3SUM Hypothesis~\cite{GajentaanO95}, and the APSP Hypothesis~\cite{VassilevskaW18}. In this modern view, a problem is considered ``closed'' when we can establish both an~$n^{c+o(1)}$ upper bound and a matching~$n^{c-o(1)}$ lower bound conditioned on one of these three assumptions. The framework has become a major success story, resolving the exact complexity of numerous important problems in P.

Yet many other central problems currently resist lower bounds under these ``primary'' assumptions. A common workaround is to resort to stronger ``secondary'' hypotheses (such as Quantified SETH~\cite{BringmannC20}, Strong 3SUM~\cite{AmirCLL14}, and Minimum-Weight $k$-Clique~\cite{AbboudVW14}, to name a few). Often, these secondary assumptions postulate that the primary problems are hard even on restricted instances (as is the case for Strong 3SUM and Minimum-Weight $k$-Clique). On the one hand, these variants are extremely helpful: they are typically easier to reduce from, and we still regard them as ``morally equivalent'' to their primary counterparts.

On the other hand, these secondary hypotheses are, of course, less convincing. Unlike the three core assumptions, they often lack decades of effort dedicated to breaking them. As a result, while conditional lower bounds based on secondary hypotheses are still seen as strong hardness barriers, they must be taken with a grain of salt. Another drawback is conceptual: With every new hypothesis we move a little further away from the dream goal of classifying the complexity in P conditioned on a \emph{small} set of core assumptions. This naturally raises an important quest in fine-grained complexity: to \emph{relate} our primary and secondary hypotheses to each other, ideally reducing all of our secondary assumptions back to the original three primary ones. Unfortunately, so far there has been little progress in this direction.

In this paper, we take one of the first significant steps. We consider the APSP Hypothesis and two of its prominent secondary variants, and prove that these three hypotheses are not only morally related but are in fact \emph{equivalent}, modulo the assumption $\omega = 2$ and some plausible additive combinatorics assumption. Along the way we derive strong APSP-based lower bounds---many of which resolve long-standing open problems.

\subsection{All-Pairs Shortest Paths} \label{sec:intro:sec:apsp}
The \emph{All-Pairs Shortest Paths (APSP)} problem is to compute all pairs of distances in an edge-weighted graph. Few other graph problems have sparked as much attention as the quest for fast algorithms for APSP. Yet in the 60 years since the classic $O(n^3)$-time Floyd--Warshall algorithm~\cite{Floyd62,Warshall62}, a long line of improvements~\cite{Floyd62,Warshall62,Fredman76,Dobosiewicz90,Takaoka92,Han04,Takaoka04,Zwick04,Takaoka05,Chan08,Han08,Chan10,HanT16} only shaved off polylogarithmic factors, culminating in Williams’~\cite{Williams18,ChanW21} current best bound of \smash{$n^3 / 2^{O(\sqrt{\log n})}$}. Whether this can be improved to truly subcubic time~$O(n^{2.99})$, say, remains a major open question. With the rise of fine-grained complexity in the 2010s, this barrier was reinterpreted as a hardness assumption~\cite{VassilevskaW18}, which has since evolved into one of the three primary hypotheses of the area:

\begin{hypothesis}[APSP] \label{hypo:apsp}
For every constant $\epsilon > 0$ there is a constant $c$ such that APSP on graphs with weights in $\set{0, \dots, n^c}$ cannot be solved in time $O(n^{3-\epsilon})$.\footnote{Unless stated otherwise we assume that graphs are directed. However, it is known that \cref{hypo:apsp} holds equivalently for directed and undirected graphs~\cite{VassilevskaW18}.}\footnote{Here and throughout we only consider nonnegative weights, although it is equivalent to assume weights in $\set{-n^c, \dots, n^c}$. Indeed, with Johnson's trick~\cite{Johnson77}, it suffices to solve a Single-Source Shortest Paths instance (e.g., using~\cite{BernsteinNW22}) to compute a potential function that transforms negative weights into nonnegative integers.}\footnote{In many papers the dependence between $\epsilon$ and $c$ is not made explicit. In particular, another reasonable formulation of the hypothesis is to change the quantifiers, i.e., to postulate that there is some absolute constant $c$ such that for all $\epsilon > 0$ APSP on graphs with weights $\set{0, \dots, n^c}$ cannot be solved in time $O(n^{3-\epsilon})$. We consider the weaker \cref{hypo:apsp} throughout (so that our results take the strongest form). In fact, one could weaken the hypothesis further to weights bounded by \smash{$\exp(n^{o(1)})$} and our results would still apply.}
\end{hypothesis}

The APSP Hypothesis forms the basis for fine-grained lower bounds for many important problems, including natural graph and matrix problems~\cite{VassilevskaW18,BackursDT16,HenzingerLNV17,LincolnVW18,AbboudGV23,KociumakaP23}, tree edit distance~\cite{BringmannGMW18,NoglerPSVXY25}, triangle listing~\cite{VassilevskaX20b}, as well as approximate~\cite{AbboudBKZ22,ChanX24} and dynamic problems~\cite{AbboudV14,VassilevskaX20b}, even in planar graphs~\cite{AbboudD16}. In fact, many such problems have been shown not merely to be APSP-hard, but APSP-equivalent.

\medskip
The two secondary hypotheses we consider can both be seen as strengthenings of the APSP hypothesis concerning graphs with \emph{small} weights and \emph{no} weights. (In fact, the restriction to small-weight instances is a common secondary pattern, see e.g.\ the Strong 3SUM~\cite{AmirCLL14}, Strong Zero Triangle~\cite{AbboudBKZ22} and Strong Min-Plus Convolution~\cite{ChanVX21} Hypotheses.) Specifically, the fastest-known algorithm for graphs with weights in $\set{0, \dots, u}$ runs in time $\tilde O(n^\omega u)$~\cite{AlonGM97}, where $2 \leq \omega < 2.372$ is the matrix multiplication exponent~\cite{AlmanDVXXZ25}. Thus, for a universe of size $u = n^{3-\omega}$ the fastest-known algorithms take cubic time~$n^{3-o(1)}$, and it is plausible to hypothesize that this is best-possible. Chan, Vassilevska W., and Xu~\cite{ChanVX23} first proposed this as the following hardness assumption, which has become the basis for several conditional lower bounds~\cite{ChanVX23,HorowiczK25}:

\begin{hypothesis}[Strong APSP] \label{hypo:strong-apsp}
For every constant $\epsilon > 0$, APSP on graphs with weights in $\set{0, \dots, n^{3-\omega}}$ cannot be solved in time $O(n^{3-\epsilon})$.
\end{hypothesis}

Conversely, APSP on graphs with weights smaller than $n^{3-\omega}$ and therefore also \emph{unweighted} directed graphs is known to be in subcubic time. The fastest-known algorithm for this problem is due to Zwick~\cite{Zwick02} and runs in time $\tilde O(n^{2+\mu}) = \tilde O(n^{2.528})$. This running time depends on a rectangular matrix multiplication constant\footnote{Namely, the constant $\mu$ satisfying that $\omega(1, \mu, 1) = 1 + 2\mu$, where $\omega(\cdot, \cdot, \cdot)$ is the rectangular matrix multiplication exponent.} $0.5 \leq \mu \leq 0.528$ and becomes~$\tilde O(n^{2.5})$ if~\makebox{$\omega = 2$}. This time bound has not been challenged for more than 20 years, which recently led Chan, Vassilevska W., and Xu~\cite{ChanVX21} to hypothesize that Zwick's algorithm is best-possible.

\begin{hypothesis}[Directed Unweighted APSP] \label{hypo:dir-unw-apsp}
For every constant $\epsilon > 0$, APSP on directed unweighted graphs cannot be solved in time $O(n^{2+\mu-\epsilon})$.
\end{hypothesis}

This hypothesis turns out to be surprisingly powerful---it is equivalent to various $\tilde O(n^{2+\mu})$-time problems~\cite{ChanVX21}, and conditionally resolves the complexity of many well-studied graph problems with ``intermediate'' complexity; see the discussion below in \cref{sec:intro:sec:intermediate}.

\paragraph{Min-Plus Product}
To relate these three hypotheses, it is helpful to introduce some more background. The \emph{min-plus product} of two matrices $A$ and $B$ of size $n_1 \times n_2$ and $n_2 \times n_3$, respectively, is the matrix~\makebox{$A * B$} of size $n_1 \times n_3$ defined by $(A * B)[i, j] = \min_k(A[i, k] + B[k, j])$. We refer to the problem of computing $A * B$ as the \emph{Min-Plus Product} problem, and we denote its complexity by $\MinPlus(n_1, n_2, n_3)$. We sometimes annotate bounds on the maximum entry $u$ in the three matrices by $\MinPlus(n_1, n_2, n_3 \mid u)$. It is a long-known fact that the complexity of APSP is \emph{exactly} $\MinPlus(n, n, n)$ (up to $n^{o(1)}$ factors). Via two elegant reductions~\cite{ShoshanZ99,ChanW21} one can similarly give equivalent characterizations of the Strong and Directed Unweighted APSP Hypotheses in terms of some Min-Plus Product problems:

\begin{hypothesis}[APSP] \label{hypo:apsp-min-plus}
$\MinPlus(n, n, n) = n^{3 \pm o(1)}$.
\end{hypothesis}

\begin{hypothesis}[Strong APSP~\cite{ShoshanZ99}] \label{hypo:strong-apsp-min-plus}
$\MinPlus(n, n, n \mid u \leq n^{3-\omega}) = n^{3 \pm o(1)}$.
\end{hypothesis}

\begin{hypothesis}[Directed Unweighted APSP~\cite{ChanVX21}] \label{hypo:dir-unw-apsp-min-plus}
$\MinPlus(n, n^\mu, n \mid u \leq n^{1-\mu}) = n^{2 + \mu \pm o(1)}$.
\end{hypothesis}

From these characterizations it is easy to read off that, when $\omega = 2$, the Directed Unweighted APSP Hypothesis implies the Strong APSP Hypothesis.\footnote{Note that $\MinPlus(n, \sqrt{n}, n \mid u \leq \sqrt{n}) = O(n \cdot \MinPlus(\sqrt{n}, \sqrt{n}, \sqrt{n} \mid u \leq \sqrt{n}))$. Thus, \makebox{$\MinPlus(n, n, n \mid u \leq n) = O(n^{3-\epsilon})$} implies that $\MinPlus(n, \sqrt{n}, n \mid u \leq \sqrt{n}) = O(n^{2.5-\epsilon/2})$.} In summary, the three assumptions get \emph{increasingly stronger}---i.e., the Directed Unweighted APSP Hypothesis implies the Strong APSP Hypothesis (if $\omega = 2$) which in turn implies the APSP Hypothesis. 

\subsection{Universe Reductions for APSP} \label{sec:intro:sec:univ-reduct}
The main goal of this paper is to establish that these three hypotheses are actually equivalent. It is reasonable to expect that such an equivalence result hinges on the assumption that $\omega = 2$; we will later comment on this issue in detail, but for now assume that $\omega = 2$. The two missing ``hard'' directions can both be seen as \emph{universe reductions} for APSP: from large (polynomially bounded) weights to weights bounded by $n$ (\emph{large-universe reduction}), and from weights bounded by $n$ to no weights at all (\emph{small-universe reduction}). We start with the latter task which turns out to be the slightly simpler one.

\subsubsection{Small-Universe Reduction: From Small to No Weights}
There has been partial progress towards a small-universe reduction. In the same paper that introduces the Strong APSP Hypothesis, Chan, Vassilevska W., and Xu~\cite{ChanVX23} develop an exciting technique called ``Fredman's trick meets dominance product'' and thereby establish a non-tight reduction: Directed Unweighted APSP cannot be solved in time $O(n^{7/3-\epsilon})$ unless the Strong APSP Hypothesis fails. Ultimately, they left open if the two hypotheses are equivalent, but their technique is inspiring. We propose a new approach, \emph{low-rank APSP,} and based on this new technique we successfully prove the desired equivalence result:

\begin{theorem}[Strong APSP Implies Directed Unweighted APSP] \label{thm:strong-apsp-implies-dir-unw-apsp}
APSP in directed unweighted graphs cannot be solved in time $O(n^{2.5-\epsilon})$ (for any constant $\epsilon > 0$), unless the Strong APSP Hypothesis fails. In particular, conditioned on $\omega = 2$, the Strong APSP and Directed Unweighted APSP Hypotheses are equivalent.
\end{theorem}

\cref{thm:strong-apsp-implies-dir-unw-apsp} constitutes one of the very rare cases that two fine-grained hypotheses are shown to be equivalent---in fact, we are not aware of any prior case where two explicitly established fine-grained assumptions have later been shown to be equivalent. Moreover, \cref{thm:strong-apsp-implies-dir-unw-apsp} constitutes an equivalence result of two problems with the same input and output size but \emph{different} complexities ($n^3$ versus~$n^{2.5}$). This, too, is quite exceptional; there are only few other examples such as the Mono-Convolution problem by Lincoln, Vassilevska W., and Polak~\cite{LincolnPV20}.

More generally, our universe reduction correctly interpolates between unweighted and $[n]$-weighted graphs, both in the directed and undirected settings. Specifically, for directed graphs with weights $\set{0, \dots, u}$ Zwick's algorithm~\cite{Zwick02} takes time $\tilde O(n^{2+1/(4-\omega)} u^{1/(4-\omega)})$. If $\omega = 2$ then this running time becomes $\tilde O(n^{2.5} \sqrt{u})$. For undirected graphs with weights $\set{0, \dots, u}$ the Shoshan--Zwick algorithm~\cite{ShoshanZ99} takes time $\tilde O(n^\omega u)$. We show both algorithms are best-possible if $\omega = 2$, for all values of $u$:

\begin{restatable}[Zwick's Algorithm is Optimal]{theorem}{thmzwickoptimal} \label{thm:zwick-opt}
Let $0 \leq \delta \leq 1$ be a constant. APSP in directed graphs with weights $\set{0, \dots, n^\delta}$ cannot be solved in time $O(n^{2.5+\delta/2-\epsilon})$ (for any constant $\epsilon > 0$), unless the Strong APSP Hypothesis fails.
\end{restatable}

\begin{restatable}[Shoshan--Zwick Algorithm is Optimal]{theorem}{thmshoshanzwickoptimal} \label{thm:shoshan-zwick-opt}
Let $0 \leq \delta \leq 1$ be a constant. APSP in undirected graphs with weights $\set{0, \dots, n^\delta}$ cannot be solved in time $O(n^{2+\delta-\epsilon})$ (for any constant $\epsilon > 0$), unless the Strong APSP Hypothesis fails.
\end{restatable}

\subsubsection{Large-Universe Reduction: From Large to Small Weights}
Next, focus on the universe reduction from polynomially large weights to weights bounded by $n$. This reduction appears even more difficult. In fact, we are neither aware of any prior progress on this question, nor are we aware of any other (exact) graph distance problem where a similar weight reduction is known. In the related, more restricted setting where one graph has to be mapped one-to-one to a graph with smaller weights such a reduction turns out to be impossible~\cite{BernsteinBW24}.

Our approach involves two steps. The first builds on the new \emph{low-rank APSP} technique from before. The second builds on \emph{additive combinatorics}. In fact, the \emph{low-rank APSP} technique also relies on some additive combinatorics, so this work should be seen as part of a recent trend to exploit additive combinatorics in fine-grained complexity and algorithm design~\cite{ChanL15,BringmannN20,DudekGS20,BringmannW21,BringmannN21,AbboudBF23,ChanVX23,JinX23,ChenLMZ24a,Bringmann24,Jin24,ChenLMZ24b,ChenLMZ24c,ChenLMZ24d,BringmannFN25,Fischer25,FischerJX25,AbboudFJVX25}.

\paragraph{Step 1: Doubling Reduction}
We start with the first and technically more involved step. Specifically, we establish that APSP on general graphs with polynomially bounded weights can be reduced to graphs with at most $n$ distinct weights (i.e., much fewer than the trivial upper bound of $n^2$). In addition, we can enforce that there are only few distinct \emph{sums} of two weights. To formalize this requirement we use the following notation: For an integer set $X$ we write $X + X = \set{x + y : x, y \in X}$ to denote its so-called \emph{sumset}. In the additive combinatorics literature the ratio $|X + X| / |X|$ is typically called the \emph{doubling} of $X$ (hence the term ``doubling reduction''). Formally we prove:

\begin{theorem}[Uniform Low-Doubling APSP] \label{thm:apsp-unif-low-doubling}
Let $\kappa > 0$ be a constant. APSP in graphs with integer weights $X \subseteq \Int$ where $|X| \leq n$ and $|X + X| \leq n^\kappa |X|$ cannot be solved in time $O(n^{3-\epsilon})$ (for any constant~\makebox{$\epsilon > 0$}), unless the APSP Hypothesis fails.
\end{theorem}

\cref{thm:apsp-unif-low-doubling} is already a strong and arguably surprising result. Even ignoring the extra constraint on the doubling of $X$, it demonstrates that the worst-case APSP instances have only $n$ distinct weights (rather than the naive upper bound of $n^2$)! An appropriate technical generalization of \cref{thm:apsp-unif-low-doubling} is the key tool for various interesting fine-grained lower bounds (that we will survey soon). We are optimistic that it will be similarly helpful in the design of future APSP-based lower bounds.

\paragraph{Step 2: Sum-Order-Preserving Hashing}
But how is \cref{thm:apsp-unif-low-doubling} useful for the desired universe reduction? To understand this, let us take a step back and forget about \cref{thm:apsp-unif-low-doubling} for a minute, to review a naive \emph{hashing-based} approach for the universe reduction. At first thought hashing appears to be a reasonable idea---after all, for many problems (such as 3SUM) we can analogously obtain weight reductions using additive hash functions. But as we will see, naively this approach is bound to fail.

In the language of \cref{hypo:apsp-min-plus,hypo:strong-apsp-min-plus} our goal is to reduce the computation of a min-plus product of two arbitrary $n \times n$ matrices $A, B$ to the computation of a min-plus product with entries bounded by $n$. Let $X \subseteq \Int$ be the set of entries that appear in $A$ and $B$. Suppose there was a hash function~\makebox{$h : X \to \set{0, \dots, n}$} satisfying the property that for all $x_1, x_2, y_1, y_2 \in X$:
\begin{equation*}
    x_1 + x_2 < y_1 + y_2 \quad\text{implies}\quad h(x_1) + h(x_2) < h(y_1) + h(y_2).
\end{equation*}
We call such a function $h$ \emph{sum-order-preserving.} We could then define matrices $A' = h(A)$ and $B' = h(B)$ (where we apply $h$ entry-wise) which, as desired, have entries bounded by $n$, and we could read off $A * B$ from $A' * B'$ with the aid of some standard tricks.\footnote{Specifically, call $k$ a \emph{witness} of $(i, j)$ in $A * B$ if $(A * B)[i, j] = A[i, k] + B[k, j]$. Then the sum-order-preserving property implies that any witness $k$ of $(i, j)$ in $A' * B'$ is also a witness of $(i, j)$ in $A * B$. Moreover, by a standard trick we can turn any algorithm to compute the min-plus product of two matrices into an algorithm that also reports a witness for each entry $(i, j)$; see \cref{lem:min-plus-wit}.}

Unfortunately, clearly such a hash function cannot exist in general! There are up to $|X|^2$ distinct sums of the form $x + y$, which in the worst case is up to $n^4$, but only $O(n)$ possible values for $h(x) + h(y)$. By the pigeonhole principle there is necessarily a collision $h(x_1) + h(x_2) = h(y_1) + h(y_2)$ for some $x_1 + x_2 < y_1 + y_2$. So it seems that the hashing-based idea is fundamentally flawed.

This is where \cref{thm:apsp-unif-low-doubling} comes into play. It allows us to assume that the sumset~$X + X$ does not have size up to $n^4$ as the trivial bound would suggest, but rather only roughly \emph{linear} in $n$. For this reason it would be sufficient to have a sum-order-preserving hash function $h : X \to \set{0, \dots, |X + X|}$. The previous pigeonhole argument does \emph{not} rule out such a function. So could it possibly be the case that such a hash function always exists? At least in the relevant special case when $X$ has small doubling, i.e., when $|X + X|$ is small?

The surprising answer is: yes, almost! By combining two results from additive combinatorics---the first by Amirkhanyan, Bush, and Croot~\cite{AmirkhanyanBC18} on so-called order-preserving Freiman isomorphisms, and second, the quasi-polynomial bounds for the Freiman-Ruzsa theorem due to Sanders~\cite{Sanders12}---one obtains the following theorem stating that for any small-doubling set $X$ there exists a reasonably large subset $Y \subseteq X$ that can be hashed by a sum-order-preserving hash function. (See \cref{sec:hashing:sec:quasi-poly} for how to derive \cref{thm:sum-order-preserving-quasi-poly} from~\cite{AmirkhanyanBC18,Sanders12}.)

\begin{restatable}[Sum-Order-Preserving Hashing with Quasi-Polynomial Bounds~\cite{AmirkhanyanBC18,Sanders12}]{theorem}{thmsumorderpreservingquasipoly} \label{thm:sum-order-preserving-quasi-poly}
For every integer set~$X$ with doubling $|X + X| \leq K |X|$ there is a subset $Y \subseteq X$ of size $|Y| \geq |X| / \exp((\log K)^{O(1)})$ and a sum-order-preserving function $h : Y \to \set{0, \dots, |X|}$.
\end{restatable}

Unfortunately, \cref{thm:sum-order-preserving-quasi-poly} is slightly too weak. In principle it is acceptable that only a subset $Y \subseteq X$ admits the sum-order-preserving hash function (as with some further overhead we can cover $X$ by few translates of $Y$); however, the quantitative bounds on $|Y|$ are not strong enough for our purposes. \cref{thm:sum-order-preserving-quasi-poly} entails that $Y$ covers at least a fraction of~$X$ that depends \emph{quasi-polynomially} on $K$, but we would require a dependence that is at most \emph{polynomial.}

This is a common challenge in additive combinatorics. In fact, it is a major conjecture in additive combinatorics, termed the \emph{Polynomial Freiman-Ruzsa (PFR) Conjecture} (see e.g.~\cite{Zhao23} and also~\cite{LovettR18,Manners17}), that Sanders'~\cite{Sanders12} aforementioned quasi-polynomial bounds for the Freiman-Ruzsa theorem can be improved to polynomial.

Notably, in a recent breakthrough, Gowers, Green, Manners, and Tao~\cite{GowersGMT25} have positively resolved the PFR conjecture in characteristic $2$. Their result does not have immediate implications for the integer case---in fact, they leave the integer case as a ``challenging open problem''. In an even more recent preprint, Raghavan~\cite{Raghavan25} further improved the integer PFR bound to $\exp((\log K)^{1+o(1)})$. Given these promising developments there is reason to be optimistic that the PFR Conjecture will eventually be positively resolved and that ideally along the way the bounds in \cref{thm:sum-order-preserving-quasi-poly} can be improved to polynomial. We formulate this as the following plausible hypothesis:

\begin{restatable}[Sum-Order-Preserving Hashing with Polynomial Bounds]{hypothesis}{hypohashing} \label{hypo:hashing}
There is a constant $c$ such that for every integer set $X$ with doubling $|X + X| \leq K |X|$ there is a subset $Y \subseteq X$ of size $|Y| \geq \Omega(|X| / K^c)$ and a sum-order-preserving function $h : Y \to \set{0, \dots, |X|}$. Moreover, given $X$ one can compute~$Y$ and $h$ in time $|X|^{1+o(1)} K^{O(1)}$.
\end{restatable}

We defer further discussion of \cref{hypo:hashing} to \cref{sec:hashing:sec:poly}, and instead finally return to the APSP universe reduction. Combined with \cref{thm:apsp-unif-low-doubling} and a couple of other tricks, \cref{hypo:hashing} implies the desired universe reduction for APSP:

\begin{restatable}[APSP Conditionally Implies Strong APSP]{theorem}{thmapspimpliesstrongapsp} \label{thm:apsp-implies-strong-apsp}
Conditioned on \cref{hypo:hashing}, there is a subcubic reduction from APSP to APSP over the universe $\set{0, \dots, n}$. In particular, conditioned on \cref{hypo:hashing} and on the assumption that $\omega = 2$, the APSP and Strong APSP Hypotheses are equivalent.
\end{restatable}

\subsubsection{Critical Remarks}
Ultimately we reach our goal, proving that the APSP, Strong APSP and Directed Unweighted APSP Hypotheses are equivalent, only conditioned on two strong assumptions: $\omega = 2$ and \cref{hypo:hashing}. Let us take a moment to scrutinize these two assumptions.

On the one hand, it is to be expected that we have to assume $\omega = 2$. The reason is that the Strong APSP and Directed Unweighted APSP Hypotheses both depend on two \emph{different} matrix multiplication constants, $\omega$~and~$\mu$, and so the equivalence would (in a sense) tightly relate $\omega$ and $\mu$. However, these constants are only known to be loosely related by trivial bounds like $\mu \leq 1/(4-\omega)$. Only if $\omega = 2$ (and thus $\mu = \frac12$) would this trivial bound become tight, and so only if $\omega = 2$ would the hypothetical equivalence not have surprising implications for the matrix multiplication constants. Finally, we emphasize that all of our results remain meaningful even in a world with~\makebox{$\omega > 2$}.

On the other hand, despite our confidence that \cref{hypo:hashing} is plausible, it is a new assumption and should thus be considered with some care. We find it somewhat ironic that we show that some fine-grained hypotheses are equivalent at the cost of introducing yet another new hypothesis. However, we believe that there is significant value in this conditional equivalence result. First, \cref{hypo:hashing} is not a \emph{fine-grained complexity} assumption, but rather a purely \emph{mathematical} assumption, which we hope will be eventually resolved or refuted by the additive combinatorics community alongside their effort to resolve the polynomial Freiman-Ruzsa conjecture. In this sense, we view it as remotely analogous to the (admittedly much more established) Extended Riemann Hypothesis that some modern algorithms rely on.

Second, intuitively it seems unavoidable to rely on \cref{hypo:hashing} or a similar additive combinatorics assumption. The vague reason is that we can embed $[u]$ in many ways into low-doubling sets $X$, for instance, by stretching $[u]$ to an arithmetic progression $X = \set{a, 2a, \dots, ua}$. A universe reduction would effectively \emph{invert} this embedding. But then it would implicitly answer the structural question behind the Freiman-Ruzsa theorem, namely if all low-doubling sets can be covered by progression-like sets.

Third, and perhaps most importantly, for the APSP-based lower bounds we develop along the way it is \emph{not} necessary to assume \cref{hypo:hashing}. We will now elaborate on these lower bounds in detail.

\subsection{APSP-Based Lower Bounds for Intermediate Problems} \label{sec:intro:sec:intermediate}
An extensive line of work in the fine-grained algorithms and complexity literature is concerned with studying graph problems of ``intermediate'' complexity, i.e., with complexity strictly between $\tilde O(n^\omega)$ and $O(n^3)$. On the upper bound side this trend started in the 2000s~\cite{Zwick02,VassilevskaWY07,ShapiraYZ11,DuanP09,Vassilevska10,DuanGZ18,DuanJW19,Chan10,Yuster09,AbboudFJVX25,VassilevskaW18,BenderFPSS05,CzumajKL07,GrandoniILPU21,BringmannKW19}. Most algorithms follow the same high-level scheme: One identifies a matrix multiplication-type problem that \emph{exactly} captures the complexity of the graph problem, and then designs specialized algorithms for that matrix problem. A prominent example is \emph{Directed Unweighted APSP} which, as mentioned before, can be solved by Zwick's algorithm~\cite{Zwick02} in time $\tilde O(n^{2+\mu})$. The corresponding matrix problem is the min-plus product of two rectangular matrices with bounded entries.

Progress on the lower bound side is more recent. Several works~\cite{BarrKPR19,LincolnPV20,VassilevskaX20b,ChanVX21} established fine-grained connections \emph{between} intermediate problems (even on the finer-grained scale that distinguishes between running times like $n^{(3+\omega)/2 \pm o(1)}$ and $n^{2+\mu \pm o(1)}$, say~\cite{VassilevskaX20b}). These connections imply matching lower bounds only conditioned on secondary hypotheses such as the Directed Unweighted APSP Hypothesis (which itself is about a problem of intermediate complexity). As mentioned before, there is one result besides these mostly simple connections that stands out: In their work on \emph{Fredman's Trick meets Dominance Product}, Chan, Vassilevska W., and Xu give non-matching lower bounds for various intermediate problems conditioned on the Strong APSP Hypothesis. This is still a secondary assumption, but their ideas are inspiring and can to some degree be seen as the baseline for this paper.

In summary, despite significant effort, prior to this work no non-trivial lower bounds under the primary hypotheses are known for any intermediate problem---not even non-matching $n^{2.1-o(1)}$ lower bounds.

From our equivalence result (\cref{thm:strong-apsp-implies-dir-unw-apsp,thm:apsp-implies-strong-apsp}) it follows immediately that we can turn the Directed Unweighted APSP-based lower bounds into APSP-based ones, \emph{assuming that $\omega = 2$ and the additive combinatorics \cref{hypo:hashing}.} This result already leads to exciting new APSP-based lower bounds for many problems that are additionally conditioned on two extra assumptions (\cref{hypo:hashing} and $\omega = 2$).

However, it turns out that we can even remove these two extra assumptions! The reason is that our reductions only depend on step 1 (the unconditional doubling reduction) from before. We thus obtain clean $n^{2.5-o(1)}$-time lower bounds for many important intermediate problems based on the APSP Hypothesis. In all cases, these match the known upper bounds if $\omega = 2$. This statement is essentially the strongest conditional lower bound one could hope for,\footnote{Perhaps a first instinct would be to hope for higher lower bounds, say $n^{(3+\omega)/2 - o(1)}$. Lower bounds of this type can only be expected for hypotheses which themselves depend on $\omega$, like the Strong APSP Hypothesis. The APSP Hypothesis is agnostic of $\omega$, hence a lower bound higher than $n^{2.5}$ would imply lower bounds on $\omega$ conditioned on the APSP Hypothesis.} so our work successfully \emph{closes} the respective intermediate problems. We now discuss three particularly interesting cases in detail, and leave further lower bounds to \cref{sec:intermediate}.

\paragraph{Lower Bound 1: Node-Weighted APSP}
The first subcubic-time algorithm for APSP on directed \emph{node-weighted} graphs is due to Chan~\cite{Chan10} and runs in time $\tilde O(n^{(9 + \omega)/4}) = O(n^{2.843})$. This was later improved by Yuster~\cite{Yuster09} using rectangular matrix multiplication. Only very recently, Abboud, Fischer, Jin, Vassilevska~W., and Xi~\cite{AbboudFJVX25} obtained an algorithm in time~$\tilde O(n^{(3+\omega)/2})$ (with slight improvements using rectangular matrix multiplication). In~\cite{ChanVX21} it was observed that Node-Weighted APSP (even in undirected graphs) is at least as hard as Directed Unweighted APSP. Here we give the strengthened lower bound based on the primary APSP Hypothesis:

\begin{restatable}[Node-Weighted APSP]{theorem}{thmnodewgtapsp} \label{thm:node-wgt-apsp}
APSP in undirected node-weighted graphs cannot be solved in time $O(n^{2.5-\epsilon})$ (for any constant $\epsilon > 0$), unless the APSP Hypothesis fails.
\end{restatable}

\paragraph{Lower Bound 2: All-Pairs Bottleneck Paths}
The \emph{All-Pairs Bottleneck Paths (APBP)} problem is to compute, for all pairs of nodes $s, t$ in a directed edge-capacitated graph, the maximum flow that can be routed on one path from $s$ to~$t$. The complexity of APBP is exactly captured by the complexity of computing the \emph{min-max product} of two matrices (see \cref{sec:intermediate:sec:min-max-prod}), and based on this insight Vassilevska, Williams, and Yuster~\cite{VassilevskaWY07} designed the first subcubic-time algorithm for APBP, running time $\tilde O(n^{2+\omega/3}) = O(n^{2.791})$. At the time they raised the question if APBP can be solved in matrix multiplication time $\tilde O(n^\omega)$. Their algorithm was later improved by Duan and Pettie~\cite{DuanP09} to time $\tilde O(n^{(3+\omega)/2})$. We show that Duan and Pettie's algorithm is best-possible (if $\omega = 2$), and consequently that time $\tilde O(n^\omega)$ cannot be achieved.

\begin{restatable}[All-Pairs Bottleneck Paths]{theorem}{thmapbp} \label{thm:apbp}
All-Pairs Bottleneck Paths cannot be solved in time $O(n^{2.5-\epsilon})$ (for any constant~$\epsilon > 0$), unless the APSP Hypothesis fails.
\end{restatable}

Via known reductions, \cref{thm:apbp} settles the complexity of even more problems. For instance, the \emph{All-Pairs Nondecreasing Paths (APNP)} problem is to compute, for all pairs of nodes $s, t$ in a directed edge-weighted graph, the smallest weight $w$ such that there is an $s$-$t$-path along which the weights are nondecreasing and at most $w$. When the edge weights model departure times, this can be interpreted as traveling from~$s$ to~$t$ in the fastest possible time. The first subcubic-time algorithm is due to Vassilevska~\cite{Vassilevska10}, and later improvements due to Duan, Gu, Zhang~\cite{DuanGZ18} and Duan, Jin, Wu~\cite{DuanJW19} have optimized the time complexity to $\tilde O(n^{(3+\omega)/2})$. APNP is known to generalize APBP~\cite{Vassilevska10}, and thus \cref{thm:apbp} provides a matching APSP-based lower bound.

As a final example, Bringmann, Künnemann, and Węgrzycki~\cite{BringmannKW19} proved that APBP is further equivalent to \emph{$(1+\epsilon)$-Approximate APSP} for \emph{strongly polynomial} algorithms (i.e., for algorithms for which the number of arithmetic operations does not depend on the universe size $u$). Thus, combined with \cref{thm:apbp} we obtain a matching lower bound.

\paragraph{Lower Bound 3: Bounded-Difference and Monotone Min-Plus Product}
An important special case of the Min-Plus Product problem is to consider matrices $A$ where all adjacent entries differ at most by a constant. Bringmann, Grandoni, Saha, and Vassilevska W.~\cite{BringmannGSV19} first designed a subcubic-time algorithm for this special case, and as a consequence developed faster algorithms for string problems like language edit distance and RNA-folding; see also~\cite{Mao21}. This inspired a line of research~\cite{VassilevskaX20a,GuPVX21,ChiDX22,ChiDXZ22} to optimize the subcubic running time that recently led to the state-of-the-art $\tilde O(n^{(3+\omega)/2})$-time algorithm due to Chi, Duan, Xie, and Zhang~\cite{ChiDXZ22}. In fact, they give two different algorithms to treat the strictly more general cases where $A$ is \emph{row-monotone} or \emph{column-monotone}. We prove that the first algorithm is optimal, conditioned on the APSP Hypothesis; we leave open whether the second algorithm is similarly optimal.

\begin{restatable}[Row-Bounded-Difference Row-Monotone Min-Plus Product]{theorem}{thmminplusbd} \label{thm:min-plus-bd}
The min-plus product of a row-bounded-difference, row-monotone matrix $A \in \Int^{n \times n}$ and a column-bounded-difference, column-monotone matrix $B \in \Int^{n \times n}$ cannot be computed in time $O(n^{2.5-\epsilon})$ (for any constant $\epsilon > 0$), unless the APSP Hypothesis fails.
\end{restatable}

\paragraph{And More \texorpdfstring{\dots}{...}}
In \cref{sec:intermediate} we give some more matching APSP-based lower bounds for intermediate-complexity problems such as \emph{Min Product} and \emph{Min-Equality Product}, and recap how via known reductions these further imply lower bounds for even more problems such as \emph{All-Edges Monochromatic Triangle}~\cite{VassilevskaX20b}. Finally, we give a \emph{non-matching} $n^{7/3-o(1)}$ lower bound for the important \emph{Min-Witness Product} problem.

\subsection{Proof Ideas}
On a technical level, our approach differs from most other fine-grained reductions. Typically, to design a reduction, we study the problem to reduce \emph{to}, aiming to understand the hard instances to reduce to. Here, we instead study the problem to reduce \emph{from}, namely APSP or Min-Plus Product, and give subcubic algorithms for a large class of input graphs (or matrices) satisfying a certain \emph{structural property}. Then, to design APSP-based hardness reductions, we can first apply this subcubic-time algorithm to handle the structured part of the graph (or matrix) and reduce to its unstructured core, which is often easier to reduce from. This perspective can be seen as a recent trend with comparably few successful applications~\cite{AbboudBKZ22,AbboudBF23,ChanVX23,JinX23,ChanX24}, including the aforementioned work introducing the ``Fredman's trick meets dominance product'' technique~\cite{ChanVX23}.

A significant challenge is to identify the \emph{right} structural property. We propose a new, expressive one: a rank measure for matrices that we call the \emph{select-plus rank}.

\subsubsection{Select-Plus Rank}

\begin{samepage}
\begin{restatable}[Select-Plus Rank]{definition}{defrank} \label{def:rank}
The \emph{select-plus rank} of a matrix $A \in (\Int \cup \set{\bot})^{n \times m}$, denoted by $r(A)$, is the smallest integer $r \geq 0$ such that there are integer matrices $U \in \Int^{n \times r}$ and $V \in \Int^{r \times m}$ satisfying for all $(i, j) \in [n] \times [m]$:
\begin{equation*}
    A[i, j] \in \set{U[i, 1] + V[1, j], \dots, U[i, r] + V[r, j]} \cup \set{\bot}.
\end{equation*}
\end{restatable}
\end{samepage}

This definition is reminiscent of standard matrix rank, which can be characterized as the smallest $r$ such that $A = U V$ for some $U \in \Int^{n \times r}$ and $V \in \Int^{r \times m}$, i.e., such that each entry can be written as:
\begin{equation*}
    A[i, j] = U[i, 1] \cdot V[1, j] + \dots + U[i, r] \cdot V[r, j].
\end{equation*}
Hence, the select-plus rank can be seen as the modification where we replace multiplication by addition, and addition by ``selection''.\footnote{Another reasonable notion could be the \emph{min-plus rank} of a matrix, defined analogously with the condition that $A[i, j] = \min\set{U[i, 1] + V[1, j], \dots, U[i, r] + V[r, j]}$. This is a strictly more restrictive notion, i.e., the select-plus rank is always upper-bounded by the min-plus rank. Therefore, all of our algorithmic results immediately also apply to the min-plus rank. However, for some applications (specifically, \cref{lem:sm-univ-reduct}) we crucially exploit the extra freedom that the select-plus rank offers.} In many aspects it also behaves like the standard matrix rank, e.g., the select-plus rank of an $n \times n$ matrix ranges from $0$ to $n$.

Our key technical result is the following algorithm.

\begin{theorem}[Low-Rank Min-Plus Product] \label{thm:min-plus-low-rank}
The min-plus product of two given matrices $A, B$ can be computed in deterministic time $n^{3+o(1)} \cdot (r / n^{3-\omega})^{\Omega(1)}$, provided that we have access to a select-plus rank-$r$ decomposition of $A$ or $B$ (or of a matrix $C$ that approximates $A * B$ with entry-wise additive error $\pm O(1)$).
\end{theorem}

The point is that when the select-plus rank is at most $r \leq n^{3-\omega-\epsilon}$ the running time is truly subcubic, $n^{3-\Omega(\epsilon)}$. In particular, if $\omega = 2$ the running time is subcubic for ranks up to $r \leq n^{1-\epsilon}$---i.e., unless all matrices have almost full rank, we obtain a subcubic-time algorithm for Min-Plus Product!

\cref{thm:min-plus-low-rank} implies that low-rank APSP, i.e., APSP on graphs whose adjacency lists have select-plus rank at most $n^{3-\omega-\epsilon}$, can be solved in subcubic time. This result unifies and generalizes several previously studied classes of graphs known to admit subcubic-time algorithms, including sparse graphs, graphs with small weights in~$[n^{3-\omega-\epsilon}]$~\cite{Seidel95,GalilM97,AlonGM97,ShoshanZ99,Zwick02}, node-weighted graphs~\cite{Chan10,Yuster09,AbboudFJVX25}, and even graphs with at most~$n^{3-\omega-\epsilon}$ distinct edge weights per node~\cite{Yuster09,AbboudFJVX25}.

In the following, we will first explore how the claimed universe reductions and fine-grained lower bounds can be derived from \cref{thm:min-plus-low-rank} (\cref{sec:intro:sec:overview:sec:apps}), and then describe on a high level how \cref{thm:min-plus-low-rank} can be proved (\cref{sec:intro:sec:overview:sec:min-plus-low-rank}).

\subsubsection{Low-Rank Min-Plus Product in Action} \label{sec:intro:sec:overview:sec:apps}
The main power of our new low-rank technique is that it allows us to significantly restrict the set of ``hard'' Min-Plus Product instances. Such restrictions are already known, though mostly simple ones. For example, call any index $k$ attaining the minimum $\min_k (A[i, k] + B[k, j])$ a \emph{witness} of $(A * B)[i, j]$. In a hard instance~$A * B$ we expect that most entries have only very few, say at most $n^\epsilon$, witnesses. Otherwise we could narrow our attention to $\tilde O(n^{1-\epsilon})$ randomly sampled indices $k$---with high probability these include at least one witness, and the resulting smaller product can be computed in subcubic time $\tilde O(n^{3-\epsilon})$. The immediate consequence is that in all lower bound constructions based on Min-Plus Product we can now assume the given matrices~$A, B$ to fulfill an extra requirement (namely that most entries in $A * B$ have few witnesses). Of course, this particular insight about the number of witnesses has limited impact.

Instead, a quantity which we really care about (for reasons that we will outline soon) is the number of \emph{$q$\=/pseudo-witnesses}, i.e., the number of indices $k$ attaining the minimum \smash{$\min_k (\floor{\frac{A[i, k]}{q}} + \floor{\frac{B[k, j]}{q}})$}. In words, $k$ is a $q$-pseudo-witness if it becomes a witness after rounding all entries to multiples of $q$. In the previous paragraph we argued by elementary means that we expect at most $n^\epsilon$ $1$-pseudo-witnesses (a.k.a.\ witnesses) per output entry. With the help of low-rank Min-Plus Product one can show that, more generally for any~$q \geq 1$, there are at most~$q \cdot n^\epsilon$ $q$-pseudo-witnesses per output entry in the hard instances. As we will see this fact has surprisingly strong implications! For instance, the equivalence of the Strong APSP and Directed Unweighted APSP Hypotheses follows almost immediately.

We now describe the main idea behind this fact. Focus on the simplified setting that \emph{all} entries $(A * B)[i, j]$ have more than $q \cdot n^\epsilon$ $q$-pseudo-witnesses. We show that $A * B$ has small select-plus rank, $r(A * B) \leq \tilde O(n^{1-\epsilon})$, and thus \cref{thm:min-plus-low-rank} allows to compute $A * B$ in subcubic time. In particular, $A * B$ cannot be a hard instance. To prove the rank bound we proceed in two steps:
\begin{enumerate}
    \item Consider the matrices \smash{$A' = \floor{\frac{A}{q}}$} and \smash{$B' = \floor{\frac{B}{q}}$}. The first step is to show that \smash{$r(A' * B') \leq \tilde O(n^{1-\epsilon} / q)$}. To see this, take a subset $\mathcal K$ of \smash{$\tilde O(n^{1-\epsilon} / q)$} randomly sampled indices $k$. With high probability each output entry has at least one $q$-pseudo-witness in the sample. So consider the restriction $U$ of $A'$ to the columns in~$\mathcal K$, and the restriction $V$ of $B'$ to the rows in $\mathcal K$. Each entry $(A' * B')[i, j]$ can be expressed as a sum of the form $U[i, k] + V[k, j]$ for some $k \in \mathcal K$, which, by \cref{def:rank}, implies that $r(A' * B') \leq |\mathcal K| = \tilde O(n^{1-\epsilon} / q)$.
    \item Next, observe that we can express $A * B = q \cdot (A' * B') + R$ for some remainder matrix $R$ with entries bounded by $O(q)$. That matrix trivially has select-plus rank $r(R) \leq O(q)$ (\cref{fac:rank-triv-univ}). Since the select-plus rank behaves submultiplicatively (\cref{fac:rank-submult}), it follows that $r(A * B) \leq r(A' * B') \cdot r(R) \leq \tilde O(n^{1-\epsilon} / q \cdot q) = \tilde O(n^{1-\epsilon})$ as claimed.
\end{enumerate}

The description so far demonstrates generically that low-rank Min-Plus Product can be helpful in the design of reductions. We now give some very abstract pointers how our specific applications benefit from this approach.

\paragraph{Application 1: Small-Universe Reduction}
\cref{thm:strong-apsp-implies-dir-unw-apsp} follows rather easily now by picking~\makebox{$q = n^{1/2}$} and adapting the ideas from~\cite{ChanVX23}. Specifically, knowing that there are at most $n^{1/2+\epsilon}$ $n^{1/2}$-pseudo-witnesses per output entry, we can afford to \emph{list} all such pseudo-witnesses. When everything is set up carefully this list also contains all proper witnesses, so we can easily read off the min-plus product~$A * B$. Based on this idea one can show that $\MinPlus(n, n, n \mid u \leq n)$ reduces to $n^{1/2+\epsilon}$ instances of $\MinPlus(n, n^{1/2}, n \mid u \leq n^{1/2})$ (in case that $\omega = 2$) which entails the desired reduction. See \cref{sec:univ-reduct:sec:small} for the details.

\paragraph{Application 2: Doubling Reduction}
The proof of \cref{thm:apsp-unif-low-doubling} is more complicated and requires developing more technical ideas. This result applies to arbitrarily large polynomial universes $u$, so our approach is to ``grow'' $q$ from $1$ to $u$ in small steps. In some step the number of $q$-pseudo-witnesses must jump from less than~$n^\epsilon$ to at least $n^\epsilon$---this is the step where we can effectively apply the previous ideas. On the one hand, the number of pseudo-witnesses is large enough so that the (appropriately rounded) min-plus product matrix has small rank. On the other hand, the number of pseudo-witnesses is small enough so that we can afford to \emph{list} all pseudo-witnesses (and thereby all witnesses). This idea essentially allows us to assume that the initial Min-Plus Product instance is low-rank, so it remains to reduce low-rank Min-Plus Product to uniform Min-Plus Product. We obtain this reduction as a by-product of our low-rank Min-Plus Product algorithm (to be described next in \cref{sec:intro:sec:overview:sec:min-plus-low-rank}). We defer the many technical details to \cref{sec:univ-reduct:sec:doubling}.

\paragraph{Application 3: Conditional Lower Bounds}
Finally, our conditional lower bounds follow from the doubling reduction when appropriately generalized to \emph{rectangular} matrices (see \cref{cor:min-plus-rect-low-doubling}). Besides applying this big hammer, most of the individual reductions require only simple or previously known ideas. One exception is \cref{thm:min-plus-bd}. See \cref{sec:intermediate} for details and \cref{fig:reducts} for the resulting web of reductions. 

\begin{figure}[t]
\begin{tikzpicture}[
    prob base/.style={
        draw=black,
        fill=PaleBlue,
        inner sep=0.1cm,
        minimum width=3.1cm,
        minimum height=0.95cm,
        font=\fontsize{9}{11}\selectfont,
        align=center
    },
    prob/.style={
        prob base,
        rounded corners,
    },
    large/.style={minimum width=4.46cm},
    important/.style={
        font=\fontsize{9}{11}\selectfont\bfseries,
        fill=LightBlue,
    },
    reduct/.style={
        draw,
        >=latex,
        shorten <=.1cm,
        shorten >=.1cm,
        shift up/.style={transform canvas={yshift=.12cm}},
        shift down/.style={transform canvas={yshift=-.12cm}},
        shift left/.style={transform canvas={xshift=-.12cm}},
        shift right/.style={transform canvas={xshift=.12cm}},
    },
    every edge quotes/.append style={font=\fontsize{9}{10}\selectfont},
]

\matrix[
    every node/.style={anchor=north},
    column sep=1.7cm,
    row sep=1cm,
    row 2/.style={row sep=2.1cm},
] at (0, 0) {
    \node[prob, important] (apsp) {APSP}; &
    &
    \node[prob, important] (strong-apsp) {Strong APSP}; \\
    \node[prob] (min-plus-sq) {$\MinPlus(n, n, n)$}; &
    \node[prob] (min-plus-sq-unif) {$n$-Uniform\\$\MinPlus(n, n, n)$}; &
    \node[prob] (min-plus-sq-univ) {$n$-Universe\\$\MinPlus(n, n, n)$}; \\
    \node[prob] (min-plus-sqrt) {$\MinPlus(n, \sqrt{n}, n)$}; &
    \node[prob] (min-plus-sqrt-unif) {$\sqrt{n}$-Uniform\\$\MinPlus(n, \sqrt{n}, n)$}; &
    \node[prob] (min-plus-sqrt-univ) {$\sqrt{n}$-Universe\\$\MinPlus(n, \sqrt{n}, n)$}; \\
    &
    \node[prob] (min-max) {Min-Max Product}; &
    \node[prob, important] (dir-unw-apsp) {Dir.\ Unw.\ APSP}; \\
};

\node[prob base, cloud, cloud puffs=14, aspect=2.5, inner sep=-.55cm, below=.8cm of min-plus-sqrt] (node-wgt-apsp) {Node-Weighted APSP,\\Row-Monotone Min-Plus Product,\\Min-Equality Product,\\\dots\\[-.7cm]};
\node[prob base, cloud, cloud puffs=18, aspect=2.3, inner sep=-.5cm, below=1cm of min-max] (apbp) {All-Pairs Bottleneck Paths,\\All-Pairs Nondecreasing Paths,\\Approximate APSP (without Scaling),\\\dots};

\path[reduct, <->] (apsp) to (min-plus-sq);
\path[reduct, <->] (strong-apsp) to["\cite{ShoshanZ99}"] (min-plus-sq-univ);
\path[reduct, <->] (dir-unw-apsp) to["\cite{ChanVX21}", swap] (min-plus-sqrt-univ);
\path[reduct, ->, shift up] (min-plus-sq) to["Lem.~\ref{lem:doubling-reduct}" above=.1cm] (min-plus-sq-unif);
\path[reduct, ->, shift down] (min-plus-sq-unif) to (min-plus-sq);
\path[reduct, ->, shift up, dashed] (min-plus-sq-unif) to["Lem.~\ref{lem:min-plus-sum-order-preserving}" above=.1cm] (min-plus-sq-univ);
\path[reduct, ->, shift down] (min-plus-sq-univ) to (min-plus-sq-unif);
\path[reduct, ->, shift up] (min-plus-sqrt) to["Lem.~\ref{lem:doubling-reduct}" above=.1cm] (min-plus-sqrt-unif);
\path[reduct, ->, shift down] (min-plus-sqrt-unif) to (min-plus-sqrt);
\path[reduct, ->, shift up, dashed] (min-plus-sqrt-unif) to["Lem.~\ref{lem:min-plus-sum-order-preserving}" above=.1cm] (min-plus-sqrt-univ);
\path[reduct, ->, shift down] (min-plus-sqrt-univ) to (min-plus-sqrt-unif);
\path[reduct, <->] (min-plus-sq) to (min-plus-sqrt);
\path[reduct, ->] (min-plus-sqrt-unif) to (min-plus-sq-unif);
\path[reduct, ->, shift left] (min-plus-sq-univ) to["Cor.~\ref{cor:min-plus-rect-sm-univ}", swap] (min-plus-sqrt-univ);
\path[reduct, ->, shift right] (min-plus-sqrt-univ) to (min-plus-sq-univ);
\path[reduct, ->] (min-plus-sqrt-unif) to["Lem.~\ref{lem:min-max-prod}"] (min-max);
\path[reduct, ->] (min-plus-sqrt-unif.south west) to["Sec.~\ref{sec:intermediate}", swap, pos=.73] (node-wgt-apsp);
\path[reduct, ->] (dir-unw-apsp) to["\cite{LincolnPV20}", swap] (min-max);
\path[reduct, ->] (min-max) to["\cite{VassilevskaWY07,Vassilevska10,BringmannKW19}"] (apbp);

\begin{pgfonlayer}{bg}
    \path (-.5\textwidth+0.4pt, 0) coordinate (pic-left);
    \path (.5\textwidth-0.4pt, 0) coordinate (pic-right);
    \path (apsp) + (0, 1cm) coordinate (n3-top);
    \path (min-plus-sq-univ) + (0, -1cm) coordinate (n3-bottom);
    \path (min-plus-sqrt) + (0, 1cm) coordinate (n25-top);
    \path (apbp) + (0, -2.2cm) coordinate (n25-bottom);
    \path[draw, dashed, rounded corners] (pic-left |- n3-top)
        node[below right=.3cm] {$n^{3\pm o(1)}$}
        rectangle (pic-right |- n3-bottom);
    \path[draw, dashed, rounded corners] (pic-left |- n25-top)
        node[below right=.3cm] {$n^{2.5\pm o(1)}$}
        rectangle (pic-right |- n25-bottom);
\end{pgfonlayer}
\end{tikzpicture}

\caption{Illustrates our fine-grained reductions and resulting lower bounds assuming that $\omega = 2$. Each arrow symbolizes a (tight) fine-grained reduction. Unlabeled arrows correspond to trivial reductions. The dashed arrows are conditioned on the additive combinatorics \cref{hypo:hashing}.} \label{fig:reducts}
\end{figure}
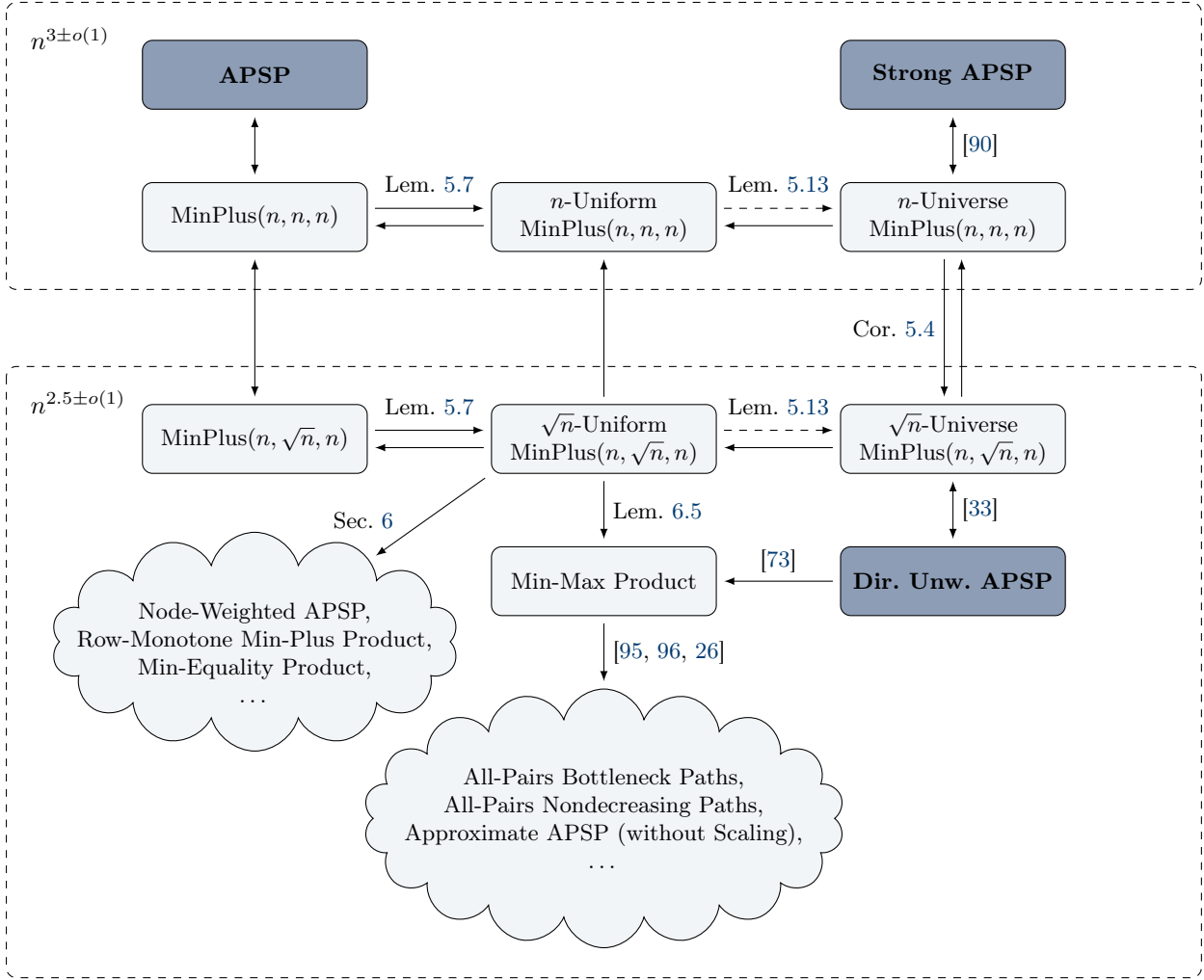

\paragraph{Even More Applications?}
Overall, we are confident that our new low-rank APSP algorithm will find more applications in the future---perhaps in the design of more fine-grained reductions, or, even more ambitiously, perhaps even as a step towards challenging the APSP Hypothesis.

\subsubsection{Subcubic Algorithm for Low-Rank Min-Plus Product} \label{sec:intro:sec:overview:sec:min-plus-low-rank}
The proof of \cref{thm:min-plus-low-rank} is extremely involved. It builds on the recent algorithm for APSP with few distinct edge weights per node due to Abboud, Fischer, Jin, Vassilevska W., and Xi~\cite{AbboudFJVX25}, which itself is already a complicated chain of reductions, and adds more layers to the chain. Even worse: To obtain the most general form of \cref{thm:min-plus-low-rank} (which is necessary e.g.\ for the doubling reduction) we are forced to open up their algorithm, recast it in a restricted framework of reductions (so-called \emph{potential-adjusting} reductions to be defined in \cref{sec:prelims:sec:pot}), and generalize all steps to rectangular matrices. Fortunately, along the way we manage to simplify one of the highly technical steps in~\cite{AbboudFJVX25} (again leveraging the expressiveness of the novel select-plus rank), and in doing so, we even obtain a slightly faster APSP algorithm for the few-weights case.\footnote{To be precise, the previously fastest APSP algorithm with at most $n^{3-\omega-\epsilon}$ distinct edge weights per node~\cite{AbboudFJVX25} runs in subcubic time \smash{$O(n^{3-\Omega_\epsilon(1)})$} where the hidden dependence on $\epsilon$ is polynomial. This unusual dependence is due to a complicated recursion. In our newer version we obtain a more streamlined recursion leading to a linear dependence $\Omega(\epsilon)$ as in \cref{thm:min-plus-low-rank}.}

As in~\cite{AbboudFJVX25}, our subcubic-time algorithm also works for the strictly harder low-rank \emph{Exact Triangle} problem. In this problem we are additionally given a matrix $C$, and the goal is to report all entries in the min-plus product $A * B$ that coincide with $C$. We now describe the key insights behind this algorithm for low-rank Exact Triangle.

\paragraph{Main Idea: Reduction to Slice-Uniform Instances}
Let $(A, B, C)$ be a given Exact Triangle instance with select-plus rank $r(C) \leq r$. Let $C_\ell$ denote the submatrix of $C$ consisting of all entries that are captured by the $\ell$-th part of the rank decomposition (i.e., $C[i, j] = U[i, \ell] + V[\ell, j]$), where all other entries are replaced by some dummy symbol $\bot$. Our strategy is to iteratively solve the subinstances $(A, B, C_1), \dots, (A, B, C_r)$ one-by-one. As we aim for total time $O(n^{3-\epsilon})$, each such subinstance must be solved in time $O(n^{3-\epsilon} / r)$. Recall that the most interesting case is when $r$ is very close to $n$, so this time budget is just barely superquadratic. Focus on an instance $(A, B, C_\ell)$, and define matrices $A', B'$ by
\begin{align*}
    A'[i, k] &= A[i, k] - U[i, \ell], \\
    B'[k, j] &= -(B[k, j] - V[\ell, j]).
\end{align*}
Then any exact triangle $A[i, k] + B[k, j] = C_\ell[i, j]$ corresponds to a solution of the much simpler equality condition $A'[i, k] = B'[k, j]$. This remaining problem is typically called \emph{Equality Product} (where in our case the set of relevant outputs is only the non-$\bot$ entries in $C_\ell$). Up to this point the algorithm can be seen as the ``Fredman's trick'' part of the ``Fredman's trick meets dominance product'' technique~\cite{ChanVX23}. Next, this technique would solve the Equality Product instances each in time $\tilde O(n^{(3-\omega)/2})$ by Matoušek's algorithm~\cite{Matousek91} applied as a black-box. This running time is too slow for our purposes, so from now on our algorithm differs.

Matoušek's algorithm can be viewed as a simple heavy/light trick, so it is natural to proceed similarly. Call an entry in $A'$ \emph{light} if in its respective column in the same integer entry appears at most $n^{1-\epsilon} / r$ times, and \emph{heavy} otherwise. This threshold is chosen in such a way that a simple brute-force algorithm can deal with instances containing \emph{no} heavy entries in our time budget, $O(n^{3-\epsilon} / r)$.

More interestingly, consider the other extreme where \emph{all} entries in $A'$ are heavy. If we ever reach this situation then we will ``give up'' on the iterative algorithm, go back to the original instance $(A, B, C)$, and solve it in one shot. How? Let~$C'$ be the matrix defined by
\begin{equation*}
    C'[i, j] = C[i, j] - U[i, \ell].
\end{equation*}
This is chosen in a way that the instances $(A, B, C)$ and $(A', B, C')$ are exactly equivalent (as the term~$U[i, \ell]$ cancels in $A'[i, k] + B[k, j] = C'[i, j]$). It thus suffices to solve the augmented instance $(A', B, C')$. The benefit is that the augmented instance satisfies the extra constraint that each column in $A'$ contains a truly sublinear number $d \leq r \cdot n^\epsilon$ of distinct entries (as each heavy entry repeats at least $n^{1-\epsilon} / r$ times per column in $A'$). This is exactly the \emph{$d$-slice-uniform} special case known to be in subcubic time due to Abboud, Fischer, Jin, Vassilevska W., and Xi~\cite{AbboudFJVX25}.

\paragraph{Technical Challenge: Regular Rank Decompositions}
Unfortunately, the previous paragraph just treats the unrealistic situations that either \emph{no} or \emph{all} entries are heavy---but what if only \emph{some} entries are heavy? This remaining case poses some serious technical challenges. Our solution involves two subcases. The good subcase is when the number of heavy entries is still quite large, say at least $n^{2-\delta}$ for some sufficiently small~$\delta > 0$. In this case a simple modification of the previous approach works: after solving an instance~$(A', B, C')$ simply continue with the iterative algorithm, but \emph{remove} the at least $n^{2-\delta}$ heavy entries from $A$ in all future iterations. We thereby solve at most $n^\delta$ instances $(A', B, C')$ in total.

The bad subcase is when the number of heavy entries is less than $n^{2-\delta}$. Here, a natural idea is to exploit the sparsity of $C_\ell$. Specifically, we could enumerate all heavy entries $A'[i, k]$ and further all~\makebox{$C_\ell[i, j] \neq \bot$}. For each pair we can test in constant time whether $A'[i, k] = B'[k, j]$. Heuristically speaking, if the entries of $C$ were uniformly distributed among the matrices $C_1, \dots, C_r$, then we would expect only $O(n/r)$ entries~\makebox{$C_\ell[i, j] \neq \bot$}, leading to a total time of $O(n^{3-\delta} / r)$ as intended. The challenge lies in making this heuristic argument rigorous. More precisely, we would like to assume that the given rank decomposition is \emph{row-regular}, i.e., that the matrices $C_1, \dots, C_r$ each have at most $\tilde O(n / r)$ non-$\bot$ entries per row. Is this possible without loss of generality?

It is illuminating to consider the following example. Suppose that $C$ contains a large rank-$1$ submatrix, of size $s \times s$ for some $s \gg n / r$, and suppose that this submatrix is fully contained in some part $C_\ell$. Then that part violates the row-regularity condition. Even worse: There could be up to $n^2 / s^2$ such low-rank submatrices, leading to several violations. Luckily, such submatrices essentially turn out to be the only obstacles. We show that, perhaps surprisingly, we can construct an alternative decomposition for (the union of) these obstructions with rank $\tilde O(n / s) \ll r$. That is, conceptually, the only way to violate the regularity condition is when some parts of $C$ have \emph{strictly smaller} rank.

Formally, we show that an arbitrary rank-$r$ matrix $C$ can be decomposed into three parts $C_\row$, $C_\col$, and~$C_\sm$ such that $C_\row$ has a row-regular rank-$r$ decomposition, $C_\col$ has an analogously defined column-regular rank-$r$ decomposition, and $C_\sm$ has rank at most $r/2$ (see \cref{sec:rank:sec:decomp-reg}). Then $C_\row$ is as planned for our low-rank Min-Plus Product algorithm, $C_\col$ can be dealt with by symmetry, and $C_\sm$ can be dealt with by \emph{recursion}.

\subsection{Real Weights}
A remarkable consequence of our results is that we conditionally resolve the complexity of All-Pairs Shortest Paths for \emph{all} variants: undirected or directed, with edge or node weights, ranging from unweighted instances over small weights $\set{0, \dots, n}$ up to polynomially large integer weights. Out of these variants, only undirected unweighted APSP is in time $\tilde O(n^\omega)$ by Seidel's algorithm~\cite{Seidel95}. For all other variants---i.e., directed graphs, or undirected graphs with edge or node weights in $\set{0, \dots, n^\delta}$ for some constant $\delta > 0$---the respective fastest-known algorithms with super-quadratic complexity turn out to be optimal (in the sense that beating any such algorithm implies a subcubic-time algorithm for general APSP, at least if $\omega = 2$, and conditioned on \cref{hypo:apsp}). See \cref{sec:intermediate:sec:node-wgt-apsp} for the full treatment of node-weighted graphs.

A remaining interesting question is to further connect the APSP problem on integer-weighted graphs with its more general counterpart on \emph{real-weighted} graphs. This problem, and the implications of the weaker \emph{Real APSP Hypothesis} (along with other real-valued hypotheses) have been explored in~\cite{ChanVX22}. A particularly ambitious open question is whether there is a ``universe reduction'' from the (uncountably infinite) real universe to the integers. While parts of our results apply to real-weighted instances as well, other parts inherently rely on the finite bit representation of integers (such as the scaling trick to reduce from Min-Plus Product to Exact Triangle), and so we leave further progress in this direction for future research.
\section{Preliminaries} \label{sec:prelims}
We write $[n] = \set{1, \dots, n}$. Throughout we typically consider matrices $A, B, C \in (\Int \cup \set{\bot})^{n \times m}$ that consist of integer entries plus some designated symbol $\bot$ representing a missing entry (e.g., when taking $A$ to be the adjacency matrix of a weighted graph, $\bot$ represents a missing edge).\footnote{In some contexts it would be natural to write $\infty$, but to be consistent overall we decided to write $\bot$.} We follow the conventions~\makebox{$x + \bot = \bot + x = \bot$} and $\min\set{x, \bot} = x$ for all $x \in \Int \cup \set{\bot}$, and $\min \emptyset = \bot$. We say that $A$ is \emph{partitioned} into two matrices $A_1, A_2$ of the same size, denoted by $A = A_1 \sqcup A_2$, if each entry $A[i, j] \neq \bot$ appears in \emph{exactly one} of the two matrices (i.e., $A_1[i, j] = A[i, j]$ and $A_2[i, j] = \bot$, or $A_1[i, j] = \bot$ and~$A_2[i, j] = A[i, j]$).

Let $\MM(n_1, n_2, n_3)$ denote the time complexity of multiplying an $n_1 \times n_2$ with an $n_2 \times n_3$ integer matrix. Let $2 \leq \omega < 2.372$~\cite{AlmanDVXXZ25} denote the exponent of square matrix multiplication.

\subsection{APSP, Min-Plus Product and Exact Triangle} \label{sec:prelims:sec:probs}
\begin{definition}[APSP] \label{def:apsp}
The \emph{All-Pairs Shortest Paths (APSP)} problem is to compute the $u$-$v$-distances for all pairs of nodes $u, v$ in a given directed edge-weighted $n$-vertex graph.
\end{definition}

\begin{definition}[Min-Plus Product] \label{def:min-plus}
The \emph{min-plus product} of $A \in (\Int \cup \set{\bot})^{n_1 \times n_2}, B \in (\Int \cup \set{\bot})^{n_2 \times n_3}$ is the matrix $A * B \in (\Int \cup \set{\bot})^{n_1 \times n_3}$ defined by
\begin{equation*}
    (A * B)[i, j] = \min_{k \in [n_2]} \parens*{A[i, k] + B[k, j]}.
\end{equation*}
The \emph{Min-Plus Product} problem is to compute the min-plus product $A * B$ of two given matrices $A, B$.
\end{definition}

\begin{definition}[Exact Triangle] \label{def:exact-tri}
Let $A \in (\Int \cup \set{\bot})^{n_1 \times n_2}, B \in (\Int \cup \set{\bot})^{n_2 \times n_3}$ and $C \in (\Int \cup \set{\bot})^{n_1 \times n_3}$. A \emph{triangle} is a triple $(i, k, j) \in [n_1] \times [n_2] \times [n_3]$ such that $A[i, k] \neq \bot, B[k, j] \neq \bot, C[i, j] \neq \bot$, and an \emph{exact triangle} is a triangle $(i, k, j)$ such that $A[i, k] + B[k, j] = C[i, j]$.\footnote{Of course, one could equivalently define an exact triangle to satisfy that $A[i, k] + B[k, j] + C[i, j] = 0$; in this case the problem is typically called the \emph{Zero Triangle} problem. We stick to the version in \cref{def:exact-tri} which is nicer to relate to the Min-Plus Product problem.} The \emph{(All-Edges) Exact Triangle}\footnote{The established name in the fine-grained literature is ``All-Edges'' Exact Triangle. However, in this paper we are exclusively concerned with the All-Edges version and so we will usually drop the ``All-Edges'' prefix for the sake of brevity.} problem is, given $A, B, C$, to decide for each edge $(i, k) \in [n_1] \times [n_2]$ and $(k, j) \in [n_2] \times [n_3]$ and $(i, j) \in [n_1] \times [n_3]$ if it is involved in an exact triangle $(i, k, j)$.
\end{definition}

It is long-known that the APSP and Min-Plus Product problems have the same complexity up to logarithmic factors. Vassilevska~W.\ and Williams~\cite{VassilevskaW18} further proved that Min-Plus Product reduces to All-Edges Exact Triangle by a scaling trick that only incurs an overhead logarithmic in the universe size. 

\subsection{Parameters} \label{sec:prelims:sec:params}
We study the Min-Plus Product and Exact Triangle problems constrained by various parameters, often by multiple parameters at the same time. To simplify notation, we propose a systematic naming scheme. For any Min-Plus Product instance $(A, B)$ or any Exact Triangle instance~$(A, B, C)$, consider the following six parameters:

\begin{parameter}[$r$-Rank]
The select-plus rank of $A$ or $B$ (or $C$) is at most $r$. We assume that as part of the input we additionally receive a rank-$r$ decomposition of that respective matrix.
\end{parameter}

\begin{parameter}[$d$-Slice-Uniform]
The rows or columns in $A$ or $B$ (or $C$) have at most $d$ distinct non-$\bot$ entries each.\footnote{In~\cite{AbboudFJVX25} this constraint was instead called \emph{$d$-weights}.}
\end{parameter}

\begin{parameter}[$D$-Uniform]
The total number of distinct non-$\bot$ entries in all matrices $A$ and $B$ (and $C$) is at most $D$.
\end{parameter}

\begin{parameter}[$\rho$-Regular]
In every row and column of $A$ and $B$ (and $C$), no non-$\bot$ entry appears in more than a $\rho$-fraction of the respective row or column.
\end{parameter}

\begin{parameter}[$K$-Doubling]
Let $X$ be the set of integer entries in $A$ and $B$ (and $C$). Then $X$ has doubling $|X + X| \leq K |X|$.
\end{parameter}

\begin{parameter}[$u$-Universe]
All non-$\bot$ entries in the matrices $A$ and $B$ (and $C$) are from $\set{0, \dots, u}$.
\end{parameter}

For instance, the \emph{$r$-rank Exact Triangle} problem is to solve an Exact Triangle instance $(A, B, C)$ for which we additionally receive access to a select-plus rank-$r$ decomposition of $A$, $B$ or $C$. We also introduce the following notation. We write
\begin{equation*}
    \MinPlus(n_1, n_2, n_3 \mid \mathit{constraint}; \dots; \mathit{constraint})
\end{equation*}
and
\begin{equation*}
    \ExactTri(n_1, n_2, n_3 \mid \mathit{constraint}; \dots; \mathit{constraint})
\end{equation*}
to denote the best-possible running times of the Min-Plus Product and Exact Triangle problems, respectively, restricted to instances that satisfy the list of constraints. Each $\mathit{constraint}$ takes two possible forms. The first is $\mathit{parameter} \leq \mathit{value}$, for any of the parameters defined above, to indicate that this parameter is bounded by $\mathit{value}$. E.g., $\MinPlus(n, n, n \mid d \leq n^{0.9})$ refers to the running time of the $n^{0.9}$-slice-uniform Min-Plus Product problem (where in one of the two matrices the rows or columns contain at most $n^{0.9}$ distinct entries each). The second form is a free parameter. E.g., the fact that Min-Plus Product with entries in~$\set{0, \dots, u}$ can be solved in time $\tilde O(n^\omega u)$~\cite{AlonGM97} can be expressed concisely as $\MinPlus(n, n, n \mid u) = \tilde O(n^\omega u)$. As a more complicated example consider
\begin{equation*}
    \ExactTri(n, n, n \mid D \leq n^{0.9}; \rho \leq 1/D; K).
\end{equation*}
This expresses the running time of the Exact Triangle problem, restricted to matrices $A, B, C$ with at most~$D$ distinct non-$\bot$ entries such that $D \leq n^{0.9}$, where each entry in $A, B$ and $C$ appears in at most a $1/D$-fraction of its row and column, and where the set of integer entries in the three matrices has doubling at most $K$. (We will indeed encounter this problem in \cref{lem:exact-tri-unif-reg-to-low-doubling}.)

\begin{remark}
Note that many of these parameters are trivially related to each other. For instance, each $u$-universe instance is trivially $u$-uniform, and each $D$-uniform instance is trivially $D$-slice-uniform. It is also not hard to see that each $d$-slice-uniform instance is trivially $d$-rank. In \cref{sec:exact-tri-low-rank} we will essentially show the opposites of some of these trivial directions, thereby proving that many of these parameters are \emph{equivalent} in terms of subcubic-time algorithms.
\end{remark}

\begin{remark}
We will usually assume that all entries are polynomially bounded, so we always implicitly add the constraint $u \leq (n_1 n_2 n_3)^{O(1)}$. Having said that, all reductions here also apply to larger universes~$u$ with an overhead of \smash{$(\log u)^{O(1)}$}.\footnote{There is a technical detail: Somewhere along our chain of reductions, in \cref{lem:pop-sum-decomp}, we rely on a deterministic algorithm with a running time overhead of \smash{$2^{\tilde O(\sqrt{\log u})} = u^{o(1)}$}. However, this deterministic algorithm can be replaced by a randomized algorithm with only polylogarithmic overhead.}
\end{remark}

\begin{remark}
It is straightforward to verify that Vassilevska~W.\ and Williams'~\cite{VassilevskaW18} reduction from Min-Plus Product to Exact Triangle, $\MinPlus(n_1, n_2, n_3 \mid u) = O(\ExactTri(n_1, n_2, n_3 \mid u) \cdot \log u)$, also preserves the parameters $r$, $d$, and $D$.
\end{remark}

\subsection{Potential Adjustments} \label{sec:prelims:sec:pot}
It turns out to be useful to define a particularly constrained type of fine-grained reductions. We often want to transform an Exact Triangle instance $(A, B, C)$ into one or more instances that are equivalent in the following sense: Each exact triangle is preserved, and for all other triangles $(i, k, j)$ we do not change the value $A[i, k] + B[k, j] - C[i, j]$. Informally, this requirement is exactly what allows us to simulate reductions for Min-Plus Product (see \cref{lem:doubling-reduct}). The following definitions make these requirements formal:

\begin{definition}[Potential Adjustment]
Let $(A, B, C), (A', B', C')$ be Exact Triangle instances of the same size $n_1 \times n_2 \times n_3$. We say that $(A', B', C')$ is a \emph{potential adjustment} of $(A, B, C)$ if there are \emph{potential functions} $u \in \Int^{n_1}, v \in \Int^{n_2}, w \in \Int^{n_3}$ so that
\begin{equation*}
    \makeatletter\expandafter\def\expandafter\@arraycr\expandafter{\@arraycr\noalign{\vskip\jot}}\makeatother%
    \begin{array}{@{}r@{\;}c@{\;}c@{\:}c@{\:}c@{\:}c@{\:}c@{\;\;\;}c@{\;\;\;}r@{\;}c@{\;}c@{\;\;\;}c@{\ }c@{\;}c@{\;}c@{\:}c@{\:}c@{}}
        A'[i, k] & = & A[i, k] & + & u[i] & + & v[k] & \text{or} & A'[i, k] & = & \bot & \text{for all} & (i, k) & \in & [n_1] & \times & [n_2], \\
        B'[k, j] & = & B[k, j] & - & v[k] & + & w[j] & \text{or} & B'[k, j] & = & \bot & \text{for all} & (k, j) & \in & [n_2] & \times & [n_3], \\
        C'[i, j] & = & C[i, j] & + & u[i] & + & w[j] & \text{or} & C'[i, j] & = & \bot & \text{for all} & (i, j) & \in & [n_1] & \times & [n_3].
    \end{array}
\end{equation*}
\end{definition}

That is, in a potential adjustment, one is allowed to ``delete edges'' and to add ``node weights'' in such a way that $A[i, k] + B[k, j] - C[i, j]$ remains unchanged for all surviving triangles. This use of potential adjustments is akin to that in several shortest-path problems, e.g., the well-known Johnson trick~\cite{Johnson77}.

\begin{definition}[Potential-Adjusting Reduction] \label{def:pot-reduct}
A \emph{potential-adjusting $1$-to-$N$ reduction in time $T$} is an algorithm that runs in time $T$, takes as input a (possibly constrained) Exact Triangle instance $(A, B, C)$ of size $n_1 \times n_2 \times n_3$, and returns as output a pair $(\mathcal I, \mathcal T)$ where
\begin{itemize}
    \item $\mathcal I$ is a set of $N$ (possibly constrained) potential adjustments of $(A, B, C)$, and
    \item $\mathcal T \subseteq [n_1] \times [n_2] \times [n_3]$ with size $|\mathcal T| \leq T$,
\end{itemize}
so that, for each exact triangle $(i, k, j)$ in $(A, B, C)$, $(i, k, j)$ appears as an (exact) triangle in an instance in~$\mathcal I$, or $(i, k, j) \in \mathcal T$.
\end{definition}

\begin{remark}
An even stricter notion would be to enforce that each exact triangle appears \emph{either} in $\mathcal T$ or in \emph{exactly one} instance in $\mathcal I$. This would directly imply reductions between counting versions of the problems as well. It can be checked that all of our reductions can actually achieve this stricter constraint (with some minor modifications here and there)---however, since the reductions are already quite involved we have decided to stick to the simpler \cref{def:pot-reduct}.
\end{remark}

\subsection{Additive Combinatorics} \label{sec:prelims:sec:add-comb}
Let $X, Y \subseteq \Int$ be two finite sets. We write $X + Y = \set{x + y : x \in X, y \in Y}$ to denote their \emph{sumset}. We also set $X - Y = \set{x - y : x \in X, y \in Y}$ and $-X = \set{-x : x \in X}$, and for an integer $x$ we write $X + x = X + \set{x}$. Let $r_{X+Y}(z)$ denote the \emph{multiplicity} of $z$ in the sumset $X + Y$, i.e., the number of representations $z = x + y$ with $x \in X, y \in Y$, or equivalently, $r_{X+Y}(z) = |X \cap (z - Y)|$. We occasionally rely on the following two basic facts from additive combinatorics; see e.g.\ the textbooks~\cite{TaoV06,Zhao23}.

\begin{lemma}[Plünnecke-Ruzsa Inequality~\cite{Plünnecke70,Ruzsa89}] \label{lem:pluennecke-ruzsa}
Let $X, Y$ be finite sets with $|X + Y| \leq K |X|$. Then, letting~$n Y$ denote the $n$-fold sumset $Y + \dots + Y$, it holds for all nonnegative integers $n, m$ that:
\begin{equation*}
    |n Y - m Y| \leq K^{n+m} |X|.
\end{equation*}
\end{lemma}

\begin{lemma}[Ruzsa's Triangle Inequality~\cite{Ruzsa96}] \label{lem:ruzsa-triangle}
For all finite sets $X, Y, Z$:
\begin{equation*}
    |X + Y| \leq \frac{|X + Z| \, |Y + Z|}{|Z|}.
\end{equation*}
\end{lemma}

We also need the following algorithm to efficiently compute sumsets:

\begin{lemma}[Sparse Convolution~\cite{BringmannFN22}] \label{lem:sparse-conv}
There is an algorithm that, given sets $X, Y \subseteq \set{-u, \dots, u}$, computes their sumset $X + Y$ along with the multiplicities $r_{X+Y}(\cdot)$ in deterministic time $|X + Y| \cdot (\log u)^{O(1)}$.
\end{lemma}

\subsection{Machine Model} \label{sec:prelims:sec:machine}
Throughout we consider the Word RAM model with logarithmic word size $\Theta(\log n)$, where $n$ is the input size. Moreover, for all graph and matrix problems we always only consider weights that are polynomially (i.e., $n^{O(1)}$) bounded and can thus be stored in $O(1)$ machine words. As is the standard in fine-grained complexity, we allow randomized algorithms (that succeed with high probability $1 - n^{-c}$ for an arbitrarily large constant $c$) in hypotheses and reductions.
\section{Select-Plus Rank} \label{sec:rank}
In this section, we recap the central definition of select-plus rank, establish some basic facts (\cref{sec:rank:sec:basics}), and develop a refined notion of rank decompositions that will be crucial later on (\cref{sec:rank:sec:decomp-reg}).

\defrank*

We reiterate that this alternative rank notion is reminiscent of the standard matrix rank, which can be characterized as the smallest number $r$ such that $A$ can be written as $U V$ for matrices $U$ of size $n \times r$ and~$V$ of size~\makebox{$r \times m$}, respectively. In the following we will often refer to the select-plus rank simply as the ``rank'' of a matrix.

\subsection{Basic Facts} \label{sec:rank:sec:basics}
We gather some basic facts concerning the select-plus rank. The first two involve trivial upper bounds on the rank based on the size and maximum entry.

\begin{fact} \label{fac:rank-triv-size}
$r(A) \leq \min\set{n, m}$ for all matrices $A \in (\Int \cup \set{\bot})^{n \times m}$.
\end{fact}
\begin{proof}
Take $U$ to be all-zero and $V = A$, or vice versa.
\end{proof}

\begin{fact} \label{fac:rank-triv-univ}
$r(A) \leq u$ for all matrices $A \in ([u] \cup \set{\bot})^{n \times m}$.
\end{fact}
\begin{proof}
Take $U$ to be all-zero and define $V$ by $V[k, j] = k$.
\end{proof}

\begin{fact}[Submultiplicativity] \label{fac:rank-submult}
$r(A_1 + A_2) \leq r(A_1) \cdot r(A_2)$ for all matrices $A_1, A_2 \in (\Int \cup \set{\bot})^{n \times m}$.
\end{fact}
\begin{proof}
Let $U_1, V_1$ be matrices witnessing $r(A_1)$ and let $U_2, V_2$ be matrices witnessing $r(A_2)$. Consider the following matrices $U, V$, where we index the rows of $U$ and columns of $V$ by pairs $(k_1, k_2) \in [r(A_1)] \times [r(A_2)]$:
\begin{align*}
    U[i, (k_1, k_2)] &= U_1[i, k_1] + U_2[i, k_2], \\
    V[(k_1, k_2), j] &= V_1[k_1, j] + V_2[k_2, j].
\end{align*}
Each entry in $A_1 + A_2$ can be expressed as
\begin{equation*}
    A_1[i, j] + A_2[i, j] = U_1[i, k_1] + V_1[k_1, j] + U_2[i, k_2] + V_2[k_2, j] = U[i, (k_1, k_2)] + V[(k_1, k_2), j],
\end{equation*}
for some pair $(k_1, k_2) \in [r(A_1)] \times [r(A_2)]$. Thus, $U, V$ witnesses that $r(A_1 + A_2) \leq r(A_1) \cdot r(A_2)$.
\end{proof}

\subsection{Regular Rank Decompositions} \label{sec:rank:sec:decomp-reg}
We will now make the concept of a select-plus rank decomposition formal.

\begin{definition}[Rank Decomposition] \label{def:rank-decomp}
Let $A, U, V$ and $r$ be as in \cref{def:rank}. Let $S \in ([r] \cup \set{\bot})^{n \times m}$ be a matrix such that
\begin{equation*}
    A[i, j] =
    \begin{cases}
        U[i, S[i, j]] + V[S[i, j], j] &\text{if $S[i, j] \neq \bot$,} \\
        \bot &\text{if $S[i, j] = \bot$.}
    \end{cases}
\end{equation*}
Then we call $(U, V, S)$ a \emph{select-plus rank-$r$ decomposition} of $A$.
\end{definition}

Of course, knowing only $A$, $U$, and $V$ it is easy to complete the rank decomposition by constructing $S$ in time $O(n m r)$ (for all our applications this running time is strictly subcubic and hence tolerable). The benefit of \cref{def:rank-decomp} is rather that it is convenient to have access to the ``selector'' matrix $S$.

In a later key step, it will be crucial that the given rank decomposition satisfies an additional structural property, namely that it is \emph{row-} or \emph{column-regular} defined as follows.

\begin{definition}[Regular Rank Decomposition] \label{def:decomp-reg}
Let $(U, V, S)$ be a rank-$r$ decomposition of some $n \times m$ matrix, and let $R \geq 1$. We call $(U, V, S)$ \emph{$R$-row-regular} if, for all $\ell \in [r]$ and all $i \in [n]$,
\begin{equation*}
    |\set{ j \in [m] : S[i, j] = \ell }| \leq R \cdot \frac{m}{r}.
\end{equation*}
Analogously, we call $(U, V, S)$ \emph{$R$-column-regular} if, for all $\ell \in [r]$ and all $j \in [m]$,
\begin{equation*}
    |\set{ i \in [n] : S[i, j] = \ell }| \leq R \cdot \frac{n}{r}.
\end{equation*}
\end{definition}

The following lemma states that we can transform any rank decomposition into a regular one without loss of generality. Specifically, it states that any matrix $A$ with a given rank-$r$ decomposition can be split into a row-regular part $A_\row$, a column-regular part $A_\col$, and a \emph{small} part $A_\sm$ with strictly smaller rank,~$r/2$. Algorithmically we will simply \emph{recurse} on this small part, which effectively means that we can always treat~$A$ as partly row-regular and partly column-regular.

\begin{lemma}[Regular Rank Decomposition] \label{lem:decomp-reg}
For any matrix $A \in (\Int \cup \set{\bot})^{n \times m}$ with rank $r(A) \leq r$ there is a partition $A = A_\row \sqcup A_\col \sqcup A_\sm$ such that
\begin{itemize}
    \item $A_\row$ has an $O(\log (n m))$-row-regular rank-$r$ decomposition,
    \item $A_\col$ has an $O(\log (n m))$-column-regular rank-$r$ decomposition, and
    \item $A_\sm$ has rank $r(A_\sm) \leq r/2$.
\end{itemize}
Moreover, given $A$ and a corresponding rank-$r$ decomposition we can compute the matrices $A_\row, A_\col, A_\sm$ along with their claimed decompositions in deterministic time $\tilde O(n m r)$.
\end{lemma}

The proof of \cref{lem:decomp-reg} relies on the following covering-type lemma.

\begin{restatable}[Conflict-Free Covering]{lemma}{lemcovering} \label{lem:covering}
Let $x_1, \dots, x_n \in [r]$ and let $C_1, \dots, C_n \subseteq [r]$ be sets of size at most~$s$ such that $x_i \not\in C_i$ for all $i$. We say that a set $S \subseteq [r]$ \emph{covers} item $i$ if $x_i \in S$ and $C_i \cap S = \emptyset$. Then there is a collection $\mathcal S$ of size $|\mathcal S| \leq O(s \log n)$ so that each~\makebox{$i \in [n]$} is covered by some set $S \in \mathcal S$. Moreover, there is a deterministic $\tilde O(n r)$-time algorithm that constructs $\mathcal S$, along with a mapping $[n] \to \mathcal S$ indicating for each item~$i$ by which set it is covered.
\end{restatable}
\begin{proof}[Proof Sketch]
Here, we only provide a quick proof sketch based on a probabilistic argument; in \cref{sec:covering} we will then efficiently derandomize this argument based on the method of conditional expectations. The construction of $\mathcal S$ is simple: Take $O(s \log n)$ independently random sets $S \subseteq [r]$ with sample rate~$1/(2s)$. Each such set $S$ covers any fixed index $i$ with probability at least $\Omega(1/s)$. Indeed, the probability that $x_i \in S$ is $1/(2s)$, and independently we avoid all conflicts, $C_i \cap S = \emptyset$, with probability at least~\makebox{$1 - |C_i| / (2s) \geq 1/2$}. Hence, after $O(s \log n)$ repetitions we covered all items with high probability.
\end{proof}

\begin{proof}[Proof of \cref{lem:decomp-reg}]
Let $(U, V, S)$ denote the rank-$r$ decomposition of $A$, and let $R = \Theta(\log (n m))$ (where we determine the implied constant later). If $r \leq R$ then this rank-$r$ decomposition is trivially $R$-row-regular (and we can simply take $A_\row = A$ without modifying the decomposition), so suppose that $r > R$. For each~$\ell \in [r]$ define the sets
\begin{align*}
    I_\ell &= \set*{i \in [n] : |\set{ j \in [m] : S[i, j] = \ell }| > R \cdot \frac{m}{r} }, \\
    J_\ell &= \set*{j \in [m] : |\set{ i \in [n] : S[i, j] = \ell }| > R \cdot \frac{n}{r} }.
\end{align*}
That is, $I_\ell$ is the set of rows that violate the row-regularity condition for $\ell$, and $J_\ell$ is the set of columns that violate the column-regularity condition for $\ell$. Let~\makebox{$A = A_\row \sqcup A_\col \sqcup A_\sm$} denote the partition of $A$ where in~$A_\row$ we retain the entries $(i, j)$ with $i \not\in I_{S[i, j]}$, in~$A_\col$ we retain the entries $(i, j)$ with $i \in I_{S[i, j]}$ and~$j \not\in J_{S[i, j]}$, and in~$A_\sm$ we retain the entries $(i, j)$ with $i \in I_{S[i, j]}$ and $j \in J_{S[i, j]}$; we will call these entries $(i, j)$ \emph{irregular}. In the following we argue that the three matrices $A_\row, A_\col, A_\sm$ admit the claimed rank decompositions.

\proofparagraph{Rank Decomposition of \boldmath$A_\row$ and $A_\col$}
It is simple to construct an $R$-row-regular rank-$r$ decomposition $(U_\row, V_\row, S_\row)$ of~$A_\row$: Take $U_\row = U$ and $V_\row = V$, and let $S_\row$ be the restriction of $S$ to the entries $(i, j)$ satisfying that $i \not\in I_{S[i, j]}$. This is clearly a valid rank-$r$ decomposition of $A_\row$. To see that it is also $R$-row-regular, focus on any $\ell \in [r]$ and any row $i \in [n]$. In $S_\row$ we have only retained the entries~$S[i, j] = \ell$ with $i \not\in I_\ell$, hence by definition $\abs{\set{j \in [m] : S_\row[i, j] = \ell}} \leq R \cdot \frac{m}{r}$. The construction for~$A_\col$ is analogous.

\proofparagraph{Rank Decomposition of \boldmath$A_\sm$}
Next, we construct a rank-$r/2$ decomposition $(U_\sm, V_\sm, S_\sm)$ of $A_\sm$. In this step we will crucially rely on the covering lemma from before. Specifically, for each irregular pair $(i, j)$ define
\begin{align*}
    x_{i, j} &= S[i, j],
\intertext{and}
    C_{i, j} &= \set{\ell \in [r] : \text{$i \in I_\ell$ or $j \in J_\ell$}} \setminus \set{S[i, j]}.
\end{align*}
We run \cref{lem:covering} on the items $x_{i, j} \in [r]$ and conflict sets $C_{i, j} \subseteq [r]$. Let $\mathcal S$ denote the resulting collection of sets, and write $T_{i, j} \in \mathcal S$ for the set covering $(i, j)$. Let $U_\sm \in \Int^{[n] \times \mathcal S}, V_\sm \in \Int^{\mathcal S \times [m]}, S_\sm \in (\mathcal S \cup \set{\bot})^{n \times m}$ be the matrices defined as follows (here, for simplicity of notation we assume that the columns of $U_\sm$, the rows of $V_\sm$, and the entries of $S_\sm$ are indexed by sets $T \in \mathcal S$, but this can easily be replaced by the integers~$[|\mathcal S|]$):
\begin{align*}
    U_\sm[i, T] &=
    \begin{cases}
        \mathrlap{U[i, \ell]}\phantom{V[\ell, j]} &\text{for some $\ell \in T$ with $i \in I_\ell$,} \\
        0 &\text{if such $\ell$ does not exist,}
    \end{cases} \\
    V_\sm[T, j] &=
    \begin{cases}
        V[\ell, j] &\text{for some $\ell \in T$ with $j \in J_\ell$,} \\
        0 &\text{if such $\ell$ does not exist,}
    \end{cases} \\
    S_\sm[i, j] &=
    \begin{cases}
        \mathrlap{T_{i, j}}\phantom{V[\ell, j]} &\text{if $(i, j)$ is irregular,} \\
        \bot &\text{otherwise.}
    \end{cases}
\end{align*}
(The assignment $0$ of $U_\sm, V_\sm$ in the failure case is arbitrary.) We claim that $(U_\sm, V_\sm, S_\sm)$ is a valid decomposition of $A_\sm$. This is trivially true for the regular pairs, so fix some irregular pair $(i, j)$. \cref{lem:covering} guarantees that $x_{i, j} \in T_{i, j}$ and $C_{i, j} \cap T_{i, j} = \emptyset$. This means that $\ell := x_{i, j} = S[i, j]$ is the only element in $T_{i, j}$ that satisfies $i \in I_\ell$ or $j \in J_\ell$, and thus
\begin{align*}
    U_\sm[i, T_{i, j}] &= U[i, \ell], \\
    V_\sm[T_{i, j}, j] &= V[\ell, j].
\end{align*}
Hence,
\begin{equation*}
    U_\sm[i, S_\sm[i, j]] + V_\sm[S_\sm[i, j], j] = U_\sm[i, T_{i, j}] + V_\sm[T_{i, j}, j] = U[i, \ell] + V[\ell, j] = A_\sm[i, j].
\end{equation*}
It follows that we have indeed constructed a valid rank-$|\mathcal S|$ decomposition. We finally argue that $|\mathcal S| \leq r / 2$ for an appropriate choice of parameters. First observe that each conflict set $C_{i, j}$ has size at most $s := 2 r / R$, as each index $i$ is part of at most $r / R$ sets $I_\ell$, and similarly each index $j$ is part of at most $r / R$ sets $J_\ell$. Therefore, by \cref{lem:covering} we have $|\mathcal S| \leq O(s \log (nm)) = O(r \log (nm) / R)$. Hence, picking~\makebox{$R = \Theta(\log(n m))$} for some sufficiently large hidden constant, we get that $|\mathcal S| \leq r/2$ as was intended.

\proofparagraph{Running Time}
To compute the sets $I_\ell$ and $J_\ell$ we first precompute the statistics $|\set{j \in [m] : S[i, j] = \ell}|$ (for all $i, \ell$) and $|\set{i \in [n] : S[i, j] = \ell}|$ (for all $j, \ell$) in time $\tilde O(n m + n r + m r)$. Afterwards we can read off $I_\ell$ and~$J_\ell$ in the same running time. Then it is straightforward to compute the partition~\makebox{$A = A_\row \sqcup A_\col \sqcup A_\sm$} in time $\tilde O(n m)$. It remains to analyze the time to compute the rank decompositions of $A_\row$, $A_\col$ and~$A_\sm$. The decompositions of \smash{$A_\row$} and \smash{$A_\col$} are computable from $(U, V, S)$ in linear time $\tilde O(n m)$. For the decomposition of~$A_\sm$ the dominating step is the computation of $\mathcal S$ in time $\tilde O(n m r)$ by \cref{lem:covering}; the remaining steps run in time $\tilde O(n m + n r + m r)$.
\end{proof}
\section{Low-Rank Exact Triangle} \label{sec:exact-tri-low-rank}
The main result of this section is that Low-Rank Exact Triangle can be solved in subcubic time:

\begin{theorem}[Low-Rank Exact Triangle] \label{thm:exact-tri-low-rank}
The $r$-rank Exact Triangle problem can be solved in deterministic time
\begin{equation*}
    \ExactTri(n, n, n \mid r) \leq n^{3+o(1)} \cdot \parens*{\frac{r}{n^{3-\omega}}}^{1/200000}.
\end{equation*}
\end{theorem}

The point of this theorem is that for $r \leq n^{3-\omega-\epsilon}$ this running time is truly subcubic, $O(n^{3-\Omega(\epsilon)})$. We have not attempted to optimize the constant 200000. Likely, it can be improved dramatically, but for all our applications the constant does not matter. The proof of \cref{thm:exact-tri-low-rank} is an algorithmic reduction from the Low-Rank Exact Triangle problem via various intermediate steps to the Uniform Low-Doubling Exact Triangle problem:

\begin{restatable}[Reduction from Low-Rank to Uniform Low-Doubling Exact Triangle]{lemma}{lemexacttrilowranktolowdoubling} \label{lem:exact-tri-low-rank-to-low-doubling}
For any parameter $K \geq 1$ there is a potential-adjusting reduction from any \emph{$r$-rank} Exact Triangle instance to $K (n_1 n_2 n_3)^{o(1)}$ instances that are each \emph{$D$-uniform}, \emph{$1/D$-regular} and~\emph{$K$\=/doubling} for some $D \leq r$. The reduction runs in deterministic time $(n_1 n_2 n_3 / K^{1/98000} + (n_1 n_2 + n_1 n_3 + n_2 n_3) r K)^{1+o(1)}$. In particular:
\begin{align*}
    &\ExactTri(n_1, n_2, n_3 \mid r) \\
    &\qquad \leq \parens*{K \cdot \ExactTri(n_1, n_2, n_3 \mid D \leq r; \rho \leq 1/D; K) + \frac{n_1 n_2 n_3}{K^{1/98000}} + (n_1 n_2 + n_2 n_3 + n_1 n_3) r K}^{1+o(1)}.
\end{align*}
\end{restatable}

Intuitively, this lemma states that we can reduce an $r$-rank instance in subcubic time to few instances that are $r$-uniform and $K$-doubling, for an arbitrarily small polynomial $K = n^{\epsilon'}$. (In fact, the instances are also regular, but we do not need to rely on this extra constraint here.) This is a special case of Exact Triangle that can be solved in subcubic time by a simple algebraic algorithm; see~\cite[Lemma~5.5]{AbboudFJVX25}\footnote{A subtlety in~\cite[Lemma~5.5]{AbboudFJVX25} is that it only determines the edges $(i, j) \in [n_1] \times [n_3]$ that are involved in exact triangles, but not the edges $(i, k) \in [n_1] \times [n_2]$ and $(k, j) \in [n_2] \times [n_3]$. However, since in \cref{lem:exact-tri-low-doubling} $X$ denotes the set of entries in all three matrices $A, B, C$, we can simply apply~\cite[Lemma~5.5]{AbboudFJVX25} three times on all rotations of the instance $(A, B, C)$. Alternatively, as stated in~\cite{AbboudFJVX25}, one can apply the Baur-Strassen trick~\cite{Fischer24,BaurS83}.}.

\begin{lemma}[Uniform Low-Doubling Exact Triangle~\cite{AbboudFJVX25}] \label{lem:exact-tri-low-doubling}
Any Exact Triangle instance $(A, B, C)$ can be solved in deterministic time $\tilde O(\MM(n_1, n_2, n_3) \cdot |X + X|)$, where $X$ is the set of integer entries in~$A$,~$B$ and~$C$. In particular:
\begin{equation*}
    \ExactTri(n_1, n_2, n_3 \mid D ; K) = \tilde O(\MM(n_1, n_2, n_3) \cdot D K).
\end{equation*}
\end{lemma}

\begin{proof}[Proof of \cref{thm:exact-tri-low-rank}]
Combining \cref{lem:exact-tri-low-doubling,lem:exact-tri-low-rank-to-low-doubling} we obtain that, for any parameter $K \geq 1$,
\begin{align*}
    \ExactTri(n, n, n \mid r)
    &= \parens*{K \cdot \ExactTri(n, n, n \mid D \leq r ; K) + \frac{n^3}{K^{1/98000}} + n^2 r K}^{1+o(1)} \\
    &= \parens*{K \cdot n^\omega r K + \frac{n^3}{K^{1/98000}} + n^2 r K}^{1+o(1)} \\
    &= \parens*{n^\omega r K^2 + \frac{n^3}{K^{1/98000}}}^{1+o(1)}.
\end{align*}
We pick $K = (n^{3-\omega} / r)^{98/200}$ and argue that then the running time is as claimed. Indeed, the second term is clearly $n^{3+o(1)} (r / n^{3-\omega})^{1/200000}$. The first term can be rewritten as $n^{3+o(1)} (r / n^{3-\omega}) (n^{3-\omega} / r)^{98/100}$ which is at most $n^{3+o(1)} (r / n^{3-\omega})^{2/100}$ and thus dominated by the second term.
\end{proof}

\begin{figure}[t]
\newlength\probwidth\deflength\probwidth{22mm}%
\newlength\probpad\deflength\probpad{2mm}%
\newlength\probsep\deflength\probsep{0.25\textwidth-0.25\probwidth-0.5\probpad-0.1pt}%
\centering%
\begin{tikzpicture}[
    prob/.style={
        rectangle split,
        rectangle split parts=2,
        text width=\probwidth,
        inner sep=\probpad,
        execute at begin node=\def\captionnow{\nodepart{two}},
        draw,
        rectangle split part fill={LightBlue, PaleBlue},
        rounded corners=2mm,
        align=center,
        every text node part/.style={align=left},
    },
    reduct/.style={
        draw,
        >=latex,
        shorten <=1mm,
        shorten >=1mm,
        rounded corners,
    },
]
\path[every node/.append style={prob, anchor=north}]
    (0, 0) node (0) {Low-Rank\captionnow$r \leq n^{3-\omega-\epsilon}$}
    ++(\probsep, 0) node (1) {Slice-Uniform\captionnow$d \leq n^{3-\omega-\epsilon}$}
    ++(\probsep, 0) node (2) {Uniform\captionnow$D \leq n^{3-\omega-\epsilon}$}
    ++(\probsep, 0) node (3) {Uniform\\Regular\captionnow$\begin{aligned}D &\leq n^{3-\omega-\epsilon} \\ \rho &\leq 1/D\end{aligned}$}
    ++(\probsep, 0) node (4) {Uniform\\Low-Doubling\captionnow$\begin{aligned}D &\leq n^{3-\omega-\epsilon} \\ K &\leq \smash{n^{\epsilon'}}\end{aligned}$};

\newcommand\reduct[3]{%
    \path[reduct, ->] ([shift={(6pt, 0pt)}]#1.south)
        -- ([shift={(6pt, -16pt)}]#1.south |- #2.south)
        -- node[below, font=\small] {#3} ([shift={(-6pt, -16pt)}]#2.south)
        -- ([shift={(-6pt, 0pt)}]#2.south);}
\reduct{0}{1}{\vphantom{[}\cref{sec:exact-tri-low-rank:sec:low-rank-to-slice-unif}}
\reduct{1}{2}{\vphantom{[}\cref{sec:exact-tri-low-rank:sec:slice-unif-to-unif}~\cite{AbboudFJVX25}}
\reduct{2}{3}{\vphantom{[}\cref{sec:exact-tri-low-rank:sec:low-rank-to-unif-regular}}
\reduct{3}{4}{\vphantom{[}\cref{sec:exact-tri-low-rank:sec:unif-reg-to-low-doubling}~\cite{AbboudFJVX25}}
\path[reduct, ->, dashed] ([shift={(6pt, -16pt + 2mm)}]2.south |- 3.south)
    -- ([shift={(6pt, -16pt)}]2.south |- 3.south)
    -- ([shift={(-6pt, -16pt)}]0.south |- 3.south)
    -- ([shift={(-6pt, 0pt)}]0.south);
\end{tikzpicture}
\vspace{-3mm}%
\caption{Illustrates the four steps in the reduction from Low-Rank Exact Triangle to Uniform Low-Doubling Exact Triangle. (The dashed arrow symbolizes a recursive dependence; see \cref{sec:exact-tri-low-rank:sec:low-rank-to-unif-regular}.)} \label{fig:low-rank-steps}
\end{figure}
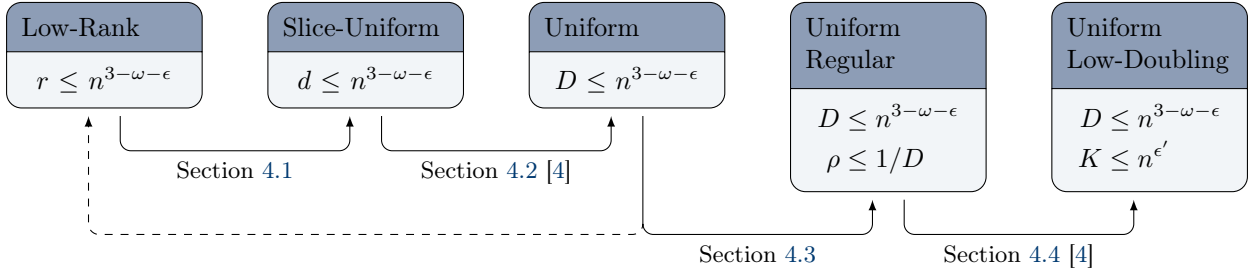

For the remainder of this section, we focus on the proof of \cref{lem:exact-tri-low-rank-to-low-doubling}. The reduction can be split into four individual steps, which we treat in \crefrange{sec:exact-tri-low-rank:sec:low-rank-to-slice-unif}{sec:exact-tri-low-rank:sec:unif-reg-to-low-doubling}. See \cref{fig:low-rank-steps} for an outline. Step 1 is new and can be regarded as one of the key technical parts of this paper. Steps 2 and 4 are already established in~\cite{AbboudFJVX25} and can be taken almost without modifications; we provide full proofs here for the sake of completeness and because we need appropriately generalized statements throughout (such as a rectangular version, and the assertion that all reductions are potential-adjusting). Step 3 could be taken as in~\cite{AbboudFJVX25}, but we give an alternative proof that significantly improves the parameters (otherwise we would get the weaker statement that $n^{3-\omega-\epsilon}$-rank Exact Triangle can be solved in time $O(n^{3-\Omega_\epsilon(1)})$ where the hidden dependence on $\epsilon$ is worse than $\Omega(\epsilon)$).

Throughout, we will consider the rectangular case of size $n_1 \times n_2 \times n_3$. Moreover, we will often implicitly assume that the parameter in question ($r, d$ or $D$) satisfies that $r \leq n_1, n_2, n_3$. This is without loss of generality, as all algorithms we consider spend time at least $\Omega(r \cdot (n_1 n_2 + n_1 n_3 + n_2 n_3))$, and hence if the assumption was violated we could alternatively solve the instance naively in time $O(n_1 n_2 n_3)$.

\subsection{Reduction from Low-Rank to Slice-Uniform Exact Triangle} \label{sec:exact-tri-low-rank:sec:low-rank-to-slice-unif}
The following lemma captures the first step in our chain of reductions. An important ingredient for its proof is the previously established \cref{lem:decomp-reg}.

\begin{lemma}[Reduction from Low-Rank to Slice-Uniform Exact Triangle] \label{lem:exact-tri-low-rank-to-slice-unif}
For any parameter $t \geq 1$ there is a potential-adjusting reduction from any \emph{rank\=/$r$} Exact Triangle instance to $O(t^2 \log r)$ \emph{$r$-slice-uniform} instances. The reduction runs in deterministic time $\tilde O(n_1 n_2 n_3 / t + (n_1 n_2 + n_2 n_3 + n_1 n_3) r t^2)$. In particular:
\begin{equation*}
    \ExactTri(n_1, n_2, n_3 \mid r) \leq \tilde O\parens*{t^2 \cdot \ExactTri(n_1, n_2, n_3 \mid d \leq r) + \frac{n_1 n_2 n_3}{t} + (n_1 n_2 + n_2 n_3 + n_1 n_3) r t^2}.
\end{equation*}
\end{lemma}
\begin{proof}
Let us assume that we are given a low-rank decomposition of $C$; the statements for $A$ and $B$ follow by first rotating the given instance. By \cref{lem:decomp-reg} we may further assume that the given rank\=/$r$ decomposition of~$C$ is $R$-row-regular for some $R = O(\log (n_1 n_3))$. Formally, we apply the lemma to partition~$C$ into~\makebox{$C_\row \sqcup C_\col \sqcup C_\sm$}, where $C_\row$ has an $R$-row-regular rank-$r$ decomposition, $C_\col$ has an $R$-column-regular rank-$r$ decomposition, and $C_\sm$ has a rank-$r/2$ decomposition. It then remains to solve the Exact Triangle instances $(A, B, C_\row)$, $(A, B, C_\col)$ and $(A, B, C_\sm)$. We will describe how to solve $(A, B, C_\row)$ in the following. The instance $(A, B, C_\col)$ can be equivalently viewed as $(B^T, A^T, C_\col^T)$; in this case, too, we have access to a $R$-row-regular rank-$r$ decomposition of $C_\col^T$, so we can solve the instance also as described in the following. Finally, we solve $(A, B, C_\sm)$ recursively---since the rank $r$ halves with every recursive call this only leads to a logarithmic overhead.

Let $(U, V, S)$ denote the $R$-row-regular rank-$r$ decomposition of $C$. Let $\mathcal T \gets \emptyset$ be a set of exceptional triples, and let $\mathcal I \gets \emptyset$ be a set of $r$-slice-uniform potential adjustments of $(A, B, C)$. Consider the following algorithm. In an outer loop, we enumerate $\ell \gets 1, \dots, r$. We maintain the invariant that after the $\ell$-th iteration each exact triangle $(i, k, j)$ with $S[i, j] = \ell$ has been inserted into $\mathcal T$ or is present in one of the instances in $\mathcal I$.
\begin{enumerate}
    \item\label{lem:exact-tri-low-rank-to-slice-unif:itm:setup} \emph{(Heavy/Light Setup)} Define the matrices $A_\ell$ and $B_\ell$ of sizes $n_1 \times n_2$ and $n_2 \times n_3$, respectively, by
    \begin{align*}
        A_\ell[i, k] &= A[i, k] - U[i, \ell], \\
        B_\ell[k, j] &= -(B[k, j] - V[\ell, j]).
    \end{align*}
    Note that for all pairs $(i, j)$ with $S[i, j] = \ell$ the conditions $A[i, k] + B[k, j] = C[i, j]$ and $A_\ell[i, k] = B_\ell[k, j]$ are equivalent. We now classify the non-$\bot$ entries $(i, k)$ in $A_\ell$ as follows: We say that $(i, k)$ is \emph{light} if
    \begin{equation*}
        |\set{i' \in [n_1] : A_\ell[i, k] = A_\ell[i', k]}| \leq \frac{n_1}{r t},
    \end{equation*}
    and \emph{heavy} otherwise. In the following two steps we separately deal with the exact triangles $(i, k, j)$ involving light and heavy edges $(i, k)$, respectively.

    \begin{framed}
        \emph{Running Time:} $O((n_1 n_2 + n_2 n_3 + n_1 n_3) r)$ over all iterations.
    \end{framed}

    \item\label{lem:exact-tri-low-rank-to-slice-unif:itm:light} \emph{(Light Edges)} We enumerate all triples $(i, k, j)$ with $A_\ell[i, k] = B_\ell[k, j]$ for which $(i, k)$ is light as follows. Enumerate all pairs $(k, j) \in [n_2] \times [n_3]$. Only if there is a light edge $(i', k)$ do we enumerate the at most~$n_1 / (r t)$ rows $i$ with $A_\ell[i, k] = B_\ell[k, j]$. We insert each triple $(i, k, j)$ with $A[i, k] + B[k, j] = C[i, j]$ into the set $\mathcal T$. This correctly maintains the invariant for all exact triangles involving light edges.

    \begin{framed}
        \emph{Running Time:} $O(n_1 n_2 n_3 / (r t))$ per iteration, and thus $O(n_1 n_2 n_3 / t)$ across all iterations.
    \end{framed}

    \item\label{lem:exact-tri-low-rank-to-slice-unif:itm:heavy} \emph{(Heavy Edges)} We distinguish two cases depending on the number of heavy edges:
    \begin{enumerate}
        \item\label{lem:exact-tri-low-rank-to-slice-unif:itm:setup:itm:few} \emph{(Few Heavy Edges)} If there are at most $n_1 n_2 / t$ heavy edges, then we apply brute-force as follows. Enumerate all heavy pairs $(i, k)$, and enumerate all columns $j$ with $S[i, j] = \ell$. Each triple $(i, k, j)$ enumerated in this way that satisfies $A[i, k] + B[k, j] = C[i, j]$ is inserted into $\mathcal T$. This again clearly maintains the invariant for all exact triangles involving a heavy edge.

        \begin{framed}
            \emph{Running Time:} For each row $i$ there are at most $R \cdot (n_3 / r)$ columns $j$ with $S[i, j] = \ell$ as the given rank decomposition is assumed to be $R$-row-regular. Hence the running time is $O(n_1 n_2 / t \cdot R \cdot (n_3 / r)) = \tilde O(n_1 n_2 n_3 / (r t))$ per iteration, and thus $\tilde O(n_1 n_2 n_3 / t)$ across all iterations.
        \end{framed}

        \item\label{lem:exact-tri-low-rank-to-slice-unif:itm:setup:itm:many} \emph{(Many Heavy Edges)} The remaining case is that there are more than $n_1 n_2 / t$ heavy pairs. Let $H$ denote the set of heavy pairs $(i, k)$, and let~$A_H$ be the matrix $A$ restricted to the entries in $H$ (where all other entries are set to $\bot$). Our strategy is to solve the entire instance $(A_H, B, C)$ in this step (by means of the reduction), and afterwards to remove all heavy entries from $A$ \emph{globally}. To this end, construct matrices $A'$ and $C'$ by
        \begin{alignat*}{2}
            A'[i, k] &= A[i, k] - U[i, \ell] \qquad &&\text{if $(i, k) \in H$,} \\
            A'[i, k] &= \bot \qquad &&\text{otherwise,}
        \intertext{and}
            C'[i, j] &= C[i, j] - U[i, \ell]. \qquad
        \end{alignat*}
        By definition $(A', B, C')$ is a potential-adjusting restriction of $(A, B, C)$ that retains all the exact triangles involving heavy edges. Moreover, we claim that each column in $A'$ contains at most $r t$ distinct entries. Indeed, by construction $A'$ coincides with~$A_\ell$ on the non-$\bot$ (i.e., heavy) entries, and by definition each heavy entry appears at least~$n_1 / (r t)$ times in its column in $A_\ell$. It follows that the instance $(A', B, C')$ is $rt$-slice-uniform. We now further reduce $(A', B, C')$ to $t$ instances that are $r$-slice-uniform (in the trivial way) and insert these instances into the collection $\mathcal I$. 

        Afterwards, we remove all entries in $H$ \emph{globally}, i.e., we set $A[i, k] \gets \bot$ for all $(i, k) \in H$ for all future iterations.

        Regarding the correctness, observe that each exact triangle $(i, k, j)$ with $(i, k) \in H$ and \makebox{$S[i, j] = \ell$} is preserved in $(A', B, C')$. In fact, we have dealt with \emph{all} exact triangles with $(i, k) \in H$ (irrespective of $S[i, j]$), so we can indeed safely remove all pairs in $H$ from $A$ for the future iterations.

        \begin{framed}
            \emph{Running Time:} The construction of $H$ and of the $t$ equivalent instances takes linear time $\tilde O((n_1 n_2 + n_1 n_3 + n_2 n_3) t)$ per iteration. Moreover, with each execution of step~\ref{lem:exact-tri-low-rank-to-slice-unif:itm:setup:itm:many} we increase the size of $\mathcal I$ by at most $t$. Notably, however, we execute this step~\ref{lem:exact-tri-low-rank-to-slice-unif:itm:setup:itm:many} at most $t$ times in total, as after each iteration we replace $|H| \geq n_1 n_2 / t$ non-$\bot$ entries in $A$ by $\bot$. It follows that the total running time is $\tilde O((n_1 n_2 + n_1 n_3 + n_2 n_3) t^2)$ and that $|\mathcal I| \leq t^2$ throughout.
        \end{framed}
    \end{enumerate}
\end{enumerate}
This completes the proof. The correctness should be clear from the in-text explanations, and the total running time as analyzed before is $\tilde O(n_1 n_2 n_3 / t + (n_1 n_2 + n_1 n_3 + n_2 n_3) r t^2)$ as claimed. (Here we only loosely bound the dependence on $t$.)
\end{proof}

\subsection{Reduction from Slice-Uniform to Uniform Exact Triangle} \label{sec:exact-tri-low-rank:sec:slice-unif-to-unif}
The second step is captured by the following lemma. Its proof can be fully attributed to~\cite{AbboudFJVX25}---no new ideas are necessary. We include a complete proof here only because we need the more general rectangular version, and to point out that the reduction is potential-adjusting. Readers familiar with~\cite{AbboudFJVX25} can safely skip this section.

\begin{lemma}[Reduction from Slice-Uniform to Uniform Exact Triangle~\cite{AbboudFJVX25}] \label{lem:exact-tri-slice-unif-to-unif}
For any parameter $t \geq 1$ there is a potential-adjusting reduction from any \emph{$d$\=/slice-uniform} Exact Triangle instance to $O(t^{41} \log^2(n_1 n_2 n_3))$ \emph{$d$\=/uniform} instances. The reduction runs in deterministic time \makebox{$\tilde O(n_1 n_2 n_3 / t + (n_1 n_2 + n_2 n_3 + n_1 n_3) d t^{41})$}. In particular:
\begin{equation*}
    \ExactTri(n_1, n_2, n_3 \mid d) \leq \tilde O\parens*{t^{41} \cdot \ExactTri(n_1, n_2, n_3 \mid D \leq d) + \frac{n_1 n_2 n_3}{t} + (n_1 n_2 + n_2 n_3 + n_1 n_3) d t^{41}}.
\end{equation*}
\end{lemma}

The proof of \cref{lem:exact-tri-slice-unif-to-unif} requires some setup. Recall that $r_{X+Y}(z)$ denotes the multiplicity of $z$ in the sumset $X + Y$, i.e., the number of representations of the form $z = x + y$ for $(x, y) \in X \times Y$.

\begin{definition}[Popular Sums]
For two integer sets $X, Y \subseteq \Int$, we define their set of \emph{$s$-popular sums} as
\begin{equation*}
    P_s(X, Y) = \set{ z \in X + Y : r_{X+Y}(z) \geq s}.
\end{equation*}
\end{definition}

\begin{observation} \label{obs:pop-sum-size}
$|P_s(X, Y)| \leq |X| \, |Y| / s$.
\end{observation}

It is easy to compute (a superset of) the set of $s$-popular sums in randomized time $\tilde O(|X| \, |Y| / s)$. The idea is to subsample $X' \subseteq X$ with rate roughly $1/s$; with good probability the popular sums are preserved in~$X' + Y$. Fischer, Jin, and Xu~\cite{FischerJX25} showed that with more effort this task can be derandomized in essentially the same running time.

\begin{lemma}[{{Approximating Popular Sums~\cite[Theorem~1.9]{FischerJX25}}}] \label{lem:approx-pop-sums}
Given sets $X, Y \subseteq [u]$ and $s \geq 1$, we can compute $P_{2s}(X, Y) \subseteq P \subseteq P_s(X, Y)$ in randomized time $\tilde O(|X| \, |Y| / s)$ or deterministic time~\makebox{$|X| \, |Y| / s \cdot u^{o(1)}$}.
\end{lemma}

The following lemma---the so-called ``popular sum decomposition''---is the key insight behind this reduction step from~\cite{AbboudFJVX25}. Intuitively, it states that sets $X_1, \dots, X_n$ and $Y_1, \dots, Y_m$ can be decomposed into few structured parts that look exactly alike (up to shifts), plus some remainders $X_1^*, \dots, X_n^*$ and $Y_1^*, \dots, Y_n^*$. Almost all remainder pairs $X_i^*, Y_j^*$, however, have \emph{no} popular sum.

\begin{lemma}[Popular Sum Decomposition~\cite{AbboudFJVX25}] \label{lem:pop-sum-decomp}
Let $X_1, \dots, X_n, Y_1, \dots, Y_m \subseteq [u]$ be sets of size at most $d$, and let $p \geq 1$. Then there exist partitions
\begin{align*}
    X_i &= X_i^* \sqcup \bigsqcup_{g=1}^{p^2} X_{i, g}, \\
    Y_j &= Y_j^* \sqcup \bigsqcup_{h=1}^{p^2} Y_{j, h}
\end{align*}
satisfying the following two properties:
\begin{enumerate}
    \item There are sets $S_g$ of size $|S_g| \leq d$ and shifts $s_{i, g}$ so that $X_{i, g} \subseteq S_g + s_{i, g}$ (for all $i \in [n], g \in [p^2]$), and\newline there are sets $T_h$ of size $|T_h| \leq d$ and shifts $t_{j, h}$ so that $Y_{j, h} \subseteq T_h + t_{j, h}$ (for all $j \in [m], h \in [p^2]$).
    \item $\abs{\set{(i, j) \in [n] \times [m] : P_{2d / p}(X_i^*, Y_j) \neq \emptyset}} \leq n m / p$, and\newline $\abs{\set{(i, j) \in [n] \times [m] : P_{2d / p}(X_i, Y_j^*) \neq \emptyset}} \leq n m / p$.
\end{enumerate}
Moreover, there is an algorithm that, given the sets $X_i, Y_j$, computes the partitions along with the sets $S_g, T_h$ and the shifts $s_{i, \ell}, t_{j, \ell}$ in randomized time \smash{$\tilde O(nmdp^3)$} or deterministic time \smash{$n m d p^3 \cdot u^{o(1)}$}. 
\end{lemma}
\begin{proof}
We only show how to construct the partition of $X$ and the accompanying sets~$S_g$ and shifts~$s_{i, g}$. The partition of $Y$ is obtained symmetrically by exchanging $X_1, \dots, X_n$ and $Y_1, \dots, Y_m$. Consider the following algorithm. For each $g \gets 1, 2, \dots$ we run the following steps:
\begin{enumerate}
    \item\label{lem:pop-sum-decomp:itm:approx-pop} Apply \cref{lem:approx-pop-sums} to compute sets $P_{2d / p}(X_i, Y_j) \subseteq P_{i, j} \subseteq P_{d / p}(X_i, Y_j)$ for all $(i, j) \in [n] \times [m]$.
    \item\label{lem:pop-sum-decomp:itm:stop} If $\abs{\set{(i, j) \in [n] \times [m] : P_{i, j} \neq \emptyset}} \leq n m / p$: Pick $X_i^* \gets X_i$ for all $i \in [n]$ and stop the algorithm.
    \item\label{lem:pop-sum-decomp:itm:next} Otherwise: Pick some $j \in [m]$ with $\abs{\set{i \in [n] : P_{i, j} \neq \emptyset}} \geq n / p$. We choose $S_g = -Y_j$. Let us call an index~$i$ \emph{good} if $P_{i, j} \neq \emptyset$ and \emph{bad} otherwise. For each bad $i$ we choose $X_{i, g} = \emptyset$ and $s_{i, g}$ arbitrarily. For each good $i$ we instead choose an arbitrary shift $s_{i, g} \in P_{i, j}$. Then we pick $X_{i, g} = X_i \cap (S_g + s_{i, g})$. This forms the next part in the partition of $X_i$, so we can safely remove it, $X_i \gets X_i \setminus X_{i, g}$.
\end{enumerate}

\proofparagraph{Correctness}
It is clear that if the algorithm terminates then we have satisfied the two claimed properties: By construction $X_i$ is indeed partitioned into $X_{i, 1}, X_{i, 2}, \dots$ and $X_i^*$, we have that $X_{i, g} \subseteq S_g + s_{i, g}$, and when the algorithm terminates then also $\abs{\set{(i, j) \in [n] \times [m] : P_{2d / p}(X_i^*, Y_j) \neq \emptyset}} \leq n m / p$ (as $P_{2d / p}(X_i, Y_j) \subseteq P_{i, j}$).

It remains to prove that the algorithm terminates after at most $p^2$ iterations. To see this, focus on an arbitrary iteration $g$, let $j$ be the index picked in step~\ref{lem:pop-sum-decomp:itm:next}, and let $i$ be good in this iteration. We claim that the constructed set $X_{i, g}$ has size at least $d / p$. Indeed, as $s_{i, g} \in P_{i, j} \subseteq P_{d / p}(X_i, Y_j)$ we have that
\begin{equation*}
    |X_{i, g}| = |X_i \cap (s_{i, g} + S_g)| = |X_i \cap (s_{i, g} - Y_j)| = r_{X_i+Y_j}(s_{i, g}) \geq \frac{d}{p}.
\end{equation*}
Thus, we reduce the size $|X_i|$ of every good index $i$ by at least $d / p$. In total there are at least $n / p$ such good indices, hence in each step we reduce $\sum_{i \in [n]} |X_i|$ by at least $nd / p^2$. Initially we have $\sum_{i \in [n]} |X_i| \leq nd$, and so the claim follows.

\proofparagraph{Running Time}
Focus on a single iteration. In step~\ref{lem:pop-sum-decomp:itm:approx-pop} we apply \cref{lem:approx-pop-sums} $nm$ times. Each execution takes deterministic time $d^2 / (d / p) \cdot u^{o(1)} = d p \cdot u^{o(1)}$, so step~\ref{lem:pop-sum-decomp:itm:approx-pop} takes time $n m d p \cdot u^{o(1)}$. Testing the condition in step~\ref{lem:pop-sum-decomp:itm:stop} takes negligible time $O(n m)$. In step~\ref{lem:pop-sum-decomp:itm:next} we spend time proportional to the total set sizes, $O(n m d)$. Summing over these contributions and over the at most $p^2$ iterations, the claimed time bound $n m d p^3 \cdot u^{o(1)}$ follows. The randomized time bound follows analogously.
\end{proof}

\begin{proof}[Proof of \cref{lem:exact-tri-slice-unif-to-unif}]
Let $(A, B, C)$ denote the given $d$-slice-uniform instance. Note that the statement is symmetric with respect to the sizes $n_1, n_2, n_3$, hence by rotating and/or transposing the given instance we may assume without loss of generality that the rows of $A$ contain at most~$d$ distinct entries each.

\proofparagraph{Outer Regularization}
Before we give the core reduction, we run the following simple preprocessing step to ensure that all entries in $A$ appear roughly equally often in their columns, and similarly for $B$. We partition $A$ into $L = \ceil{\log n_2}$ matrices $A_0, \dots, A_L$ where in the matrix $A_\ell$ we include all entries that appear at least $2^\ell$ and at most $2^{\ell+1}$ times in their respective \emph{row}. Similarly partition $B$ into $B_0, \dots, B_L$ according to the frequencies per \emph{column}. It remains to deal with all pairs of instances $(A_\ell, B_{\ell'}, C)$ (for~\makebox{$\ell, \ell' \in [L]$}).

For the rest of the proof, fix one such instance, and let us write $A = A_\ell$ and $B = B_{\ell'}$ for simplicity of notation. Let $d_A = \min\set{d, n_2 / 2^\ell}$ and $d_B = n_2 / 2^{\ell'}$, and verify that there are at most $d_A$ distinct entries per row in $A$, and at most $d_B$ distinct entries per column in $B$. Let $d' = \max\set{d_A, d_B}$.

Finally, we argue that we may assume that $d_B \leq d_A \cdot t$ in the following. Indeed, if instead $d_B \geq d_A \cdot t$ then we can solve the instance by brute-force as follows. Enumerate all pairs $(i, j) \in [n_1] \times [n_3]$, enumerate the at most $d_A$ entries $x$ in the $i$-th row of $A$, and enumerate all $k \in [n_2]$ with $B[k, j] = C[i, j] - x$. We insert all exact triangles $(i, k, j)$ enumerated in this step into $\mathcal T$.

\begin{framed}
    \emph{Running Time:} $\tilde O(n_1 n_3 \cdot d_A \cdot (n / d_B)) = \tilde O(n_1 n_2 n_3 / t)$.
\end{framed}

\proofparagraph{Uniformization}
Let $X_i$ denote the set of entries in the $i$-th row of $A$. Let $Y_j$ denote the set of entries in the $j$-th column of $B$. All these sets have size at most $d'$. Run the popular sum decomposition (\cref{lem:pop-sum-decomp}) for these sets, and some parameter $p \geq 1$ to be determined later. We call an exact triangle $(i, k, j)$ \emph{$A$-exceptional} if $A[i, k] \in X_i^*$, \emph{$B$-exceptional} if $B[k, j] \in Y_j^*$, and \emph{ordinary} otherwise.

\begin{framed}
    \emph{Running Time:} $\tilde O(n_1 n_3 d p^3 \cdot u^{o(1)})$ (recall that throughout $u = (n_1 n_2 n_3)^{O(1)}$, hence $u^{o(1)} = (n_1 n_2 n_3)^{o(1)}$).
\end{framed}

\begin{enumerate}
    \item\label{lem:exact-tri-slice-unif-to-unif:itm:exc} \emph{(Exceptional Triangles)} We will explicitly enumerate all exceptional exact triangles. Here we describe how to list all $A$-exceptional exact triangles and omit the analogous argument for listing $B$-exceptional triangles. The first step is to apply \cref{lem:approx-pop-sums} to compute sets $P_{4d' / p}(X_i^*, Y_j) \subseteq P_{i, j} \subseteq P_{2d' / p}(X_i^*, Y_j)$, for all pairs $(i, j) \in [n_1] \times [n_3]$. We call an $A$\=/exceptional triangle $(i, k, j)$ \emph{popular} if $C[i, j] \in P_{i, j}$ and \emph{unpopular} otherwise. We will deal with the popular and unpopular exceptional triangles in two separate steps:
    \begin{enumerate}
        \item\label{lem:exact-tri-slice-unif-to-unif:itm:exc:itm:unpop} \emph{(Unpopular Triangles)} Enumerate all pairs $(i, j) \in [n_1] \times [n_3]$ with $C[i, j] \not\in P_{i, j}$. By definition, $C[i, j] \not\in P_{i, j} \supseteq P_{4d' / p}(X_i^*, Y_j)$ can be expressed as a sum $x + y$ of elements $x \in X_i^*$ and $y \in Y_j$ in at most $4d' / p$ ways. Enumerate all such representations $(x, y)$, then enumerate all $k \in [n_2]$ with $A[i, k] = x$ and test if $(i, k, j)$ is an exact triangle; if yes, we insert it into the set $\mathcal T$.

        \begin{framed}
            \emph{Running Time:} For each pair $(i, j)$ we spend time $O(d_A)$ to list the relevant representations, plus time $O(d' / p \cdot n_2 / d_A)$ to enumerate all pairs $x, y, k$. The total time is:
            \begin{equation*}
                O\parens*{n_1 n_3 \cdot \parens*{d_A + \frac{d'}{p} \cdot \frac{n_2}{d_A}}} = O\parens*{n_1 n_3 \cdot \parens*{d + \frac{d_A t}{p} \cdot \frac{n_2}{d_A}}} = O\parens*{n_1 n_3 d + \frac{n_1 n_2 n_3 t}{p}}.
            \end{equation*}
        \end{framed}
        
        \item\label{lem:exact-tri-slice-unif-to-unif:itm:exc:itm:pop} \emph{(Popular Triangles)} The number of pairs $(i, j)$ with \makebox{$P_{i, j} \subseteq P_{2d' / p}(X_i^*, Y_j) \neq \emptyset$} is at most $n_1 n_3 / t$ by \cref{lem:pop-sum-decomp}. For each such pair we enumerate all possible $k \in [n_2]$, test if $(i, k, j)$ is an exact triangle, and in that case store $(i, k, j)$ in $\mathcal T$.

        \begin{framed}
            \emph{Running Time:} $O(n_1 n_2 n_3 / t)$.
        \end{framed}
    \end{enumerate}
    \item\label{lem:exact-tri-slice-unif-to-unif:itm:ord} \emph{(Ordinary Triangles)} Enumerate all pairs $g, h \in [p^2]$. Let $A_h, B_g, C_{g, h}$ be the matrices defined by
    \begin{align*}
        A_g[i, k] &=
        \begin{cases}
            \mathrlap{A[i, k] - s_{i, g}}\phantom{C[i, j] - s_{i, g} - t_{j, h}} &\text{if $A[i, k] \in X_{i, g}$,} \\
            \bot &\text{otherwise,}
        \end{cases} \\[\smallskipamount]
        B_h[k, j] &=
        \begin{cases}
            \mathrlap{B[k, j] - t_{j, h}}\phantom{C[i, j] - s_{i, g} - t_{j, h}} &\text{if $B[k, j] \in Y_{j, h}$,} \\
            \bot &\text{otherwise,}
        \end{cases} \\[\smallskipamount]
        C_{g, h}[i, j] &=
        \begin{cases}
            C[i, j] - s_{i, g} - t_{j, h} &\text{if $C[i, j] \neq \bot$,} \\
            \bot &\text{otherwise.}
        \end{cases}
    \end{align*}
    Each ordinary triangle $(i, k, j)$ appears in exactly one Exact Triangle instance $(A_g, B_h, C_{g, h})$, namely the instance with $A[i, k] \in X_{i, g}$ and $B[k, j] \in Y_{j, h}$. Moreover, the value of the triangle is unchanged in that instance as the terms $s_{i, g}$ and $t_{j, h}$ cancel:
    \begin{equation*}
        A_g[i, k] + B_h[k, j] - C_{g, h}[i, j] = A[i, k] + B[k, j] - C[i, j].
    \end{equation*}
    It is easily verified that all non-$\bot$ entries in $A_g$ stem from the set $S_g$ (as $X_{i, g} - s_{i, g} \subseteq S_g$) and all non-$\bot$ entries in $B_h$ stem from the set $T_h$ (as $Y_{j, h} - t_{j, h} \subseteq T_h$). Recall that further $|S_g|, |T_h| \leq d'$, so in each instance $(A_g, B_h, C_{g, h})$ the matrices $A_g$ and $B_h$ satisfy the $d'$-uniformity condition.
    
    In the following we will further modify $C_{g, h}$ to also become $d'$-uniform. Let $q$ be another parameter to be determined later, and compute the set $P := P_{d' / q}(S_g, T_h)$ (by brute-force, say, in negligible time). Then call a pair~$(i, j)$ \emph{popular} if $C_{g, h}[i, j] \in P$ and \emph{unpopular} otherwise. Consider the following two cases:
    \begin{enumerate}
        \item\label{lem:exact-tri-slice-unif-to-unif:itm:ord:itm:unpop} \emph{(Unpopular Entries)} This case is similar to step~\ref{lem:exact-tri-slice-unif-to-unif:itm:exc:itm:unpop}. Our goal is to list all exact triangles $(i, k, j)$ in $(A_g, B_h, C_{g, h})$ where $(i, j)$ is unpopular. To this end enumerate all unpopular pairs $(i, j)$. Recall that there are at most $d' / q$ representations of $C[i, j]$ as a sum $x + y$ where $x \in S_g$ and $y \in T_h$. We enumerate all such representations, and further all $k \in [n_2]$ with $A_g[i, k] = x$. We test if $(i, k, j)$ is an exact triangle, and in this case store $(i, k, j)$ in $\mathcal T$.

        \begin{framed}
            \emph{Running Time:} For each unpopular pair $(i, j)$ we spend time $O(d')$ to list all relevant representations, plus time $O(d' / q \cdot n_2 / d_A)$ to enumerate all $x, y, k$. In total, we repeat this step~$p^4$ times, leading to a total time of:
            \begin{align*}
                O\parens*{p^4 \cdot n_1 n_3 \cdot \parens*{d' + \frac{d'}{q} \cdot \frac{n_2}{d_A}}}
                &= O\parens*{p^4 \cdot n_1 n_3 \cdot \parens*{d t + \frac{d_A t}{q} \cdot \frac{n_2}{d_A}}} \\
                &= O\parens*{n_1 n_2 n_3 d t p^4 + \frac{n_1 n_2 n_3 t p^4}{q}}.
            \end{align*}
        \end{framed}
        
        \item\label{lem:exact-tri-slice-unif-to-unif:itm:ord:itm:pop} \emph{(Popular Entries)} Let $C^\pop_{g, h}$ be the restriction of $C_{g, h}$ to the popular entries (where we replace all unpopular entries by $\bot$). By \cref{obs:pop-sum-size} we have that $|P| \leq |S_g| \, |T_h| / (d' / q) \leq d' q$, so by construction the instance \smash{$(A_g, B_h, C^\pop_{g, h})$} is $d' q$-uniform. We trivially reduce that instance further to at most \smash{$((d' q) / d)^3 \leq t^3 q^3$} instances that are $d$-uniform and store the resulting instances in~$\mathcal I$.

        \begin{framed}
            \emph{Running Time:} $\tilde O(p^4 t^3 q^3 \cdot (n_1 n_2 + n_2 n_3 + n_1 n_3))$.
        \end{framed}
        
    \end{enumerate}
\end{enumerate}
This completes the description of the reduction. Clearly it is potential-adjusting (with potential functions $u[i] = -s_{i, g}$, $v[k] = 0$, and $w[j] = -t_{j, h}$).

\proofparagraph{Running Time}
The preprocessing step takes time $\tilde O(n_1 n_2 n_3 / t)$ as argued before and leads to $O(\log^2 n_2)$ instances we have to deal with. The running time of the uniformization step is bounded as follows by summing over all the contributions analyzed before:
\begin{equation*}
    \tilde O\parens*{(n_1 n_2 + n_1 n_3 + n_2 n_3) d \cdot p^4 t^3 q^3 + \frac{n_1 n_2 n_3}{t} + \frac{n_1 n_2 n_3 t}{p} + \frac{n_1 n_2 n_3 t p^4}{q}}.
\end{equation*}
This becomes as claimed by setting the parameters $p = t^2$ and $q = t^{10}$. From this choice of parameters it also follows that $|\mathcal I| \leq O(\log^2 (n_1 n_2 n_3) \cdot p^4 t^3 q^3) \leq O(t^{41} \log^2 (n_1 n_2 n_3))$ as claimed.
\end{proof}

\subsection{Reduction from Low-Rank to Uniform Regular Exact Triangle} \label{sec:exact-tri-low-rank:sec:low-rank-to-unif-regular}
The third step is to reduce to instances that are not only $D$-uniform but also \emph{$1/D$-regular} (i.e., each integer appears in at most a $1/D$-fraction of the entries of each row and column in all three matrices). This is one of the most technical steps in~\cite{AbboudFJVX25}. Roughly, the idea is to show that each $D$-uniform instance can be partitioned into two pieces: a $D$-uniform $1/D$-regular part, and a remaining part that is $D'$-slice-uniform for some~\makebox{$D' \ll D$}. The first part is as intended. And the second part can be solved by first applying \cref{lem:exact-tri-slice-unif-to-unif} and then \emph{recursing} on the polynomially many resulting $D'$-uniform instances. By carefully choosing the involved parameters this indeed leads to a subcubic reduction---but with quite bad quantitative bounds. Specifically, the resulting algorithm for $n^{3-\omega-\epsilon}$-slice-uniform Exact Triangle runs in time \smash{$n^{3-\epsilon^{O(1)}}$}.

Here we propose a twist to the argument in~\cite{AbboudFJVX25}: Instead of recursing on polynomially many instances that are $D'$-slice-uniform for some $D' \ll D$, it turns out to be possible to recurse on just $O(1)$ instances that are $r$-rank for some~\makebox{$r \ll D$}. This, in contrast, results in a final running time of $n^{3-\Omega(\epsilon)}$ as stated in \cref{thm:exact-tri-low-rank}.

\begin{lemma}[Reduction from Low-Rank to Uniform Regular Exact Triangle] \label{lem:exact-tri-low-rank-to-unif-reg}
For any parameter $t \geq 1$ there is a potential-adjusting reduction from any \emph{rank-$r$} Exact Triangle instance to $t^{875} (n_1 n_2 n_3)^{o(1)}$ instances that are each \emph{$D$-uniform} and \emph{$1/D$-regular} for some $D \leq r$. The reduction runs in deterministic time $(n_1 n_2 n_3 / t + (n_1 n_2 + n_2 n_3 + n_1 n_3) r t^{875})^{1+o(1)}$. In particular:
\begin{align*}
    &\ExactTri(n_1, n_2, n_3 \mid r) \\
    &\qquad \leq \parens*{t^{875} \cdot \ExactTri(n_1, n_2, n_3 \mid D \leq r; \rho \leq 1/D) + \frac{n_1 n_2 n_3}{t} + (n_1 n_2 + n_2 n_3 + n_1 n_3) r t^{875}}^{1+o(1)}.
\end{align*}
\end{lemma}
\begin{proof}
Let $q$ be a parameter to be determined later. We design a recursive algorithm that transforms the given low-rank instances in three steps into (1) slice-uniform instances, (2) uniform instances, (3) uniform and regular instances. The third step is incomplete in the sense that we cannot deal with some entries~$A_\sm, B_\sm, C_\sm$ in the matrices $A, B, C$. For these remaining entries, however, we can compute rank decompositions of strictly smaller rank $r' < r$, so in a fourth step we will simply recurse on these matrices. In detail:
\begin{enumerate}
    \item\label{lem:exact-tri-low-rank-to-unif-reg:itm:slice-unif} \emph{(Low-Rank to Slice-Uniform)} Run the potential-adjusting reduction from \cref{lem:exact-tri-low-rank-to-slice-unif} with some parameter $t_\slunif$ to transform the given low-rank instance into $r$-slice-uniform instances; let $(\mathcal I_\slunif, \mathcal T_\slunif)$ denote the output.

    \begin{framed}
        \emph{Running Time:} $\tilde O(n_1 n_2 n_3 / t_\slunif + (n_1 n_2 + n_2 n_3 + n_1 n_3) r t_\slunif^2)$.
    \end{framed}
    
    \item\label{lem:exact-tri-low-rank-to-unif-reg:itm:unif} \emph{(Slice-Uniform to Uniform)} We further reduce (in a potential-adjusting way) to $(\mathcal I_\unif, \mathcal T_\unif)$ as follows. Initialize $\mathcal I_\unif \gets \emptyset, \mathcal T_\unif \gets \mathcal T_\slunif$. For each instance $(A', B', C') \in \mathcal I_\slunif$ run \cref{lem:exact-tri-slice-unif-to-unif} with some parameter $t_\unif$ to obtain an $r$-uniform instances $(\mathcal I', \mathcal T')$. Let $\mathcal I_\unif \gets \mathcal I_\unif \cup \mathcal I'$ and $\mathcal T_\unif \gets \mathcal T_\unif \cup \mathcal T'$. This is the composition of two potential-adjusting reductions, and thus also potential-adjusting.

    \begin{framed}
        \emph{Running Time:} $\tilde O(|\mathcal I_\slunif| \cdot (n_1 n_2 n_3 / t_\unif + (n_1 n_2 + n_2 n_3 + n_1 n_3) r t_\unif^2))$.
    \end{framed}
    
    \item\label{lem:exact-tri-low-rank-to-unif-reg:itm:unif-reg} \emph{(Uniform to Uniform Regular)} Next, we reduce the instances in $\mathcal I_\unif$ further to regular instances~$\mathcal I_\reg$ (where initially $\mathcal I_\reg \gets \emptyset$). In this reduction step we will miss some exact triangles, but we will make sure that each missed triangle is part of one of the three instances $(A_\sm, B, C), (A, B_\sm, C), (A, B, C_\sm)$. These three instances are required to be low-rank, hence we will also maintain low-rank decompositions of $A_\sm, B_\sm, C_\sm$. Initially, $A_\sm, B_\sm, C_\sm$ are all-$\bot$ matrices and we start with trivial rank-$0$ decompositions.

    Enumerate all instances $(A', B', C') \in \mathcal I_\unif$. Since $(A', B', C')$ is a potential adjustment of the initial instance, there are potential functions $u \in \Int^{n_1}, v \in \Int^{n_2}$ such that
    \begin{equation*}
        A'[i, k] = A[i, k] + u[i] + v[k] \;\;\; \text{or} \;\;\; A'[i, k] = \bot
    \end{equation*}
    for all entries. Let $q$ be a parameter. We call an entry $A'[i, k]$ \emph{row-heavy} if the same entry appears in more than a $q / r$-fraction of the entries in its row (i.e.,~\makebox{$|\set{ k' : A'[i, k] = A'[i, k']}| > q / r \cdot n_2$}), otherwise we call the entry \emph{column-heavy} if it appears in more than a $q/r$-fraction of the entries in its column (i.e.,~\makebox{$|\set{ i' : A'[i, k] = A'[i', k]}| > q / r \cdot n_1$}), and otherwise we call the entry \emph{light}. Let
    \begin{equation*}
        A' = A'_\row \sqcup A'_\col \sqcup A'_\light
    \end{equation*}
    denote the partition of $A'$ into the row-heavy, column-heavy and light entries, respectively. In the same way we obtain partitions of $B$ and $C$ into their row-heavy, column-heavy and light parts, respectively. We proceed in two substeps:
    \begin{enumerate}
        \item\label{lem:exact-tri-low-rank-to-unif-reg:itm:unif-reg:itm:light} \emph{(Light Entries)} By definition each entry in \smash{$A'_\light$} appears at most in a $q / r$-fraction of its row and column. By a simple transformation we can further partition $A'_\light$ into $q^2$ matrices where each entry appears in at most a $1/r$-fraction of its row and column. Hence, we can transform the subinstance $(A'_\light, B'_\light, C'_\light)$ (in a potential-adjusting way) to $q^6$ instances that are $r$-uniform and $1/r$-regular. Insert these instances into~$\mathcal I_\reg$.

        \begin{framed}
            \emph{Running Time:} $\tilde O(|\mathcal I_\unif| \cdot q^6 \cdot (n_1 n_2 + n_2 n_3 + n_1 n_3))$.
        \end{framed}
        
        \item\label{lem:exact-tri-low-rank-to-unif-reg:itm:unif-reg:itm:heavy} \emph{(Heavy Entries)} Focus on the row-heavy entries in $A'$ for now. For each such row-heavy entry~$(i, k)$ we will update $A_\sm[i, k] \gets A[i, k]$. Recall that we need to update the respective rank decomposition of~$A_\sm$ to match these new entries. To this end observe that each row in $A'_\row$ contains at most~$r / q$ distinct entries, so we can trivially compute a rank-$(r/q)$ decomposition $(U', V', S')$ of~$A'_\row$. Next we adjust this decomposition by
        \begin{align*}
            U[i, \ell] &= U'[i, \ell] - u[i], \\
            V[\ell, j] &= V'[\ell, j] - v[j],
        \end{align*}
        leading to a same-rank decomposition $(U, V, S = S')$ of the restriction of $A$ to the row-heavy entries (of $A'$). Indeed, for each heavy entry $(i, k)$ we have that
        \begin{align*}
            A[i, k]
            &= A'[i, k] - u[i] - v[k] \\
            &= U'[i, S'[i, k]] + V'[S'[i, k], k] - u[i] - v[k] \\
            &= U[i, S[i, k]] + V[S[i, k], k].
        \end{align*}
        We repeat the same procedure with the column-heavy entries $A'_\col$ to update $A_\sm$. Then, symmetrically, we update $B_\sm$ and $C_\sm$ for all row-heavy and column-heavy entries in $B'$ and $C'$. In the upcoming step~\ref{lem:exact-tri-low-rank-to-unif-reg:itm:rec} we will see why this treatment was useful. 

        \begin{framed}
            \emph{Running Time:} $\tilde O(|\mathcal I_\unif| \cdot (n_1 n_2 + n_2 n_3 + n_1 n_3))$.
        \end{framed}
    \end{enumerate}
    After we have enumerated all instances in $\mathcal I_\unif$ in this way, we claim that we have transformed the initial instance in a potential-adjusting manner into
    \begin{equation*}
        (\mathcal I_\reg \cup \set{(A_\sm, B, C), (A, B_\sm, C), (A, B, C_\sm)}, \mathcal T_\unif).
    \end{equation*}
    To see that this is correct take any exact triangle $(i, k, j)$ in $(A, B, C)$. After step~\ref{lem:exact-tri-low-rank-to-unif-reg:itm:unif}, the triangle is part of $\mathcal T_\unif$ or it appears as an exact triangle in some instance $(A', B', C') \in \mathcal I_\unif$. In the former case we are done, so focus on the latter case. We continue with two subcases: If the entries $(i, k), (k, j), (i, j)$ are all light in $A', B', C'$, respectively, then by construction $(i, k, j)$ appears as an exact triangle in some instance in~$\mathcal I_\reg$. Otherwise, if $(i, k)$ is row-heavy, say, then $(i, k, j)$ appears as an exact triangle in $(A_\sm, B, C)$.

    \item\label{lem:exact-tri-low-rank-to-unif-reg:itm:rec} \emph{(Recursion)} We recurse on the three instances $(A_\sm, B, C)$, $(A, B_\sm, C)$ and $(A, B, C_\sm)$. Recall that we have indeed constructed low-rank decompositions for the three matrices $A_\sm, B_\sm, C_\sm$ as required by the theorem statement. We will soon analyze how small these rank decompositions actually are, and discuss why this recursion leads to a terminating algorithm. Let $(\mathcal I_A, \mathcal T_A), (\mathcal I_B, \mathcal T_B), (\mathcal I_C, \mathcal T_C)$ denote the outputs of the three recursive calls. We return $(\mathcal I, \mathcal T) = (\mathcal I_A \cup \mathcal I_B \cup \mathcal I_C \cup \mathcal I_\reg, \mathcal T_A \cup \mathcal T_B \cup \mathcal T_C \cup \mathcal T_\unif)$.
\end{enumerate}

This completes the description of the recursive step. In the base case, when $r = 0$, one of the matrices contains only $\bot$-entries and so we can trivially return $\mathcal I, \mathcal T \gets \emptyset$. It follows from the in-text explanation that this algorithm is indeed a potential-adjusting reduction (assuming it terminates), and that all resulting instances are regular. It remains to argue that we have constructed sets $(\mathcal I, \mathcal T)$ of the claimed sizes.

\proofparagraph{Analysis of the Recursion}
The critical part is to analyze the rank decompositions of $A_\sm, B_\sm, C_\sm$. Note that with each iteration of step~\ref{lem:exact-tri-low-rank-to-unif-reg:itm:unif-reg:itm:heavy} we increase their sizes by at most $2 r / q$ (namely, plus $r / q$ due to the row-heavy entries, and plus $r / q$ due to the column-heavy entries). Thus, letting $n = \max\set{n_1, n_2, n_3}$ and letting~$r'$ denote the final size of the rank decompositions, we have that
\begin{equation*}
    r' \leq \frac{2r}{q} \cdot |\mathcal I_\unif| \leq O\parens*{\frac{2r}{q} \cdot |\mathcal I_\slunif| \cdot t_\unif^{41} \cdot \log^2 n } \leq O\parens*{\frac{2r}{q} \cdot t_\slunif^2 \cdot t_\unif^{41} \cdot \log^3 n }.
\end{equation*}
Pick $q = t_\unif^{41} \cdot t_\slunif^2 \cdot 2^{\sqrt{\log n}}$, then it follows that
\begin{equation*}
    r' \leq \frac{r}{2^{\Omega(\sqrt{\log n)}}}.
\end{equation*}
First, this implies that each recursive call indeed makes progress and that the algorithm terminates. Second, this further implies that the recursion reaches depth at most $O(\sqrt{\log n})$, and as each call of the algorithms spawns only three recursive calls the total number of recursive calls is \smash{$3^{O(\sqrt{\log n})}$}.

We can use this to bound the final size of $(\mathcal I, \mathcal T)$. Focus on any recursive call of the algorithm. The sets we construct have size
\begin{equation*}
    |\mathcal T_\unif| \leq |\mathcal T_\slunif| + |\mathcal I_\slunif| \cdot \tilde O\parens*{\frac{n_1 n_2 n_3}{t_\unif}} \leq \tilde O\parens*{\frac{n_1 n_2 n_3}{t_\slunif} + \frac{n_1 n_2 n_3 t_\slunif^2}{t_\unif}}
\end{equation*}
and
\begin{equation*}
    |\mathcal I_\reg| \leq q^6 \cdot |\mathcal I_\unif| \leq q^6 \cdot t_\slunif^2 \cdot t_\unif^{41} \cdot \log^3 n \leq t_\slunif^{14} \cdot t_\unif^{287} \cdot 2^{O(\sqrt{\log n})}.
\end{equation*}
The total sizes $|\mathcal T|$ and $|\mathcal I|$ are worse only by a factor $3^{O(\sqrt{\log n})}$. Choosing $t_\slunif = t$ and $t_\unif = t^3$, this is as claimed in the theorem statement.

\proofparagraph{Running Time}
Finally, focus on the running time. Per recursive call the sum of all contributions analyzed before is, loosely bounded,
\begin{align*}
    &\tilde O\parens*{\frac{n_1 n_2 n_3}{t_\slunif} + \frac{n_1 n_2 n_3 t_\slunif^2}{t_\unif} + (n_1 n_2 + n_2 n_3 + n_1 n_3) r \cdot t_\unif^{41} \cdot t_\slunif^2 \cdot q^6} \\
    &\qquad\leq \parens*{\frac{n_1 n_2 n_3}{t} + (n_1 n_2 + n_2 n_3 + n_1 n_3) r t^{875}}^{1+o(1)}.
\end{align*}
Across all recursive calls the running time worsens only by $3^{O(\sqrt{\log n})} = n^{o(1)}$.
\end{proof}

\subsection{Reduction from Uniform Regular to Uniform Low-Doubling Exact Triangle} \label{sec:exact-tri-low-rank:sec:unif-reg-to-low-doubling}
The fourth and final step is to reduce to uniform (regular) low-doubling Exact Triangle. This step is again due to~\cite{AbboudFJVX25} and~\cite{ChanL15} almost without modifications---we merely extend the proof to rectangular matrices and show the slightly stronger condition that the joint set of entries of all three matrices~$A, B, C$ has small doubling. The reduction crucially relies on the regularity condition established before and also on the BSG theorem from additive combinatorics~\cite{BalogS94,Gowers01}. Specifically, we use the following algorithmic covering version of the BSG theorem due to Chan and Lewenstein~\cite{ChanL15}.

\begin{theorem}[BSG Covering~\cite{ChanL15}]
Let $X, Y, Z \subseteq \Int$ be sets of size at most $n$, and let $L \geq 1$. Then there are subsets $X_1, \dots, X_L \subseteq X$ and $Y_1, \dots, Y_L \subseteq Y$ satisfying the following properties:
\begin{enumerate}
    \item $|X_\ell + Y_\ell| \leq O(L^5 n)$ for all $\ell \in [L]$,
    \item $|R| \leq O(n^2 / L)$ where $R = \set{(x, y) \in X \times Y : x + y \in Z} \setminus \bigcup_{\ell=1}^L (X_\ell \times Y_\ell)$.
\end{enumerate}
The sets $X_1, \dots, X_L, Y_1, \dots, Y_L$ and $R$ can be computed in deterministic time $\tilde O(n^\omega L)$.
\end{theorem}

The following simple lemma will also be useful.

\begin{lemma} \label{lem:sumset-doubling}
Let $X, Y \subseteq \Int$ be sets with doubling $|X + X| \leq K |X|$ and $|Y + Y| \leq K |Y|$. Then $Z = X + Y$ has doubling $|Z + Z| \leq K^8 |Z|$.
\end{lemma}
\begin{proof}
We write $nX = X + \dots + X$ for the $n$-fold iterated sumset of $X$. The proof is by repeated applications of Ruzsa's triangle inequality (\cref{lem:ruzsa-triangle}). The first application is with $X' = 2X$, $Y' = 2Y$ and $Z' = Z$ to conclude that
\begin{equation*}
    |Z + Z| = |2X + 2Y| \leq \frac{|2X + Z| \, |2Y + Z|}{|Z|} = \frac{|3X + Y| \, |X + 3Y|}{|Z|}.
\end{equation*}
The second application is with $X' = 3X$, $Y' = Y$ and $Z' = X$ to obtain
\begin{equation*}
    |3X + Y| \leq \frac{|3X + X| \, |X + Y|}{|X|} = \frac{|4X| \, |Z|}{|X|} \leq \frac{K^4 |X| \, |Z|}{|X|} = K^4 |Z|;
\end{equation*}
here we also used the Plünnecke-Ruzsa inequality (\cref{lem:pluennecke-ruzsa}) $|n X| \leq K^n |X|$. The third application is symmetric:
\begin{equation*}
    |X + 3Y| \leq \frac{|X + Y| \, |3Y + Y|}{|Y|} = \frac{|Z| \, |4Y|}{|Y|} \leq \frac{|Z| \, K^4 |Y|}{|Y|} = K^4 |Z|.
\end{equation*}
The statement follows from putting these three bounds together.
\end{proof}

\begin{lemma}[Reduction from Uniform Regular to Uniform Low-Doubling Exact Triangle] \label{lem:exact-tri-unif-reg-to-low-doubling}
For any parameter $K \geq 1$ there is a potential-adjusting reduction from any \emph{$D$-uniform $1/D$-regular} Exact Triangle instance to $O(K^{1/112})$ instances that are each \emph{$D$-uniform}, \emph{$1/D$-regular} and \emph{$K$-doubling}. The reduction runs in deterministic time $\tilde O(n_1 n_2 n_3 / K^{1/112} + (n_1 n_2 + n_2 n_3 + n_1 n_3) D K^{1/112})$. In particular:
\begin{align*}
    &\ExactTri(n_1, n_2, n_3 \mid D; \rho \leq 1/D) \\
    &\qquad \leq \tilde O\parens*{K^{1/112} \cdot \ExactTri(n_1, n_2, n_3 \mid D; \rho \leq 1/D; K) + \frac{n_1 n_2 n_3}{K^{1/112}} + (n_1 n_2 + n_2 n_3 + n_1 n_3) D K^{1/112}}.
\end{align*}
\end{lemma}
\begin{proof}
Let $X$ denote the set of entries in the given $D$-uniform instance, i.e., $|X| \leq D$. We apply BSG Covering with $X = Y = Z$ with some parameter $L$ to be picked later, resulting in sets $X_1, \dots, X_L, Y_1, \dots, Y_L$ and $R$. We say that a triangle $(i, k, j)$ is in the \emph{structured} part if $(A[i, k], B[k, j]) \in X_\ell \times Y_\ell$ (for some~\makebox{$\ell \in [L]$}), and otherwise, if $(A[i, k], B[k, j]) \in R$, we say that it is in the \emph{remainder}. We deal with these two types of triangles in two separate steps:
\begin{enumerate}
    \item\label{lem:exact-tri-unif-reg-to-low-doubling:itm:rem} \emph{(Remainder)} We will explicitly list all exact triangles in the remainder. Enumerate all pairs $(x, y) \in R$, enumerate all $k \in [n_2]$ and enumerate all $i \in [n_1]$ with $A[i, k] = x$ and all $j \in [n_3]$ with $B[k, j] = y$. All exact triangles $(i, k, j)$ that we detect in this way are stored in a set $\mathcal T$.

    \begin{framed}
        \emph{Running Time:} Recall that the given instance is $1/D$-regular, hence for all fixed $k, x$ there are at most $n_1 / D$ indices $i$ with $A[i, k] = x$, and similarly for $B$. It follows that the running time is:
        \begin{equation*}
            O(|R| \cdot n_2 \cdot (n_1 / D) \cdot (n_3 / D)) = O(n_1 n_2 n_3 / L).
        \end{equation*}
    \end{framed}

    \item\label{lem:exact-tri-unif-reg-to-low-doubling:itm:struct} \emph{(Structured Part)} Enumerate all $\ell \in [L]$. We compute the set $X_\ell + Y_\ell$ (by brute-force, say) and construct the matrices
    \begin{align*}
        A_\ell[i, k] &=
        \begin{cases}
            \mathrlap{A[i, k]}\phantom{B[k, j]} &\text{if $A[i, k] \in X_\ell$,} \\
            \bot &\text{otherwise,}
        \end{cases} \\
        B_\ell[k, j] &=
        \begin{cases}
            B[k, j] &\text{if $B[k, j] \in Y_\ell$,} \\
            \bot &\text{otherwise,}
        \end{cases} \\
        C_\ell[i, j] &=
        \begin{cases}
            \mathrlap{C[i, j]}\phantom{B[k, j]} &\text{if $C[i, j] \in X_\ell + Y_\ell$,} \\
            \bot &\text{otherwise.}
        \end{cases}
    \end{align*}
    Clearly each exact triangle in the structured part appears in some instance $(A_\ell, B_\ell, C_\ell)$. To proceed, we distinguish two cases:
    \begin{enumerate}
        \item\label{lem:exact-tri-unif-reg-to-low-doubling:itm:struct:itm:sparse} \emph{(Sparse Case: $|X_\ell| \leq D / L^2$ or $|Y_\ell| \leq D / L^2$)} In this case one of the matrices $A_\ell$ or $B_\ell$ is sparse. Indeed, suppose that $|X_\ell| \leq D / L^2$. Then $A_\ell$ contains at most $n_1 n_2 / L^2$ non-$\bot$ entries, as it contains at most $|X_\ell|$ distinct entries each of which appears in at most a $1/D$-fraction of the entries (by the regularity assumption). We thus solve the instance $(A_\ell, B_\ell, C_\ell)$ by brute-force and place all exact triangles in $\mathcal T$.

        \begin{framed}
            \emph{Running Time:} $O(n_1 n_2 n_3 / L^2)$ per iteration, so time $O(n_1 n_2 n_3 / L)$ in total.
        \end{framed}
        
        \item\label{lem:exact-tri-unif-reg-to-low-doubling:itm:struct:itm:dense} \emph{(Dense Case: $|X_\ell| \geq D / L^2$ and $|Y_\ell| \geq D / L^2$)} Let $Z_\ell = X_\ell \cup Y_\ell \cup (X_\ell + Y_\ell)$ and observe that all entries in the three matrices $A_\ell, B_\ell, C_\ell$ stem from the set~$Z_\ell$. In the remainder we bound the doubling of~$Z_\ell$. The doubling of $X_\ell$ can be bounded by one application of Ruzsa's triangle inequality (\cref{lem:ruzsa-triangle}):
        \begin{equation*}
            |X_\ell + X_\ell| \leq \frac{|X_\ell + Y_\ell| \, |X_\ell + Y_\ell|}{|Y_\ell|} \leq \frac{O(L^{10} D^2)}{D / L^2} = O(L^{12} D) = O(L^{14}) |X_\ell|,
        \end{equation*}
        It follows that $X_\ell \cup \set{0}$ has doubling $O(L^{14})$. Symmetrically, $Y_\ell \cup \set{0}$ has doubling $O(L^{14})$. But $Z_\ell = (X_\ell \cup \set{0}) + (Y_\ell \cup \set{0})$, and thus $Z_\ell$ has doubling at most $O(L^{14})^8 = O(L^{112})$ by \cref{lem:sumset-doubling}. (Again, we have not attempted to optimize this constant.) We pick $L = \Theta(K^{1/112})$ so that $Z_\ell$ has doubling at most $K$. The Exact Triangle instance $(A_\ell, B_\ell, C_\ell)$ is therefore a $K$\=/doubling (and $D$-uniform and $1/D$-regular) instance as desired, so we place it in $\mathcal I$.

        \begin{framed}
            \emph{Running Time:} $O(L \cdot (n_1 n_2 + n_2 n_3 + n_1 n_3))$ in total.
        \end{framed}
    \end{enumerate}
\end{enumerate}
It is clear that this is a potential-adjusting reduction with $|\mathcal I| \leq L \leq O(K)$ (crudely bounded). It remains to analyze the running time. The initial BSG Covering takes time $\tilde O(D^\omega L) = \tilde O(D^3 L) = \tilde O(n_1 n_2 D L)$ (crudely bounding~\makebox{$\omega \leq 3$} and~\makebox{$D \leq n_1, n_2$}), so summing over the running time of steps~\ref{lem:exact-tri-unif-reg-to-low-doubling:itm:rem} and~\ref{lem:exact-tri-unif-reg-to-low-doubling:itm:struct} as analyzed before the total running time is
\begin{equation*}
    \tilde O\parens*{\frac{n_1 n_2 n_3}{L} + (n_1 n_2 + n_2 n_3 + n_1 n_3) D L} = \tilde O\parens*{\frac{n_1 n_2 n_3}{K^{1/112}} + (n_1 n_2 + n_2 n_3 + n_1 n_3) D K^{1/112}},
\end{equation*}
as claimed.
\end{proof}

\subsection{Putting the Pieces Together} \label{sec:exact-tri-low-rank:sec:low-rank-to-low-doubling}
We finally put \cref{lem:exact-tri-low-rank-to-unif-reg,lem:exact-tri-unif-reg-to-low-doubling} together to complete the proof of \cref{lem:exact-tri-low-rank-to-low-doubling}.

\lemexacttrilowranktolowdoubling*
\begin{proof}
First run \cref{lem:exact-tri-low-rank-to-unif-reg} on the given low-rank Exact Triangle instance with parameter $t = K^{1/98000}$. Then transform each of the $t^{875} (n_1 n_2 n_3)^{o(1)}$ resulting uniform regular Exact Triangle instances by running \cref{lem:exact-tri-unif-reg-to-low-doubling} with parameter~$K$. This results in a set of $t^{875} (n_1 n_2 n_3)^{o(1)} \cdot K^{1/112} \leq K (n_1 n_2 n_3)^{o(1)}$ instances that are uniform, regular and low-doubling as claimed. The total running time is
\begin{align*}
    &\parens*{\frac{n_1 n_2 n_3}{t} + (n_1 n_2 + n_2 n_3 + n_1 n_3) r t^{875} + t^{875} \cdot \parens*{\frac{n_1 n_2 n_3}{K^{1/112}} + (n_1 n_2 + n_2 n_3 + n_1 n_3) r K^{1/112}}}^{1+o(1)} \\
    &\qquad= \parens*{\frac{n_1 n_2 n_3}{K^{1/98000}} + (n_1 n_2 + n_2 n_3 + n_1 n_3) r K}^{1+o(1)}.
\end{align*}
This completes the reduction and thus the proof of \cref{thm:exact-tri-low-rank}.
\end{proof}

\section{Universe Reductions for Min-Plus Product} \label{sec:univ-reduct}
In this section, we develop the universe reductions for Min-Plus Product (and thereby for APSP). The structure of this section is as follows: In \cref{sec:univ-reduct:sec:listing} we recap a standard trick from fine-grained complexity. In \cref{sec:univ-reduct:sec:small} we give the universe reduction from small universes to even smaller universes (say from $u = n$ to $u = \sqrt n$), proving \cref{thm:strong-apsp-implies-dir-unw-apsp}. In \cref{sec:univ-reduct:sec:doubling} we give the reduction from Min-Plus Product to Uniform Low-Doubling Min-Plus Product. And in \cref{sec:univ-reduct:sec:hashing} we conclude the conditional universe reduction from polynomially large to small universes, proving \cref{thm:apsp-implies-strong-apsp}.

\subsection{Listing-to-Decision Reduction} \label{sec:univ-reduct:sec:listing}
Before we start, we recap a standard \emph{listing-to-decision reduction} that allows to transform any detection algorithm into an algorithm that also lists some ``witnesses'' for a quite general class of problems. In the context of Min-Plus Product, we call $k$ a \emph{witness} of $(A * B)[i, j]$ if $A[i, k] + B[k, j] = (A * B)[i, j]$. Then:

\begin{lemma}[Listing-to-Decision Reduction for Min-Plus Product] \label{lem:min-plus-wit}
Let $(A, B)$ be a Min-Plus Product instance, and assume that there is an algorithm computing the min-plus product $A' * B'$ of any restrictions of $A$ and $B$ in time $T$. Then, for any parameter $t \geq 1$, we can compute lists of up to $t$ witnesses for all $(A * B)[i, j]$ in randomized time $\tilde O(T t)$.
\end{lemma}
\begin{proof}[Proof Sketch]
The proof is a by-now standard argument from sparse recovery; we only give a quick sketch here. Consider the restrictions $A'$ and $B'$ randomly obtained from $A$ and $B$ as follows. Sample $\ell \in [\ceil{\log n_2}]$ uniformly at random, then sample $K \subseteq [n_2]$ uniformly at random with rate $2^{-\ell}$. Let $A'$ be the matrix obtained from $A$ by setting all columns $k$ that are not present in $K$ to $\bot$. It can be shown that, for any fixed pair $(i, j)$, with probability at least $\Omega(1 / \log n_2)$ in the min-plus product $A' * B$ \emph{exactly one} witness of~$(i, j)$ is still present. In this case we say that~$(i, j)$ is \emph{successful.} Moreover, by symmetry this unique witness is sampled uniformly at random.

We now argue that by $O(\log n_2)$ calls to the fast algorithm we can identify the surviving witness for all successful pairs $(i, j)$. To this end, identify each entry $k \in [n_2]$ with its $L = \ceil{\log n_2}$-bit binary representation. We write $k[\ell]$ for the $\ell$-th bit. Let $A_\ell'$ be the matrix obtained from $A'$ where we replace the $k$-th column by~$\bot$ if $k[\ell] = 0$. We compute the min-plus products $A'_\ell * B$ for all $\ell \in [L]$. Now take any successful pair~$(i, j)$, and inspect the sequence $(A_0' * B)[i, j], \dots, (A_L' * B)[i, j]$. Some of these values are exactly equal to $(A * B)[i, j]$, others are strictly larger---the length-$L$ bit-string indicating the larger values is exactly the bit-representation of the unique witness $k$ of $(i, j)$. In particular, for all successful pairs $(i, j)$ we can read off (and verify) the surviving witness $k$.

Overall, by repeating this procedure $\tilde O(t)$ times, we expect to list $\tilde O(t)$ uniformly random witnesses for each entry $(A * B)[i, j]$. Thereby we either list $t$ distinct witnesses (if they exist), or otherwise have listed all witnesses with good probability. 
\end{proof}

With exactly the same proof one can also list witnesses for the Exact Triangle problem, see the following \cref{lem:exact-tri-wit}. In this context we call $k$ a \emph{witness} of $C[i, j]$ if $A[i, k] + B[k, j] = C[i, j]$. Similarly, in this case~$j$ is a \emph{witness} of $A[i, k]$ and~$i$ is a witness of $B[k, j]$.

\begin{lemma}[Listing-to-Decision Reduction for Exact Triangle] \label{lem:exact-tri-wit}
Let $(A, B, C)$ be an Exact Triangle instance, and assume that there is an algorithm solving $(A', B', C')$ for any restrictions~$A'$ of $A$,~$B'$ of $B$, and $C'$ of~$C$ in time $T$. Then, for any parameter $t \geq 1$, we can compute lists of up to $t$ witnesses for all outputs $A[i, k], B[k, j], C[i, j]$ in randomized time $\tilde O(T t)$.
\end{lemma}

\subsection{Small-Universe Reduction} \label{sec:univ-reduct:sec:small}
We start with the reduction from small universes $\set{0, \dots, n}$ to even smaller universes $\set{0, \dots, n^\alpha}$ (for any constant $\alpha \leq 1$). The proof of this universe reduction can be split into two parts: dealing with the \emph{unpopular} and \emph{popular} outputs. The first part is as in~\cite{ChanVX23} without modifications. The innovation is to solve the second part with ``low-rank APSP'' instead of ``Fredman's trick meets dominance product''.

\begin{lemma}[Small-Universe Reduction] \label{lem:sm-univ-reduct}
For all parameters $n_1, n_2, n_2', n_3, u, t$ with $n_2' \leq n_2$, it holds that
\begin{equation*}
    \MinPlus(n_1, n_2, n_3 \mid u) = \tilde O\parens*{\MinPlus(n_1, n_2', n_3 \mid u \leq n_2') \cdot \frac{n_2 t}{n_2'} + \ExactTri(n_1, n_2, n_3 \mid r \leq u / t)}.
\end{equation*}
\end{lemma}
\begin{proof}
Let $A, B$ be the two given matrices. The first step is to apply a scaling trick to obtain an additive approximation $\tilde C$ of the output matrix $A * B$. To this end, let $A'$ be the matrix obtained from $A$ by dropping the least significant bits of all entries (i.e., $A'[i, j] = \floor{A[i, j] / 2}$), and similarly obtain $B'$ from $B$. We \emph{recursively} compute the min-plus product $C' = A' * B'$ and pick $\tilde C = 2 C'$. It is easy to verify that for all pairs $(i, j)$:
\begin{equation*}
    \tilde C[i, j] \leq (A * B)[i, j] \leq \tilde C[i, j] + 2.
\end{equation*}

Let $p = \ceil{n_2'}$ and $q = \ceil{u / n_2'}$. Without loss of generality, we can assume that $A[i, k] \bmod q < q/2$ and $B[k, j] \bmod q < q/2$ by the following simple trick. Partition $A$ into two submatrices, one containing the entries with $A[i, k] \bmod q < q/2$ and the other containing all remaining entries. Subtract $\ceil{q/2}$ from all entries of the second matrix; then both submatrices are as claimed. Similarly partition $B$ into two parts. It remains to compute four min-plus products from which we can easily recover the original product $A * B$ (by appropriately adding $0$, $\ceil{q/2}$, or $2\ceil{q/2}$ to the respective output matrices).

Now view the matrices~$A$ and $B$ in their ``base-$(p, q)$'' representation, i.e., we rewrite the matrices as
\begin{align*}
    A &= A_1 \cdot q + A_0, \\
    B &= B_1 \cdot q + B_0,
\end{align*}
where $A_0, B_0$ have nonnegative entries less than $q / 2$ and where $A_1, B_1$ have nonnegative entries less than~$p$. We call $k$ a \emph{witness} of $(i, j)$ if $A[i, k] + B[k, j] = (A * B)[i, j]$, and we call $k$ a \emph{pseudo-witness} of $(i, j)$ if~\makebox{$A_1[i, k] + B_1[k, j] = (A_1 * B_1)[i, j]$}. Importantly, note that each witness is also a pseudo-witness (as we have that $\floor{(A[i, k] + B[i, k]) / q} = \floor{A[i, k]/ q} + \floor{B[i, k] / q}$, using that $A_0, B_0$ have entries less than~$q/2$).

We call $(i, j)$ \emph{popular} if it has at least $t n_2 / n_2'$ pseudo-witnesses, and \emph{unpopular} otherwise. In the following two cases we separately deal with the unpopular and popular pairs $(i, j)$. Initially, let $C$ be the all-$\bot$ matrix.

\begin{enumerate}
    \item\label{lem:sm-univ-reduct:itm:unpop} \emph{(Unpopular Pairs)} To deal with the unpopular pairs $(i, j)$, we will list \emph{all} pseudo-witnesses $k$ of $(i, j)$. To this end, repeat the following steps $O(t n_2 / n_2' \cdot \log (n_1 n_3))$ times:
    \begin{enumerate}
        \item Sample a uniformly random subset $\mathcal K \subseteq [n_2]$ with rate $\rho = n_2' / (2 t n_2)$. The resulting set has size $|\mathcal K| \leq O(n_2' / t) \leq O(n_2')$ with high probability.
        \item Restrict $A_1$ to the columns in $\mathcal K$, restrict $B_1$ to the rows in $\mathcal K$, and compute the min-plus product of the resulting matrices. Additionally, report a witness $k$ for each output pair $(i, j)$ by \cref{lem:min-plus-wit}. (Note that a witness for $A_1 * B_1$ is a pseudo-witness in our language.) We update $C[i, j] \gets \min\set{C[i, j], A[i, k] + B[k, j]}$ for each triple $(i, k, j)$ reported in this way.
    \end{enumerate}
    The correctness claim is that after step~\ref{lem:sm-univ-reduct:itm:unpop} is complete we have correctly assigned $C[i, j] = (A * B)[i, j]$ for all unpopular pairs $(i, j)$. Fix any such pair, and let $k$ be a witness of $(i, j)$. In each iteration the probability that we include $k$ into $\mathcal K$ is $\rho$, and independently the probability that we miss all other pseudo-witnesses of $(i, j)$ in $\mathcal K$ is at least $1 - \rho t n_2 / n_2' \geq 1/2$, by a union bound over the at most $t n_2 / n_2'$ pseudo-witnesses of $(i, j)$. Hence, with probability at least $\rho / 2$ the only pseudo-witness of $(i, j)$ in $K$ is a proper witness, and so the algorithm is forced to list the triple $(i, k, j)$. Over the $O(\rho^{-1} \log (n_1 n_3))$ iterations this event happens at least once for each unpopular pair $(i, j)$ with high probability.

    \begin{framed}
        \emph{Running Time:} $\tilde O(\MinPlus(n_1, n_2', n_3 \mid u \leq p) \cdot t n_2 / n_2')$.
    \end{framed}

    \item\label{lem:sm-univ-reduct:itm:pop} \emph{(Popular Pairs)} We deal with the popular pairs by solving $O(1)$ low-rank Exact Triangle instances. Specifically, repeat the following steps for all offsets $z \in \set{0, 1, 2}$. Recall from the first paragraph that we have access to a matrix $\tilde C$ that approximates $A * B$ with small additive error, so in the $z$-th iteration our goal is to correctly assign $C[i, j] = (A * B)[i, j]$ for all entries popular pairs $(i, j)$ with $(A * B)[i, j] = \tilde C[i, j] + z$.
    \begin{enumerate}
        \item\label{lem:sm-univ-reduct:itm:pop:itm:sample} Sample a uniformly random subset $\mathcal K \subseteq [n_2]$ with rate $\Theta(n_2' / (t n_2) \log (n_1 n_3))$. The resulting set has size $|\mathcal K| = \tilde O(n_2' / t)$ with high probability. Moreover, with high probability $\mathcal K$ contains at least one pseudo-witness for each popular pair $(i, j)$.
        \item\label{lem:sm-univ-reduct:itm:pop:itm:decomp-1} Let $U_1$ be the matrix $A_1$ restricted to the columns in $\mathcal K$, let $V_1$ be the matrix $B_1$ restricted to the rows in $\mathcal K$ and compute their min-plus product $R_1 = U_1 * V_1$. Additionally, report a witness~$k$ for each output entry $(i, j)$ by \cref{lem:min-plus-wit}. Let $S_1$ be the matrix that stores these witnesses in the form $S_1[i, j] = k$. Then observe that $(U_1, V_1, S_1)$ is a valid rank-$r_1$ decomposition of $R_1$ for $r_1 = |\mathcal K| = \tilde O(n_2' / t)$.
        
        \begin{framed}
            \emph{Running Time:} $\tilde O(\MinPlus(n_1, n_2', n_3 \mid u \leq p))$.
        \end{framed}
        
        \item\label{lem:sm-univ-reduct:itm:pop:itm:decomp} Let $R$ be the $n_1 \times n_3$ matrix defined by
        \begin{equation*}
            R[i, j] =
            \begin{cases}
                \tilde C[i, j] + z &\text{if $\floor{(\tilde C[i, j] + z)/ q} = R_1[i, j]$,} \\
                \bot &\text{otherwise.}
            \end{cases}
        \end{equation*}
        We can now find a rank-$r$ decomposition of $R$ for $r = r_1 q$. As outlined before, this step involves trivially bounding the select-plus rank of $R - q \cdot R_1$ by $O(q)$ (\cref{fac:rank-triv-univ}) and then applying submultiplicativity (\cref{fac:rank-submult}).
        
        Concretely, we compute the rank decomposition $(U, V, S)$ of $R$ from the data we have precomputed in step~\ref{lem:sm-univ-reduct:itm:pop:itm:decomp-1} as follows. Let \smash{$U \in \Int^{[n_1] \times (\mathcal K \times [q])}$} be the matrix obtained from \smash{$q \cdot U_1 \in \Int^{[n_1] \times \mathcal K}$} by replacing each column by $q$ copies, where in the $\ell$-th copy we add $\ell$ to all entries. Next, we obtain~\smash{$V \in \Int^{(\mathcal K \times [q]) \times [n_3]}$} from \smash{$q \cdot V_1 \in \Int^{\mathcal K \times [n_3]}$} by duplicating each row $q$ times (but without further additions). Finally, define $S$ by
        \begin{equation*}
            S[i, j] =
            \begin{cases}
                (S_1[i, j], R[i, j] \bmod q) &\text{if $R[i, j] \neq \bot$,} \\
                \bot &\text{if $R[i, j] = \bot$.}
            \end{cases}
        \end{equation*}
        Then, indeed, $U$ has $|\mathcal K| q = r$ columns (and $V$ has $|\mathcal K| q = r$ rows), and by construction we have that
        \begin{align*}
            R[i, j]
            &= \tilde C[i, j] + z \\
            &= \smash{\floor{\tilde C[i, j] / q} \cdot q + ((\tilde C[i, j] + z) \bmod q)} \\
            &= R_1[i, j] \cdot q + (R[i, j] \bmod q) \\
            &= U_1[i, S_1[i, j]] \cdot q + (R[i, j] \bmod q) + V_1[S_1[i, j], j] \cdot q \\
            &= U[i, S[i, j]] + V[S[i, j], j],
        \end{align*}
        for all non-$\bot$ entries. Hence, $(U, V, S)$ is a valid rank-$r$ decomposition of $R$.
        
        \begin{framed}
            \emph{Running Time:} $\tilde O(n_1 n_2 + n_1 n_3 + n_2 n_3)$ (negligible).
        \end{framed}
        
        \item\label{lem:sm-univ-reduct:itm:pop:itm:solve} Solve the $r$-Rank Exact Triangle instance $(A, B, R)$, and additionally report a witness $k$ for each output entry $(i, j)$ by \cref{lem:exact-tri-wit}. Update $C[i, j] \gets \min\set{C[i, j], A[i, k] + B[k, j]}$ for each such triple.

        \begin{framed}
            \emph{Running Time:} $\tilde O(\ExactTri(n_1, n_2, n_3 \mid r))$.
        \end{framed}
    \end{enumerate}
    We argue that these steps correctly assign all popular entries $C[i, j]$. In steps~\ref{lem:sm-univ-reduct:itm:pop:itm:sample} and~\ref{lem:sm-univ-reduct:itm:pop:itm:decomp-1} we construct a matrix $R_1$ satisfying that $(A_1 * B_1)[i, j] = R_1[i, j]$ for all popular pairs (since, as noted before, $\mathcal K$ contains a pseudo-witness with high probability). Let $z \in \set{0, 1, 2}$ be such that $(A * B)[i, j] = \tilde C[i, j] + z$. Recalling $\floor{(A * B)[i, j] / q} = (A_1 * B_1)[i, j]$, it follows that~\makebox{$\floor{(\tilde C[i, j] + z) / q} = R_1[i, j]$}. We therefore set $R[i, j] = \tilde C[i, j] + z = (A * B)[i, j]$, and thus $(i, j)$ is involved in an exact triangle in the instance $(A, B, R)$. It follows that we list a witness $k$ of $(i, j)$ and correctly assign $C[i, j]$ in step~\ref{lem:sm-univ-reduct:itm:pop:itm:solve}.
\end{enumerate}
This completes the algorithm and the correctness proof, and we finally analyze the running time. Summing over the contributions of all steps (analyzed in the boxes), the total time indeed becomes
\begin{equation*}
    \tilde O\parens*{\MinPlus(n_1, n_2', n_3 \mid u \leq n_2') \cdot \frac{n_2 t}{n_2'} + \ExactTri(n_1, n_2, n_3 \mid r \leq u / t)},
\end{equation*}
where we have bounded $p = O(n_2')$ and $r = r_1 q = \tilde O(n_2' / t \cdot u / n_2') = \tilde O(u / t)$. (Moreover, here we have omitted the negligible overhead of $\tilde O(n_1 n_2 + n_1 n_3 + n_2 n_3)$ for the various bookkeeping steps as this time is necessarily dominated by the $\ExactTri(n_1, n_2, n_3)$ call simply to read the input.) Finally, each recursive call reduces $u$ by a constant factor. Thus, the recursion reaches depth at most $O(\log u) = O(\log (n_1 n_2 n_3))$ and we incur an overhead that is hidden in the $\tilde O$ notation. 
\end{proof}

\begin{corollary} \label{cor:min-plus-rect-sm-univ}
Let $0 < \alpha \leq 1$ and $\epsilon > 0$ be constants. If $\MinPlus(n, n^{\alpha}, n \mid u \leq n^\alpha) = O(n^{2+\alpha-\epsilon})$ then the Strong APSP Hypothesis fails.
\end{corollary}
\begin{proof}
Suppose that $\MinPlus(n, n^\alpha, n \mid u \leq n^\alpha) = O(n^{2+\alpha-\epsilon})$ for some $\epsilon > 0$. \cref{thm:exact-tri-low-rank} states that $\ExactTri(n, n, n \mid r \leq n^{3-\omega-\epsilon/2}) \leq O(n^{3-\epsilon/500000})$. Thus, from \cref{lem:sm-univ-reduct} with $n_1 = n_2 = n_3 = n$, $n_2' = n^\alpha$, and $t = n^{\epsilon / 2}$ we get that
\begin{align*}
    &\MinPlus(n, n, n \mid u \leq n^{3-\omega}) \\
    &\qquad = \tilde O\parens*{\MinPlus(n, n^\alpha, n \mid u \leq n^\alpha) \cdot n^{1-\alpha+\epsilon/2} + \ExactTri(n, n, n \mid r \leq n^{3-\omega-\epsilon/2})} \\
    &\qquad = \tilde O\parens*{n^{2+\alpha-\epsilon} \cdot n^{1-\alpha+\epsilon/2} + n^{3-\epsilon/500000}} \\
    &\qquad = \tilde O\parens*{n^{3-\epsilon/500000}}.
\end{align*}
This contradicts the Strong APSP Hypothesis.
\end{proof}

\begin{remark}
The same proof shows that the premise, $\MinPlus(n, n^\alpha, n \mid u \leq n^\alpha) = O(n^{2+\alpha-\epsilon})$, is \emph{equivalent} for all $\alpha \in (0, 1]$, i.e., if the premise holds for some $\alpha \in (0, 1]$ then it also holds for all~\makebox{$\alpha' \in (0, 1]$}, provided that $\omega = 2$. Curiously, this equivalence statement was already established for all $\alpha \in (0, 1)$ in previous work by Chan, Vassilevska W., and Xu~\cite[Corollary~7.6]{ChanVX23}---only the equivalence with $\alpha = 1$ (i.e., the Strong APSP hardness) was missing.
\end{remark}

\cref{cor:min-plus-rect-sm-univ} combined with the equivalence of $\MinPlus(n, n^\mu, n \mid u \leq n^{1-\mu})$ and directed unweighted APSP due to~\cite{ChanVX21} completes the proof of \cref{thm:strong-apsp-implies-dir-unw-apsp}.

In the following, we quickly prove the strengthening of \cref{thm:strong-apsp-implies-dir-unw-apsp} that Zwick's algorithm is best-possible for arbitrary weight bounds provided that $\omega = 2$, \cref{thm:zwick-opt}. For the proof we rely on the following simple generalization of~\cite{ChanVX21}:

\begin{lemma} \label{lem:min-plus-sm-univ-to-apsp}
Let $\APSP(n \mid u)$ denote the running time of APSP in $n$-vertex graphs with weights $\set{0, \dots, u}$. Then $\MinPlus(n, n_2, n \mid u) = O(\APSP(n \mid u \leq n_2 u / n))$ for all parameters $n_2, u \leq n$.
\end{lemma}
\begin{proof}
Let $A, B$ be the given matrices of sizes $n \times n_2$ and $n_2 \times n$ with entries in $\set{0, \dots, u} \cup \set{\bot}$. Pick the parameters $q = \floor{n_2 u / n}$ and $p = 2\ceil{n / n_2}$. Since $pq \geq u$ we can rewrite the matrices in their ``$(p, q)$-ary'' representations
\begin{align*}
    A &= A_1 \cdot q + A_0, \\
    B &= B_1 \cdot q + B_0,
\end{align*}
where $A_0, B_0$ have entries in $\set{0, \dots, q-1} \cup \set{\bot}$ and $A_1, B_1$ have entries in $\set{0, \dots, p-1} \cup \set{\bot}$.

Next, we construct a graph $G$ as follows. The vertices consist of three separate sets $I \sqcup K \sqcup J$. The sets~$I$ and~$J$ are each copies of $[n]$, and $K$ is $[n_2] \times \set{-p, \dots, p}$. We choose the edges as follows: For each non-$\bot$ entry~$A[i, k]$ we insert an edge from $i \in I$ to $(k, -A_1[i, k]) \in K$ with weight $A_0[i, k]$. For each non-$\bot$ entry~$B[k, j]$ we insert an edge from $(k, B_1[k, j]) \in K$ to $j \in J$ with weight $B_0[k, j]$. And finally we turn each of the vertex sequences $(k, -p), \dots, (k, p)$ into a directed path, where each edge along the path has weight $q$.

We claim that the min-plus product $A * B$ equals exactly the $I$-to-$J$ distances in $G$. Indeed, the only directed paths from $i \in I$ to $j \in J$ take the form $i, (k, -A_1[i, k]), \dots, (k, B_1[k, j]), j$. The weight of such a path is exactly $A_0[i, k]$ (the first edge) plus $B_0[k, j]$ (the last edge) plus $q \cdot (A_1[i, k] + B_1[k, j])$ (the middle edges). In sum this is $A[i, k] + B[k, j]$, so the claim follows.

Moreover, $G$ is a graph on $|I| + |K| + |J| \leq n + O(p n_2) + n = O(n)$ vertices, and the largest edge weight in $G$ is $q \leq n_2 u / n$. Hence we can solve APSP on $G$ in time $\APSP(n \mid u \leq n_2 u / n)$ and read off the desired min-plus product from the computed distances.
\end{proof}

\thmzwickoptimal*
\begin{proof}
Suppose that $\APSP(n \mid u \leq n^{\delta}) = O(n^{2.5+\delta/2-\epsilon})$ for some constant $\epsilon > 0$. Then by \cref{lem:min-plus-sm-univ-to-apsp} it follows that $\MinPlus(n, n^{1/2+\delta/2}, n \mid u \leq n^{1/2+\delta/2}) = O(n^{2.5+\delta/2-\epsilon})$. This contradicts the Strong APSP Hypothesis by \cref{cor:min-plus-rect-sm-univ} (with $\alpha = 1/2+\delta/2$).
\end{proof}

Finally, we consider undirected graphs and prove that the Shoshan-Zwick algorithm is best-possible conditioned on the Strong APSP Hypothesis.

\thmshoshanzwickoptimal*
\begin{proof}
Suppose that APSP in undirected graphs with weights $\set{0, \dots, n^\delta}$ is in time $O(n^{2+\delta-\epsilon})$. By a standard construction it follows easily that $\MinPlus(n, n, n \mid u < n^\delta / 2) = O(n^{2+\delta-\epsilon})$. Indeed, we transform the given pair of matrices $A, B$ with entries less than $u = n^{\delta} / 2$ into an undirected graph with three vertex layers $I, K, J$, where we set the weight of any edge $(i, k) \in I \times K$ to $A[i, k] + u$ and of any edge $(k, j) \in K \times J$ to $B[k, j] + u$, and where we delete all edges with $A[i, k] = \bot$ or $B[k, j] = \bot$. Observe that the $I$-$J$ distances in the resulting graph are exactly $A * B$ plus an additive offset of $2u$ (except that for distances larger than~$2u$ which correspond to $\bot$-entries in $A * B$). The reason is that any path of length more than $2$ has length at least $4$ and thus weight at least $4u$, which exceeds any entry in $A * B$ plus $2u$.

But $\MinPlus(n, n, n \mid u < n^\delta / 2) = O(n^{2+\delta-\epsilon})$ contradicts the Strong APSP Hypothesis by \cref{cor:min-plus-rect-sm-univ} applied with $\alpha = \delta$ (observing that the constant factor $1/2$ in the universe bound can easily be removed).
\end{proof}

\subsection{Doubling Reduction} \label{sec:univ-reduct:sec:doubling}
Next, we prove the following ``doubling reduction'' which allows to reduce from arbitrary instances of Min-Plus Product to Low-Doubling instances.

\begin{lemma}[Doubling Reduction] \label{lem:doubling-reduct}
For all parameters $n_1, n_2, n_3, K$ with $n_2 \leq n_1, n_3$ it holds that:
\begin{equation*}
    \MinPlus(n_1, n_2, n_3 \mid u) \leq \parens*{K \cdot \MinPlus(n_1, n_2, n_3 \mid u; D \leq n_2; \rho \leq 1/D; K) + \frac{n_1 n_2 n_3}{K^{1/300000}}}^{1+o(1)} \cdot \log u.
\end{equation*}
\end{lemma}
\begin{proof}
Let $L = \ceil{\log u}$. The proof of the lemma relies on a scaling trick over $L$ scales, and requires some notational setup. In analogy to \cref{lem:sm-univ-reduct}, we call $k$ a \emph{witness} of $(A * B)[i, j]$ if $(A * B)[i, j] = A[i, k] + B[k, j]$, and we call~$k$ an \emph{$q$-pseudo-witness} if $A[i, k] + B[k, j] < (A * B)[i, j] + q$. (Note that the definition of pseudo-witness differs slightly from \cref{lem:sm-univ-reduct}.)

For each $0 \leq \ell \leq L$, define the matrices~$A_\ell = \floor{A / 2^\ell}$ and $B_\ell = \floor{B / 2^\ell}$. We partition each such matrix into two submatrices, $A_\ell = A_{\ell,0} \sqcup A_{\ell,1}$, where $A_{\ell, 0}$ retains all entries with $A[i, k] \bmod 2^\ell < 2^{\ell-1}$, and $A_{\ell,1}$ retains all other entries. Similarly, partition $B_\ell = B_{\ell,0} \sqcup B_{\ell,1}$. A observation related to this scaling setup is the following:

\begin{claim} \label{lem:doubling-reduct:clm:scaling}
For each $(i, j) \in [n_1] \times [n_3]$ and each scale $0 \leq \ell \leq L$, there is some pair $x, y \in \set{0, 1}$ such that:
\begin{enumerate}[label=(\roman*)]
    \item\label{lem:doubling-reduct:clm:scaling:itm:scaling} $(A_{\ell, x} * B_{\ell, y}) = \floor{((A * B)[i, j] - (x + y) 2^{\ell-1}) / 2^\ell}$,
    \item\label{lem:doubling-reduct:clm:scaling:itm:one-witness} some witness of $(A_{\ell, x} * B_{\ell, y})[i, j]$ is a witness of $(A * B)[i, j]$, and
    \item\label{lem:doubling-reduct:clm:scaling:itm:all-pseudo-witnesses} every witness of $(A_{\ell, x} * B_{\ell, y})[i, j]$ is a $2^\ell$-pseudo-witness of $(A * B)[i, j]$.
\end{enumerate}
\end{claim}
\begin{proof}
To see this, let $k^*$ be an arbitrary witness of $(A * B)[i, j]$. By construction each non-$\bot$ entry in $A_\ell$ shows up either in $A_{\ell, 0}$ or $A_{\ell, 1}$. Take $x \in \set{0, 1}$ such that $A_{\ell, x}[i, k^*] \neq \bot$. Similarly, take $y \in \set{0, 1}$ so that~\makebox{$B_{\ell, y}[k^*, j] \neq \bot$}. Now consider an arbitrary $k \in [n_2]$ with $A_{\ell, x}[i, k] \neq \bot$ and $B_{\ell, y}[k, j] \neq \bot$, and note that
\begin{equation*}
    A_{\ell, x}[i, k] + B_{\ell, y}[k, j]
    = \floor*{\frac{A[i, k]}{2^\ell}} + \floor*{\frac{B[k, j]}{2^\ell}} = \frac{A[i, k] + B[k, j] - (A[i, k] \bmod 2^\ell) - (B[k, j] \bmod 2^\ell)}{2^\ell}.
\end{equation*}
By definition, $A[i, k] \bmod 2^\ell$ falls into the range $x2^{\ell-1} + [0, 2^{\ell-1})$, and $B[k, j] \bmod 2^\ell$ is in $y2^{\ell-1} + [0, 2^{\ell-1})$. In particular, we obtain the upper bound
\begin{equation*}
    A_{\ell, x}[i, k] + B_{\ell, y}[k, j]
    \leq \frac{A[i, k] + B[k, j] - (x + y) 2^{\ell-1}}{2^\ell},
\end{equation*}
and the lower bound
\begin{equation*}
    A_{\ell, x}[i, k] + B_{\ell, y}[k, j]
    > \frac{A[i, k] + B[k, j] - (x + y) 2^{\ell-1} - 2^\ell}{2^\ell} = \frac{A[i, k] + B[k, j] - (x + y) 2^{\ell-1}}{2^\ell} - 1,
\end{equation*}
Recalling that this is an integer, we must have
\begin{equation*}
    A_{\ell, x}[i, k] + B_{\ell, y}[k, j] = \floor*{\frac{A[i, k] + B[k, j] - (x + y) 2^{\ell-1}}{2^\ell}}.
\end{equation*}
Therefore, $k^*$ is a minimizer of this expression and thus a witness of $(A_{\ell, x} * B_{\ell, y})[i, j]$, proving~\ref{lem:doubling-reduct:clm:scaling:itm:one-witness}. Moreover, for each other witness $k$ we must have $A[i, k] + B[k, j] < A[i, k^*] + B[k^*, j] + 2^\ell$, and thus $k$ is a $2^\ell$-pseudo-witness, proving~\ref{lem:doubling-reduct:clm:scaling:itm:all-pseudo-witnesses}.
\end{proof}

Returning to the proof of \cref{lem:doubling-reduct}, we are ready to state the reduction. We classify the pairs~$(i, j)$ into three categories based on their number of pseudo-witnesses:
\begin{enumerate}
    \item If $(A * B)[i, j]$ has at least $t$ $1$-pseudo-witnesses (i.e., witnesses) then we call $(i, j)$ \emph{exceptionally popular}.
    \item If $(A * B)[i, j]$ has at most $t$ $2^L$-pseudo-witnesses then we call $(i, j)$ \emph{exceptionally unpopular}.
    \item Otherwise we call $(i, j)$ \emph{ordinary}. For each ordinary pair there is some level $0 \leq \ell < L$ such that~$(i, j)$ has at most $t$ $2^\ell$-pseudo-witnesses and simultaneously at least $t$ $2^{\ell+1}$-pseudo-witnesses---we say that~$(i, j)$ is \emph{relevant} at that level $\ell$.
\end{enumerate}
These definitions are only required for the analysis, and we do not compute the classification explicitly.

We will now describe the algorithm. We initially set $C$ to be the all-$\bot$ matrix. Throughout we will only update $C$ via $C[i, j] \gets A[i, k] + B[k, j]$, so we will never underestimate the values of $A * B$. We will deal with the exceptionally popular, exceptionally unpopular and ordinary pairs in three separate steps. The first and especially the second one can be regarded as exceptional special cases. 
\begin{enumerate}
    \item \emph{(Exceptionally Popular Pairs)} Sample a uniformly random subset $\mathcal K \subseteq [n_2]$ with rate $\Theta(t^{-1} \log (n_1 n_3))$. Then enumerate all triples $(i, k, j) \in [n_1] \times \mathcal K \times [n_3]$ and update $C[i, j] \gets \min\set{C[i, j], A[i, k] + B[k, j]}$. Each exceptionally popular pair $(i, j)$ has at least $t$ witnesses $k$, so with high probability we include at least one such witness in~$\mathcal K$. In this case we correctly assign $C[i, j]$.
    
    \begin{framed}
        \emph{Running Time:} $\tilde O(n_1 n_2 n_3 / t)$.
    \end{framed}

    \item \emph{(Exceptionally Unpopular Pairs)} By definition a pair $(i, j)$ is exceptionally unpopular if it has less than~$t$ $2^L$-pseudo-witnesses. But the matrices $A_L$ and $B_L$ in fact only contain two possible entries: $0$~and~$\bot$. Hence, a pair $(i, j)$ is exceptionally unpopular if and only if the number of indices $k$ with~\makebox{$A[i, k] \neq \bot$} and $B[k, j] \neq \bot$ is at most $t$. One way to deal with these pairs is to apply \cref{lem:min-plus-wit} in combination with the fast algorithm for Min-Plus Product to list $t$ witnesses in the product $A_L * B_L$. (However, this product is of course just a Boolean matrix multiplication, so one could have also solved this step in $\tilde O(t)$ matrix products.) 
    
    \begin{framed}
        \emph{Running Time:} $\tilde O(t \cdot \MinPlus(n_1, n_2, n_3 \mid u \leq 0))$.
    \end{framed}

    \item\label{lem:doubling-reduct:itm:ord} \emph{(Ordinary Pairs)} To deal with the ordinary pairs, we prepare a low-rank matrix approximating $A * B$ (in a sense that will be clear soon). Sample a uniformly random subset $\mathcal K \subseteq [n_2]$ with rate $\Theta(t^{-1} \log (n_1 n_3))$. Let $U$ be the submatrix of~$A$ restricted to the columns in $\mathcal K$ and let $V$ be the submatrix of $B$ restricted to the rows in $\mathcal K$. Compute the min-plus product $R = U * V$ by brute-force. Note that $R$ is a matrix of rank at most $r = |\mathcal K|$ as is witnessed by the rank-$r$ decomposition $(U, V, S)$ (for some appropriate matrix $S$ that is easy to compute along the way). With high probability, the rank is bounded by $r = |\mathcal K| = \tilde O(n_2 / t)$.

    \begin{framed}
        \emph{Running Time:} $\tilde O(n_1 n_2 n_3 / t)$.
    \end{framed}
    Run the following steps for all levels $0 \leq \ell \leq L$, for all $x, y \in \set{0, 1}$ and all $z \in \set{0, 1, 2, 3}$:
    \begin{enumerate}
            \item\label{lem:doubling-reduct:itm:ord:itm:reduct}
            Let $R_{\ell, z} = \floor{R / 2^\ell}$ be the $n_1 \times n_3$ matrix defined by
            \begin{equation*}
                R_{\ell, z}[i, j] = \floor*{\frac{R[i, j]}{2^\ell}} - z.
            \end{equation*}
            Obtain a rank-$r' = O(r)$ decomposition for $R_{\ell, z}$ by appropriately rounding the rank-$r$ decomposition of~$R$.
            Run \cref{lem:exact-tri-low-rank-to-low-doubling} on the Exact Triangle instance $(A_{\ell,x}, B_{\ell,y}, R_{\ell,z})$ with parameter~$K$. This is a potential-adjusting reduction, and hence the output is $(\mathcal I, \mathcal T)$---a set of $r'$-uniform $K$-doubling Exact Triangle instances~$\mathcal I$ and a set of special triples $\mathcal T \subseteq [n_1] \times [n_2] \times [n_3]$. Afterwards, enumerate all triples $(i, k, j) \in \mathcal T$ and update $C[i, j] \gets \min\set{C[i, j], A[i, k] + B[k, j]}$.

            \begin{framed}
                \emph{Running Time:} The $O(1)$ calls to \cref{lem:exact-tri-low-rank-to-low-doubling} take time
                \begin{equation*}
                    \parens*{\frac{n_1 n_2 n_3}{K^{1/98000}} + (n_1 n_2 + n_2 n_3 + n_1 n_3) r' K}^{1+o(1)} \leq \parens*{\frac{n_1 n_2 n_3}{K^{1/98000}} + \frac{n_1 n_2 n_3 K}{t}}^{1+o(1)}.
                \end{equation*}
                For the inequality we have used that $n_2 \leq n_1, n_3$ and that $r' = O(r) = \tilde O(n_2 / t)$.
            \end{framed}
            
            \item\label{lem:doubling-reduct:itm:ord:itm:solve} Enumerate all instances $(P', Q', R') \in \mathcal I$. \cref{lem:exact-tri-low-rank-to-low-doubling} guarantees that $(P', Q', R')$ is an $D$\=/uniform $1/D$-regular $K$-doubling Exact Triangle instance for some $D \leq r'$. Hence, $(P', Q')$ is an $D$-uniform $1/D$-regular $K$-doubling Min-Plus Product instance. In fact, we may loosely bound $D \leq r' = O(r) = O(n_2)$ here. We compute the min-plus product $P' * Q'$ and list $t$ witnesses for each output entry $(i, j)$ by \cref{lem:min-plus-wit}. For each triple $(i, k, j)$ reported in this way we update $C[i, j] \gets \min\set{C[i, j], A[i, k] + B[k, j]}$.
            
            \begin{framed}
                \emph{Running Time:}
                \begin{align*}
                    &\tilde O\parens*{|\mathcal I| \cdot t \cdot \MinPlus(n_1, n_2, n_3 \mid u; D \leq n_2; \rho \leq 1/D; K)} \\
                    &\qquad\leq \parens*{K \cdot t \cdot \MinPlus(n_1, n_2, n_3 \mid u; D \leq n_2; \rho \leq 1/D; K)}^{1+o(1)}.
                \end{align*}
            \end{framed}
        \end{enumerate}
    We argue that we correctly compute $C[i, j]$ for all ordinary pairs $(i, j)$. Fix any such pair, and let~$\ell$ be a level at which $(i, j)$ is relevant. By definition, $(i, j)$ has at least $t$ $2^{\ell+1}$\=/pseudo-witnesses. Thus, with high probability we include at least one $2^{\ell+1}$\=/pseudo-witness in the randomly sampled set $\mathcal K$. It follows that
    \begin{equation*}
        (A * B)[i, j] \leq R[i, j] < (A * B)[i, j] + 2^{\ell+1}.
    \end{equation*}
    Property~\ref{lem:doubling-reduct:clm:scaling:itm:scaling} of \cref{lem:doubling-reduct:clm:scaling} further implies that there exist $x, y \in \set{0, 1}$ satisfying
    \begin{equation*}
        (A_{\ell, x} * B_{\ell, y})[i, j] = \floor*{\frac{(A * B)[i, j] - (x + y)2^{\ell-1}}{2^\ell}}.
    \end{equation*}
    Putting both statements together, we have
    \begin{equation*}
        \floor*{\frac{R[i, j]}{2^\ell}} - 3 \leq (A_{\ell, x} * B_{\ell, y})[i, j] \leq \floor*{\frac{R[i, j]}{2^\ell}},
    \end{equation*}
    and consequently, there is some $z \in \set{0, 1, 2, 3}$ such that $(A_{\ell, x} * B_{\ell, y})[i, j] = R_{\ell, z}[i, j]$. In particular, for each witness $k$ of $(A_{\ell, x} * B_{\ell, y})[i, j]$, $(i, k, j)$ forms an exact triangle in the instance $(A_{\ell, x}, B_{\ell, y}, R_{\ell, z})$. Moreover, all these witnesses are $2^\ell$-pseudo-witnesses of $(A * B)[i, j]$ by Property~\ref{lem:doubling-reduct:clm:scaling:itm:all-pseudo-witnesses}, and they include at least one proper witness $k^*$ of $(A * B)[i, j]$ by Property~\ref{lem:doubling-reduct:clm:scaling:itm:one-witness}.
    
    We argue that in the substeps~\ref{lem:doubling-reduct:itm:ord:itm:reduct} and~\ref{lem:doubling-reduct:itm:ord:itm:solve} we enumerate all witnesses of $(A_{\ell, x} * B_{\ell, y})[i, j]$, including the desired proper witness $k^*$. \cref{lem:exact-tri-low-rank-to-low-doubling} implements a potential-adjusting reduction. We are thus guaranteed that each exact triangle in $(A_{\ell, x}, B_{\ell, y}, R_{\ell, z})$ appears in $\mathcal T$ or it appears as an exact triangle in one of the instances in $\mathcal I$. If the exact triangle $(i, k^*, j)$ falls into the former case, we successfully discover it in substep~\ref{lem:doubling-reduct:itm:ord:itm:reduct}. We focus on the latter case.
    
    Focus on the instance $(A', B', R') \in \mathcal I$ in which the exact triangle $(i, k^*, j)$ is present. Recall that this instance is a potential adjustment of $(A_{\ell, x}, B_{\ell, y}, R_{\ell, z})$, and thus
    \begin{equation*}
        \makeatletter\expandafter\def\expandafter\@arraycr\expandafter{\@arraycr\noalign{\vskip\jot}}\makeatother%
        \begin{array}{@{}r@{\;}c@{\;}c@{\:}c@{\:}c@{\:}c@{\:}c@{\;\;\;}c@{\;\;\;}r@{\;}c@{\;}c@{\;\;\;}c@{\ }c@{\;}c@{\;}c@{\:}c@{\:}c@{}}
            A'[i, k] & = & A_{\ell, x}[i, k] & + & u[i] & + & v[k] & \text{or} & A'[i, k] & = & \bot & \text{for all} & (i, k) & \in & [n_1] & \times & [n_2], \\
            B'[k, j] & = & B_{\ell, y}[k, j] & - & v[k] & + & w[j] & \text{or} & B'[k, j] & = & \bot & \text{for all} & (k, j) & \in & [n_2] & \times & [n_3], \\
            R'[i, j] & = & R_{\ell, z}[i, j] & + & u[i] & + & w[j] & \text{or} & R'[i, j] & = & \bot & \text{for all} & (i, j) & \in & [n_1] & \times & [n_3],
        \end{array}
    \end{equation*}
    for some potential functions $u, v, w$. From this it follows that $(A' * B')[i, j] = R'[i, j]$. On the one hand, $k^*$ witnesses $(A' * B')[i, j] \leq R'[i, j]$. On the other hand, we cannot have $(A' * B')[i, j] < R'[i, j]$ as then, rewriting $A', B', R'$ by the previous three identities, we would also have $(A_{\ell, x} * B_{\ell, y})[i, j] < R_{\ell, z}[i, j]$, leading to a contradiction.

    Now, as $(A' * B')[i, j] = R'[i, j]$, it follows that each witness $k$ of $(A' * B')[i, j]$ must form an exact triangle $(i, k, j)$ in $(P', Q', R')$ and thus also in $(A_{\ell, x}, B_{\ell, y}, R_{\ell, z})$. Hence, $k$ is a witness of \makebox{$(A_{\ell, x} * B_{\ell, y})[i, j]$}, and Property~\ref{lem:doubling-reduct:clm:scaling:itm:all-pseudo-witnesses} implies that $k$ is a $2^\ell$-pseudo-witness of $(i, j)$. Recall that we assume that $(i, j)$ is relevant at level $\ell$, hence the total number of $\ell$-pseudo-witnesses is at most $t$. Therefore, we are guaranteed to list \emph{all} witnesses of $(A' * B')[i, j]$ by \cref{lem:min-plus-wit} (with high probability). In particular, we will list the proper witness $k^*$, and therefore update $C[i, j]$ as intended.
\end{enumerate}
This completes the description of the reduction and the correctness analysis.

\proofparagraph{Running Time}
Summing over all $O(L) = O(\log u)$ levels, the total time of the steps as analyzed above is
\begin{equation*}
    \parens*{K t \cdot \MinPlus(n_1, n_2, n_3 \mid u; D \leq n_2; \rho \leq 1/D; K) + \frac{n_1 n_2 n_3}{K^{1/98000}} + \frac{n_1 n_2 n_3 K}{t}}^{1+o(1)} \cdot \log u.
\end{equation*}
Pick $t = K^{1+1/98000} \leq K^2$. The total running time becomes
\begin{equation*}
    \parens*{K^3 \cdot \MinPlus(n_1, n_2, n_3 \mid u; D \leq n_2; \rho \leq 1/D; K) + \frac{n_1 n_2 n_3}{K^{1/98000}}}^{1+o(1)} \cdot \log u,
\end{equation*}
which is as claimed by readjusting $K$ to $K^{1/3}$.
\end{proof}

\begin{corollary} \label{cor:min-plus-rect-low-doubling}
Let $\kappa, \epsilon > 0$ and $0 < \alpha_2 \leq \alpha_1, \alpha_3$ be constants. If $\MinPlus(n^{\alpha_1}, n^{\alpha_2}, n^{\alpha_3} \mid D \leq n^{\alpha_2}; \rho \leq 1/D; K \leq n^{\kappa}) \leq O(n^{\alpha_1+\alpha_2+\alpha_3-\epsilon})$ then the APSP Hypothesis fails.
\end{corollary}
\begin{proof}
Suppose that $\MinPlus(n^{\alpha_1}, n^{\alpha_2}, n^{\alpha_3} \mid D \leq n^{\alpha_2}; \rho \leq 1/D; K \leq n^{\kappa}) \leq O(n^{\alpha_1+\alpha_2+\alpha_3-\epsilon})$, i.e., the $D$-uniform $1/D$-regular $n^\kappa$-doubling Min-Plus Product is in time $O(n^{\alpha_1+\alpha_2+\alpha_3-\epsilon})$ for all $D \leq n^{\alpha_2}$. Then by \cref{lem:doubling-reduct} applied with parameter $K = n^{\min\set{\kappa, \epsilon/2}}$ the general Min-Plus Product of size $n^{\alpha_1} \times n^{\alpha_2} \times n^{\alpha_3}$ over the universe~$[u]$ can be solved in time
\begin{align*}
    &\MinPlus(n^{\alpha_1}, n^{\alpha_2}, n^{\alpha_3} \mid u) \\
    &\qquad\leq \parens*{\MinPlus(n^{\alpha_1}, n^{\alpha_2}, n^{\alpha_3} \mid u; D \leq n^{\alpha_2}; K \leq n^{\kappa}) \cdot n^{\epsilon/2} + \frac{n^{\alpha_1 + \alpha_2 + \alpha_3}}{n^{\min\set{\kappa, \epsilon / 2}/300000}}}^{1+o(1)} \cdot \log u \\
    &\qquad\leq \parens*{n^{\alpha_1 + \alpha_2 + \alpha_3 - \epsilon / 2} + n^{\alpha_1 + \alpha_2 + \alpha_3 - \min\set{\kappa, \epsilon / 2}/300000}}^{1+o(1)} \cdot \log u \\
    &\qquad\leq \vphantom{\Big(}n^{\alpha_1 + \alpha_2 + \alpha_3 - \min\set{\kappa, \epsilon / 2}/300000 + o(1)} \cdot \log u.
\end{align*}
In particular, for any polynomially bounded universe $u \leq n^c$ (where $c$ is constant) the running time becomes truly subcubic, $O(n^{\alpha_1 + \alpha_2 + \alpha_3 - \gamma})$ for any constant $0 < \gamma < \min\set{\kappa, \epsilon / 2} / 300000$. This contradicts the APSP Hypothesis.
\end{proof}

\cref{cor:min-plus-rect-low-doubling} also immediately implies \cref{thm:apsp-unif-low-doubling}: APSP reduces to Min-Plus Product, which reduces to $n$-Uniform $n^\kappa$-Doubling Min-Plus Product by \cref{thm:apsp-unif-low-doubling}, which in turn can be seen as an APSP instance on a 3-layered graph with at most $n$ distinct weights and doubling $n^\kappa$. For many lower bounds the following simpler corollary (that ignores the regular and low-doubling constraints) turns out to be sufficient.

\begin{corollary} \label{cor:min-plus-rect-unif}
Let $\epsilon > 0$ and $0 < \alpha_2 \leq \alpha_1, \alpha_3$ be constants. If $\MinPlus(n^{\alpha_1}, n^{\alpha_2}, n^{\alpha_3} \mid D \leq n^{\alpha_2}) \leq O(n^{\alpha_1+\alpha_2+\alpha_3-\epsilon})$ then the APSP Hypothesis fails.
\end{corollary}

\subsection{Large-Universe Reduction} \label{sec:univ-reduct:sec:hashing}
In this subsection we finally provide the missing reduction from polynomially large universes $\set{0, \dots, n^{O(1)}}$ to~$\set{0, \dots, n}$, conditioned on the additive combinatorics assumption \cref{hypo:hashing}; we refer to \cref{sec:hashing} for a discussion of this hypothesis. To effectively apply \cref{hypo:hashing} we rely on the following greedy covering lemma.

\begin{lemma}[Greedy Covering] \label{lem:sumset-cover}
There is an algorithm that, given $X, Y \subseteq \Int$, computes a set $S \subseteq \Int$ of size at most $|Y - X| / |Y| \cdot \ln |X|$ such that $X \subseteq Y + S$. It runs in deterministic time $|Y - X| \, |S| \cdot (\log u)^{O(1)}$, where $u$ is an upper bound on the largest integer in $X, Y$ in absolute value.
\end{lemma}
\begin{proof}
We run a simple greedy algorithm. Initialize $S \gets \emptyset$. While $X \neq \emptyset$ we compute a shift $s \in \Int$ that maximizes $r_{Y-X}(s)$. We update $S \gets S \cup \set{s}$ and $X \gets X \setminus (Y + s)$. When the algorithm terminates, we return $S$.

It is clear that when the algorithm terminates the set $S$ is as desired. It remains to bound $|S|$, i.e., to bound the number of iterations. Let $X_0$ denote the initial set $X$. We show that with each iteration of the algorithm removes at least a $|Y| / |Y - X_0|$-fraction of the elements from $X$; from this it follows easily that after $|Y - X_0| / |Y| \cdot \ln |X_0|$ iterations the set $X$ must be empty. To see this, observe that
\begin{equation*}
    \sum_{s \in Y - X} r_{Y-X}(s) = |X| \, |Y|,
\end{equation*}
and hence there is some $s \in Y - X$ with
\begin{equation*}
    r_{Y-X}(s) \geq \frac{|X| \, |Y|}{|Y - X|} \geq \frac{|X| \, |Y|}{|Y - X_0|}.
\end{equation*}
This shift $s$ witnesses that the algorithm decreases the size of $|X|$ by at least a $|Y| / |Y - X_0|$-fraction as desired.

We finally consider the running time. By \cref{lem:sparse-conv} we can compute the set $Y - X$ along with all multiplicities $r_{Y-X}(\cdot)$ in deterministic time $O(|Y - X| \cdot (\log u)^{O(1)})$. Finding $s$ and updating $S$ and $X$ then runs in time $O(|Y - X|)$. We repeat this over $|S|$ iterations, leading to the claimed total time.
\end{proof}

\begin{remark}
This lemma relates to the well-known ``Ruzsa's Covering Lemma''~\cite{Ruzsa99}; see~\cite[Lemma~2.14]{TaoV06}. That fundamental result alternatively covers $X$ by shifts of $Y - Y$ (instead of $Y$) but achieves $|S| \leq |Y - X| / |Y|$ (without the logarithmic factor).
\end{remark}

Equipped with \cref{lem:sumset-cover}, we now show that \cref{hypo:hashing} implies a reduction from low-doubling Min-Plus Product to small-universe Min-Plus Product.

\begin{lemma} \label{lem:min-plus-sum-order-preserving}
Assuming \cref{hypo:hashing} there is a constant $c$ such that for all parameters $n_1, n_2, n_3, K$ with $n_2 \leq n_1 n_3$:
\begin{equation*}
    \MinPlus(n_1, n_2, n_3 \mid D \leq n_2; K) = \tilde O\parens*{K^c \cdot \MinPlus(n_1, n_2, n_3 \mid u \leq n_2)}.
\end{equation*}
\end{lemma}
\begin{proof}
To bound $\MinPlus(n_1, n_2, n_3 \mid D \leq n_2; K)$ we design an algorithm for the $n_2$-uniform $K$-doubling Min-Plus Product problem of size $n_1 \times n_2 \times n_3$. Let $A, B \in (X \cup \set{\bot})^{n \times n}$ be the given matrices where $X \subseteq \Int$ denotes their set of integer entries, i.e., $|X| \leq n_2$ and $|X + X| \leq K |X|$. Our goal is to compute their product $C$.

We first use the algorithm from \cref{hypo:hashing} to compute a set $Y \subseteq X$ of size $|Y| \geq |X| / K^c$ and a sum-order-preserving additive hash function $h : Y \to \set{0, \dots, |X|}$. Then apply \cref{lem:sumset-cover} to construct a set~$S$ satisfying that $X \subseteq Y + S$. We enumerate all pairs $s, t \in S$, and for each such pair construct the following two matrices:
\begin{align*}
    A_s[i, k] &=
    \begin{cases}
        \mathrlap{h(A[i, k] - s)}\phantom{h(B[k, j] - t)} &\text{if $A[i, k] - s \in Y$,} \\
        \bot &\text{otherwise,}
    \end{cases} \\
    B_t[k, j] &=
    \begin{cases}
        h(B[k, j] - t) &\text{if $B[k, j] - t \in Y$,} \\
        \bot &\text{otherwise.}
    \end{cases}
\end{align*}
Clearly this is well-defined---we only apply $h$ to elements in $Y$. Moreover, the Min-Plus Product instance $(A_s, B_t)$ has a universe size of $|X| \leq n_2$. We solve the instance and report for each output pair $(i, j)$ a witness~$k$ by \cref{lem:min-plus-wit}. We update $C[i, j] \gets \min\set{C[i, j], A[i, k] + B[k, j]}$ for each triple $(i, k, j)$ listed in this way.

\proofparagraph{Correctness}
Fix any pair $(i, j)$ and let $k$ be a witness of $(i, j)$. We have that $A[i, k], B[k, j] \in X \subseteq Y + S$, hence there are $s, t \in S$ such that $A[i, k] - s, B[k, j] - t \in Y$. In particular, we have that $A_s[i, k] = h(A[i, k] - s)$ and $B_t[k, j] = h(B[k, j] - t)$, so it follows that $(A_s * B_t)[i, j] \leq h(A[i, k] - s) + h(B[k, j] - t)$. Conversely, for any non-witness $k'$ we have
\begin{equation*}
    A[i, k] + B[k, j] < A[i, k'] + B[k', j]
\end{equation*}
and thus
\begin{equation*}
    (A[i, k] - s) + (B[k, j] - t) < (A[i, k'] - s) + (B[k', j] - t).
\end{equation*}
As the function $h$ is sum-order-preserving, it follows that
\begin{equation*}
    h(A[i, k] - s) + h(B[k, j] - t) < h(A[i, k'] - s) + h(B[k', j] - t),
\end{equation*}
or equivalently,
\begin{equation*}
    A_s[i, k] + B_t[k, j] < A_s[i, k'] + B_t[k', j]
\end{equation*}
(unless one of these entries is $\bot$). Thus, $k'$ is not a witness of $(A_s * B_t)[i, j]$. All in all, this proves that~$k$ is a witness of $(A * B)[i, j]$ if and only if $k$ is a witness of $\min_{s, t \in S} (A_s * B_t)[i, j]$. Therefore, the algorithm is forced to enumerate a proper witness $k$ of $(A * B)[i, j]$ and to correctly assign $C[i, j]$.

\proofparagraph{Running Time}
First we compute the set $Y \subseteq X$ and the hash function $h$ in time $|X|^{1+o(1)} \cdot K^{c'}$, for some constant $c'$, by \cref{hypo:hashing}. Then we apply \cref{lem:sumset-cover} which runs in time $|Y - X| \, |S| \cdot (\log u)^{O(1)}$. By the Plünnecke-Ruzsa inequality (\cref{lem:pluennecke-ruzsa}) we can bound $|Y - X| \leq |X - X| \leq K^2 |X|$. The set~$S$ is promised to have size $|S| \leq O(|Y - X| / |Y| \cdot \log |X|) \leq O(K^2 |X| / |Y| \cdot \log |X|)$ by \cref{lem:sumset-cover}, which, using the bound $|Y| \geq |X| / K^c$ from \cref{hypo:hashing}, is at most $O(K^{c+2} \log |X|)$. Hence, this step takes time at most $|X| \cdot K^{c+4} (\log u)^{O(1)}$. Then we enumerate all $|S|^2 \leq O(K^{2c+4} \log^2 |X|)$ pairs $s, t \in S$, and for each pair solve an $n_2$-universe Min-Plus product instance. The total time is:
\begin{equation*}
    \tilde O\parens*{|X|^{1+o(1)} \cdot K^{c'} + |X| \cdot K^{c+4} + K^{2c+4} \cdot \MinPlus(n_1, n_2, n_3 \mid u \leq n_2)}
\end{equation*}
Recall that $|X| \leq n_2$. Therefore, the lemma statement follows for the constant that is the maximum of $c'$ and $2c + 4$.
\end{proof}

\begin{remark}
As is apparent from the proof, \cref{lem:min-plus-sum-order-preserving} still holds even if the algorithm to compute~$Y$ and~$h$ in \cref{hypo:hashing} is slower than $|X|^{1+o(1)} K^{O(1)}$. Any algorithm in time $|X|^{3-\Omega(1)} K^{O(1)}$ would suffice.
\end{remark}

The proof of \cref{thm:apsp-implies-strong-apsp} (restated next) is now a simple combination of \cref{lem:doubling-reduct,lem:min-plus-sum-order-preserving}.

\thmapspimpliesstrongapsp*
\begin{proof}
It is clear that the Strong APSP Hypothesis implies the APSP Hypothesis. For the other direction assume that $\MinPlus(n, n, n \mid u \leq n) = O(n^{3-\epsilon})$. This is in particular the case if the Strong APSP Hypothesis fails and if simultaneously $\omega = 2$. Additionally assume \cref{hypo:hashing}. Under these assumptions it follows that:
\begin{alignat*}{2}
    \MinPlus(n, n, n \mid u)
    &= \parens*{K \cdot \MinPlus(n, n, n \mid u; D \leq n; K) + \frac{n^3}{K^{1/300000}}}^{1+o(1)} \!\!\cdot \log u \;\;\; &&\text{(\cref{lem:doubling-reduct})} \\
    &= \parens*{K^{c+1} \cdot \MinPlus(n, n, n \mid u \leq n) + \frac{n^3}{K^{1/300000}}}^{1+o(1)} \!\!\cdot \log u \; &&\text{(\cref{lem:min-plus-sum-order-preserving})} \\
    &= \parens*{K^{c+1} n^{3-\epsilon} + \frac{n^3}{K^{1/300000}}}^{1+o(1)} \!\!\cdot \log u,
\end{alignat*}
where $c$ is the constant from \cref{lem:min-plus-sum-order-preserving} and $K$ is an arbitrary parameter. Setting $K = n^{\epsilon / (2c+2)}$, the running time becomes $n^{3-\epsilon / (600000c+600000) + o(1)} \cdot \log u$, which, recalling that $u = n^{c'}$ for some arbitrarily large constant~$c'$, is truly subcubic as claimed.
\end{proof}

\begin{remark} \label{rem:hashing-weak}
The statement of \cref{thm:apsp-implies-strong-apsp} remains valid when assuming the weaker \cref{hypo:hashing-weak} instead (stated in \cref{sec:hashing}). That hypothesis leads to an overhead of $|X|^\epsilon \leq n^\epsilon$ for an arbitrarily small constant $\epsilon > 0$, and the constant $c$ from before would depend on $\epsilon$ (i.e., $c = c(\epsilon)$). 
\end{remark}
\section{Lower Bounds for Intermediate Problems} \label{sec:intermediate}
In this section, we derive APSP-based lower bounds for various graph and matrix problems with intermediate complexity. The proofs in this section are mostly technically simple consequences of \cref{cor:min-plus-rect-low-doubling} or~\ref{cor:min-plus-rect-unif}.

\subsection{Node-Weighted APSP} \label{sec:intermediate:sec:node-wgt-apsp}
Recall that node-weighted APSP can be solved in time $\tilde O(n^{(3+\omega)/2})$~\cite{Chan10,Yuster09,AbboudFJVX25}, even in directed graphs. Chan, Vassilevska W., and Xu~\cite{ChanW21} gave an $n^{2.5-o(1)}$ lower bound based on the Directed Unweighted APSP Hypothesis, even for undirected graphs. Here, we strengthen this lower bound to be based on the APSP Hypothesis.

\thmnodewgtapsp*
\begin{proof}
Suppose that Node-Weighted APSP is in time $O(n^{2.5-\epsilon})$ for some $\epsilon > 0$; we show that in this case $\MinPlus(n, \sqrt{n}, n \mid D \leq \sqrt{n}) = O(n^{2.5-\epsilon})$ which contradicts the APSP Hypothesis by \cref{cor:min-plus-rect-unif}.

Let $n$ be a square number, and let $A \in (X \cup \set{\bot})^{n \times \sqrt{n}}$ and $B \in (X \cup \set{\bot})^{\sqrt{n} \times n}$ be the given matrices, where $X \subseteq \Int$ is a set of size $\sqrt{n}$. We construct a 4-layered undirected node-weighted graph $G$ on $4n$ nodes as follows:
\begin{itemize}
    \item \emph{(Vertices)} The vertices consist of four distinct layers: $I = [n]$ (the first layer), $K_1 = [\sqrt{n}] \times X$ (the second layer), $K_2 = [\sqrt{n}] \times X$ (the third layer), and $J = [n]$ (the fourth layer).
    \item \emph{(Edges)} For each $A[i, k] \neq \bot$ we add an edge $(i, (k, A[i, k])) \in I \times K_1$, and for each $B[k, j] \neq \bot$ we add an edge $((k, B[k, j]), j) \in K_2 \times J$. Finally, add all edges $((k, x_1), (k, x_2)) \in K_1 \times K_2$.
    \item \emph{(Node Weights)} Let $u =\max_{x \in X} |x|$. All nodes in $I$ and $J$ have weight $10u$. All nodes $(k, x)$ in $K_1$ and~$K_2$ have weight $10u + x$.
\end{itemize}
We compute the pairwise distances in $G$ in time $O(n^{2.5-\epsilon})$, and claim that the resulting distances between any pair of nodes $i \in I$ and $j \in J$ is exactly $(A * B)[i, j] + 40u$. Indeed, whenever $k$ is a witness of $(A * B)[i, j]$ then the nodes $(i, (k, A[i, k]), (k, B[k, j]), j) \in I \times K_1 \times K_2 \times J$ form a path of weight
\begin{equation*}
    10u + (10u + A[i, k]) + (10u + B[k, j]) + 10u = (A * B)[i, j] + 40u.
\end{equation*}
Conversely, all $i$-$j$-paths of weight less than $42u$ necessarily have at most $3$ edges, and thus have the form $(i, (k', w_1), (k', w_2), j) \in I \times K_1 \times K_2 \times J$. But then $A[i, k'] = x_1$ and $B[k', j] = x_2$, and thus $(A * B)[i, j] \leq w_1 + w_2$.
\end{proof}

We finally characterize the complexity of node-weighted APSP in graphs with small weights $\set{0, \dots, u}$ (conditioned on the Strong APSP Hypothesis). The characterization of directed graphs is easy: Even for unweighted directed graphs (i.e., $u = 1$) we cannot beat time $O(n^{2.5})$ without breaking the Strong APSP Hypothesis (\cref{thm:strong-apsp-implies-dir-unw-apsp}).

For node-weighted graphs, the fastest-known APSP algorithm runs in time $\tilde O(\min\set{n^{(3+\omega)/2}, n^\omega u})$ (up to improvements from rectangular matrix multiplication), by combining the $\tilde O(n^{(3+\omega)/2})$-time algorithm for general node weights by~\cite{AbboudFJVX25}, with Shoshan and Zwick's $\tilde O(n^\omega u)$-time algorithm for small edge weights~\cite{ShoshanZ99}. We show that this time is optimal (if $\omega = 2$), conditioned on the Strong APSP Hypothesis:

\begin{theorem}[Node-Weighted APSP with Small Weights] \label{thm:node-wgt-apsp-sm}
Let $\delta \geq 0$. APSP in undirected node-weighted graphs with weights in $\set{0, \dots, n^\delta}$ cannot be solved in time $O(n^{\min\set{2.5, 2+\delta}-\epsilon})$ (for any constant $\epsilon > 0$), unless the Strong APSP Hypothesis fails.
\end{theorem}
\begin{proof}
The statement is trivially true for $\delta = 0$ (as it is impossible for any subquadratic-time algorithm to write down the $n^2$ pairwise distances). So suppose that $\delta > 0$, and choose $\alpha = \min\set{1/2, \delta} > 0$.

We show that an $O(n^{\min\set{2.5, 2+\delta}-\epsilon}) = O(n^{2+\alpha-\epsilon})$-time algorithm for undirected node-weighted APSP implies that $\MinPlus(n, n^\alpha, n \mid u \leq n^\alpha) = O(n^{2+\alpha-\epsilon})$, which contradicts the Strong APSP Hypothesis by \cref{cor:min-plus-rect-sm-univ}. For this reduction we mimic the proof of \cref{thm:node-wgt-apsp} almost exactly. Specifically, given matrices $A, B$ of sizes $n \times n^\alpha$ and $n^\alpha \times n$ with entries bounded by $u = n^\alpha$, we construct the same undirected graph on the four vertex parts $I = [n], K_1 = [n^\alpha] \times X, K_2 = [n^\alpha] \times X, J = [n]$, where $X = \set{0, \dots, u}$ denotes the set of weights. Analogously, the number of nodes is at most $|I|, |J| \leq n$ and $|K_1|, |K_2| \leq O(n^{2\alpha}) \leq O(n)$, and the $I$-$J$-distances encode the min-plus product $A * B$. Notably, this time the node weights are bounded by $O(u) = O(n^\alpha)$. Therefore, we can solve this constructed instance in time $O(n^{2+\alpha-\epsilon})$, leading to the desired contradiction.
\end{proof}

\subsection{Min Product} \label{sec:intermediate:sec:min-prod}
Min Product is a particularly simple intermediate matrix product problem, equivalent to the restriction of Min-Plus Product where one matrix consists only of entries $\set{0, \bot}$. This problem was studied before in~\cite{AbboudFJVX25} (under the name \emph{Boolean} Min-Plus Product). While perhaps not as important in its own right, one can easily conclude the hardness of other intermediate problems from the hardness of Min Product.

\begin{definition}[Min Product] \label{def:min-prod}
The \emph{min product} of two matrices $A \in \Int^{n \times n}$ and $B \in \set{0, 1}^{n \times n}$ is the matrix $C \in \Int^{n \times n}$ defined by
\begin{equation*}
    C[i, j] = \min_{\substack{k \in [n]\\B[k, j] = 1}}A[i, k].
\end{equation*}
The \emph{Min Product} problem is to compute the min product of two given matrices $A, B$.
\end{definition}

\begin{lemma}[Min Product] \label{lem:min-prod}
Min Product cannot be solved in time $O(n^{2.5-\epsilon})$ (for any constant $\epsilon > 0$), unless the APSP Hypothesis fails.
\end{lemma}
\begin{proof}
Suppose that Min Product can be solved in time $O(n^{2.5-\epsilon})$ for some $\epsilon > 0$; we show that in this case $\MinPlus(n, \sqrt{n}, n \mid D \leq \sqrt{n}) = O(n^{2.5-\epsilon})$ which contradicts the APSP Hypothesis by \cref{cor:min-plus-rect-unif}.

Let $A, B$ denote the two given matrices of sizes $n \times \sqrt{n}$ and $\sqrt{n} \times n$ (assuming that~$n$ is a square number at the cost of doubling~$n$ in the worst case), and let $X \subseteq \Int$ be the set of entries appearing in $A$ and $B$; we assume that $|X| \leq \sqrt{n}$. Our goal is to compute the min-plus product $A * B$. Construct the matrices $n \times n$ matrices $A', B'$ as follows, where we index the inner dimension by pairs $[\sqrt n] \times X$:
\begin{align*}
    A'[i, (k, x)] &= A[i, k] + x, \qquad \\
    B'[(k, x), j] &=
    \begin{cases}
        1 &\text{if $B[k, j] = x$,} \\
        0 &\text{otherwise.}
    \end{cases}
\end{align*}
Note that the min product $C'$ of $A'$ and $B'$ satisfies that
\begin{equation*}
    C'[i, j] = \min_{\substack{(k, x) \in [\sqrt{n}] \times X\\B'[(k, x), j] = 1}} A'[i, (k, x)] = \min_{\substack{(k, x) \in [\sqrt{n}] \times X\\B[k, j] = x}} (A[i, k] + x) = \min_{k \in [\sqrt{n}]} (A[i, k] + B[k, j]),
\end{equation*}
and thus $C' = A * B$ is the desired min-plus product. The matrices $A'$ and $B'$ can be constructed in time~$\tilde O(n^2)$, so the total time is dominated by the Min Product computation in time $O(n^{2.5-\epsilon})$ as claimed.
\end{proof}

\subsection{Min-Max Product} \label{sec:intermediate:sec:min-max-prod}
Next, we consider the Min-Max Product problem which was studied in the context of All-Pairs Bottleneck Paths and All-Pairs Nondecreasing Paths. We show that Duan and Pettie's $\tilde O(n^{(3+\omega)/2})$-time algorithm~\cite{DuanP09} is best-possible (if $\omega = 2$), conditioned on the APSP Hypothesis.

\begin{definition}[Min-Max Product]
The \emph{min-max product} of two matrices $A, B \in \Int^{n \times n}$ is the matrix $C \in \Int^{n \times n}$ defined by
\begin{equation*}
    C[i, j] = \min_{k \in [n]} \max\set{A[i, k], B[k, j]}.
\end{equation*}
The \emph{Min-Max Product} problem is to compute the min-max product of two given matrices $A, B$.
\end{definition}

\begin{lemma}[Min-Max Product] \label{lem:min-max-prod}
Min-Max Product cannot be solved in time $O(n^{2.5-\epsilon})$ (for any constant $\epsilon > 0$), unless the APSP Hypothesis fails.
\end{lemma}
\begin{proof}
The proof is almost immediate by \cref{lem:min-prod}: The min product of two matrices~$A$ and $B$ can be expressed as the min-max product of the matrices $A$ and $B'$, where $B'$ is obtained from~$B$ by replacing $1$ with $-\infty$ and $0$ by $\infty$.
\end{proof}

This lower bound entails various other lower bounds. By the known equivalence~\cite{VassilevskaWY07} we conclude our lower bound for All-Pairs Bottleneck Paths (\cref{thm:apbp}). Moreover, from known reductions~\cite{Vassilevska10} it follows that also the Min-$\leq$ Product and the All-Pairs Nondecreasing Paths problems cannot be solved in time $O(n^{2.5-\epsilon})$ unless the APSP Hypothesis fails, matching their respective upper bounds if $\omega = 2$~\cite{Vassilevska10,DuanGZ18,DuanJW19}. \cref{lem:min-max-prod} also implies a matching lower bound for the All-Edges Monochromatic Equality Triangle problem studied in~\cite{VassilevskaX20b}.

\subsection{Min-Equality Product} \label{sec:intermediate:sec:min-eq-prod}
Min-Equality Product is yet another matrix problem with intermediate complexity $\tilde O(n^{(3+\omega)/2})$, first introduced by Vassilevska W.\ and Xu in~\cite{VassilevskaX20b}. An interesting aspect of the problem is that it naturally generalizes the Equality Product problem which is one of the few intermediate problems for which we currently cannot prove conditional lower bounds.

\begin{definition}[Min-Equality Product] \label{def:min-eq-prod}
The \emph{min-equality product} of two matrices $A, B \in \Int^{n \times n}$ is the matrix $C \in \Int^{n \times n}$ defined by
\begin{equation*}
    C[i, j] = \min\set{A[i, k] : k \in [n], A[i, k] = B[k, j]}.
\end{equation*}
The \emph{Min-Equality Product} problem is to compute the min-equality product of two given matrices $A, B$.
\end{definition}

\begin{lemma}[Min-Equality Product] \label{lem:min-eq-prod}
Min-Equality Product cannot be solved in time $O(n^{2.5-\epsilon})$ (for any constant $\epsilon > 0$), unless the APSP Hypothesis fails.
\end{lemma}
\begin{proof}
Suppose that Min-Equality Product is in time $O(n^{2.5-\epsilon})$ for some $\epsilon > 0$. We show that in this case $\MinPlus(n, \sqrt{n}, n \mid D \leq \sqrt{n}; K \leq n^{\epsilon/4}) = O(n^{2.5-\epsilon/2})$ which contradicts the APSP Hypothesis by \cref{cor:min-plus-rect-low-doubling} (applied with parameters $\epsilon' = \epsilon/2, \kappa = \epsilon/4$ and $\alpha_1 = 1, \alpha_2 = 1/2, \alpha_3 = 1$).

Let $n$ be a square number, let $A, B$ be the given matrices of sizes $n \times \sqrt{n}$ and $\sqrt{n} \times n$, respectively, and let $X \subseteq \Int$ denote the set of entries that appear in $A$ and $B$. We assume that $|X| \leq \sqrt{n}$ and $|X + X| \leq K |X|$ for some $K \leq n^{\epsilon/4}$. First compute $Z = X - X$ (e.g., by brute-force), and recall that by the Plünnecke-Ruzsa inequality (\cref{lem:pluennecke-ruzsa}) we have $|X - X| \leq K^2 |X| \leq n^{1/2+\epsilon/2}$.

Construct the augmented matrices $A', B'$ indexed by $[n] \times ([\sqrt{n}] \times Z)$ and $([\sqrt{n}] \times Z) \times [n]$, respectively, defined by
\begin{align*}
    A'[i, (k, z)] &= 2A[i, k] - z, \\
    B'[(k, z), j] &= 2B[k, j] + z.
\end{align*}
We compute their min-equality product $C$ and claim that $C = A * B$ is exactly the desired min-plus product. On the one hand, we have that $A'[i, (k, z)] = B'[(k, z), j]$ if and only if $z = A[i, k] - B[k, j]$. In this case we have $A'[i, (k, z)] = 2A[i, k] - z = A[i, k] + B[k, j]$, and thus $C \geq A * B$. On the other hand, focus on any witness $(A * B)[i, j] = A[i, k] + B[k, j]$. Then clearly $z = A[i, k] - B[k, j] \in X - X = Z$, and thus by construction $C[i, j] \leq A[i, k] + B[k, j]$, and therefore $C \leq A * B$. The claim follows.

Finally, consider the running time. The matrices $A', B'$ have size $n \times n^{1+\epsilon/2}$ and $n^{1+\epsilon/2} \times n$, respectively. We can partition the inner dimension into $O(n^{\epsilon/2})$ blocks of size $n$ and thereby compute the rectangular min-equality product of $A'$ and $B'$ by $n^{\epsilon/2}$ square computations, each running in time $O(n^{2.5-\epsilon})$. The total time is $O(n^{2.5-\epsilon/2})$ as claimed.
\end{proof}

\subsection{Bounded-Difference and Monotone Min-Plus Product} \label{sec:intermediate:sec:min-plus-bd}
Next, we consider the restriction of the Min-Plus Product problem to \emph{bounded-difference} or, more generally, \emph{monotone} matrices defined as follows.

\begin{definition}[Bounded Difference]
A matrix $A \in \Int^{n \times n}$ is \emph{row-bounded-difference} if $|A[i, j] - A[i, j+1]| \leq O(1)$ for all $i, j$. Similarly, $A$ is \emph{column-bounded-difference} if $|A[i, j] - A[i+1, j]| \leq O(1)$ for all $i, j$. $A$ is \emph{bounded-difference} if it is both row- and column-bounded-difference.
\end{definition}

\begin{definition}[Monotone]
A matrix $A \in \Int^{n \times n}$ is \emph{row-monotone} if $0 \leq A[i, 1] \leq \dots \leq A[i, n] \leq O(n)$ for all rows $i$. Similarly, $A$ is \emph{column-monotone} if $0 \leq A[n, j] \leq \dots \leq A[1, j] \leq O(n)$ for all columns $j$. $A$ is \emph{monotone} if it is both row- and column-monotone.
\end{definition}

Bringmann, Grandoni, Saha, and Vassilevska W.~\cite{BringmannGSV19} designed the first subcubic-time algorithm for bounded-difference Min-Plus Product, which they used to derive subcubic-time algorithms for various string problems like language edit distance and RNA-folding. Mao~\cite{Mao21} later similarly applied their algorithm to unweighted tree edit distance. Motivated by these applications, a line of research~\cite{VassilevskaX20a,GuPVX21,ChiDX22,ChiDXZ22} sought to optimize the running time of bounded-difference and monotone Min-Plus Product. This line culminated only recently in an $\tilde O(n^{(3+\omega)/2})$-time algorithm due to Chi, Duan, Xie, and Zhang~\cite{ChiDXZ22}. Specifically, they design two different algorithms---one to compute the min-plus product of a \emph{row-monotone} matrix $A$ and an unconstrained matrix~$B$, and another one for a \emph{column-monotone} matrix $A$ and an unconstrained matrix $B$. It has remained open if these algorithms are best-possible. In this section we show that the first algorithm is conditionally optimal (if $\omega = 2$); it remains an interesting open question if the second algorithm is similarly optimal.

\thmminplusbd*

The conceptual idea behind \cref{thm:min-plus-bd} is simple. For a reduction based on the Directed Unweighted APSP Hypothesis, i.e., starting from Min-Plus Product of size $n \times \sqrt n \times n$ with entries bounded by $\sqrt n$, one can readjust the matrix entries $A[i, k] \gets A[i, k] + k \cdot \sqrt n$ and $B[k, j] \gets B[k, j] - k \cdot \sqrt n$, which leaves $A * B$ unchanged, and ensures that $A$ is row-monotone and $B$ is column-monotone. This insight was independently communicated to us by Łukasiewicz~\cite{Lukasiewicz25}. To make $A$ row-bounded-difference one can additionally blow up the inner dimension from $\sqrt n$ to $n$, by adding $\sqrt n$ new columns between any pair of consecutive columns in~$A$ to appropriately interpolate between the previous entries.

The major technical complication is that we aim to base the reduction on the APSP Hypothesis instead. We thus cannot assume that the entries in the given matrices are bounded by $\sqrt n$, but only that there are~$\sqrt n$ distinct entries that form a low-doubling set (by \cref{cor:min-plus-rect-low-doubling}). We show that after an appropriate transformation we can substitute numbers by their rank (i.e., position in the sorted order), and thereby still reduce to matrices with entries in the range roughly $\set{0, \dots, \sqrt n}$. We give the details in the following technical lemma. In this first step we only care to make $A$ row-monotone; later we comment how to achieve the other properties without loss of generality.

\begin{lemma} \label{lem:min-plus-bd}
The min-plus product of a row-bounded-difference $A \in [\sqrt{n}]^{n \times n}$ with a matrix $B \in \set{0, \bot}^{n \times n}$ containing at most $n^{1.5}$ non-$\bot$ entries cannot be computed in time $O(n^{2.5-\epsilon})$ (for any constant $\epsilon > 0$), unless the APSP Hypothesis fails.
\end{lemma}
\begin{proof}
Assume that there is an $O(n^{2.5-\epsilon})$-time algorithm for the problem described in the lemma statement. We design an algorithm for the uniform regular low-doubling Min-Plus Product problem of size $n \times \sqrt{n} \times n$, proving $\MinPlus(n, \sqrt{n}, n \mid D \leq \sqrt{n}; \rho \leq 1/D; K \leq n^{\epsilon/6}) \leq O(n^{2.5-\epsilon/6})$. This contradicts the APSP Hypothesis by \cref{cor:min-plus-rect-low-doubling}.

Let $n$ be a square number, and let $A$ and $B$ be the given $n \times \sqrt{n}$ and $\sqrt{n} \times n$ matrices with integer entries in~$X$. That is, we can assume that $D = |X| \leq \sqrt{n}$ and $|X + X| \leq K |X|$ where $K \leq n^{\epsilon/6}$. We compute~$X + X$ (by brute-force, say). Then we sort $X$ and $X + X$. We say that $y \in X$ is the \emph{successor} of~\makebox{$x \in X$} if it is the next-larger element in $X$. We also define two functions
\renewcommand\select{\operatorname{select}}
\renewcommand\rank{\operatorname{rank}}
\begin{align*}
    \select &: [|X + X|] \to X + X, \\
    \rank &: X \times X \to [|X + X|],
\end{align*}
where $\select$ maps each index $i$ to the $i$-th smallest element in the sumset $X + X$, and $\rank(x, y)$ is exactly the index $i$ such that $x + y$ is the $i$-th smallest element in the sumset $X + X$. In particular, we have that $\select(\rank(x, y)) = x + y$.

We now define another function $f : X \times X \to [|X + X|]$. Let $L > 0$ be a parameter to be determined later. Initially, set $f(x, y) \gets \rank(x, y)$. Now enumerate all pairs $(x, y) \in X^2$ in lexicographically descending order. If $y$ has a successor $y'$ and if $f(x, y) < f(x, y') - L$, then we update $f(x, y) \gets f(x, y') - L$. The resulting function $f$ satisfies the following three properties by construction:
\begin{enumerate}[label=(\roman*)]
    \item For each fixed $x \in X$ the function $f(x, \cdot)$ is increasing and $L$-bounded-difference, i.e., for any successive $y, y' \in X$ we have that $f(x, y') - f(x, y) \leq L$.
    \item $f(x, y) \geq \rank(x, y)$ for all $x, y \in X$.
    \item Call a pair $(x, y) \in X^2$ \emph{good} if $f(x, y) = \rank(x, y)$ and \emph{bad} otherwise. Let $R \subseteq X^2$ be the subset of bad pairs. Then $|R| \leq K |X|^2 / L$. Indeed, fix $x \in X$ and consider the increasing sequence $f(x, \cdot)$. For each bad pair $(x, y)$, this sequence ``jumps'' from $f(x, y)$ to $f(x, y')$ by $L$. However, the range of $f$ is~$[|X + X|]$, hence the number of such jumps is at most $|X + X| / L \leq K |X| / L$. Summing over all $|X|$ choices of $x$, we obtain the total bound $|R| \leq K |X|^2 / L$.
\end{enumerate}
The algorithm now proceeds in two cases:
\begin{itemize}
    \item \emph{(Good Case)} Let us call a pair $(i, j)$ \emph{good} if there is some witness $k$ such that $(A[i, k], B[k, j]) \in X^2$ is good, and \emph{bad} otherwise. Our goal in this step is to compute a matrix $C \geq A * B$ that correctly solves all good entries (i.e., $C[i, j] = (A * B)[i, j]$ whenever $(i, j)$ is good). To this end, construct the matrices~$A'$ indexed by $[n] \times ([\sqrt{n}] \times X)$ and~$B'$ indexed by $([\sqrt{n}] \times X) \times [n]$ as follows:
    \begin{align*}
        A'[i, (k, y)] &= f(A[i, k], y), \\
        B'[(k, y), j] &= 
    \begin{cases}
        0 &\text{if $B[k, j] = y$,} \\
        \bot &\text{otherwise.}
    \end{cases}
    \end{align*}
    Denote their min-plus product by $C' = A' * B'$. Now define the matrix $C$ by $C[i, j] = \select(C'[i, j])$. Then indeed it holds for all pairs $(i, j)$ that
    \begin{align*}
        C[i, j]
        &= \vphantom{\min_k} \select(C'[i, j]) \\
        &= \min_k \select(f(A[i, k], B[k, j])) \\
        &\geq \min_k \select(\rank(A[i, k], B[k, j])) \\
        &= \min_k (A[i, k] + B[k, j]),
    \end{align*}
    where the inequality is due to Property~(ii). For any good pair this inequality is an equality by~(iii), so the matrix $C$ is as claimed.

    It remains to describe how to compute the min-plus product $C' = A' * B'$. The matrix $A'$ is not row-bounded-difference, but it is close to that. Indeed, from Property~(i) we get that (1) along any row~$i$ there are $\sqrt{n} \cdot |X|$ pairs of adjacent entries $A'[i, (k, y)]$ and $A'[i, (k, y')]$ that differ by at most $L$. Moreover, (2) there are $\sqrt{n}$ pairs of adjacent entries of the form $A'[i, (k, y)]$ and $A'[i, (k', y')]$ (where~$k \neq k'$), and these entries can differ by up to~$|X + X|$. To turn $A'$ into a row-bounded-difference matrix we will now insert new columns to $A'$, filled with dummy entries whose only purpose is to interpolate the gaps described before. It suffices to add $\sqrt{n} \cdot |X| \cdot L$ columns to deal with the gaps of type (1), and $\sqrt{n} \cdot |X + X|$ columns to deal with the gaps of type (2). In total the number of new columns is $\sqrt{n} \cdot |X| \cdot L + \sqrt{n} \cdot |X + X| = (L + K) n$. For each column added in this way we also add a corresponding row in $B'$ filled with $\bot$-entries. Let $A'', B''$ denote the resulting matrices. This operation leaves the min-plus product unchanged, $A' * B' = A'' * B''$.

    \begin{framed}
        \emph{Running Time:} In summary we compute the min-plus product $A'' * B''$, where both matrices can be viewed as square matrices of size at most $O((K + L) n)$, where $A''$ is row-bounded-difference, has entries bounded by $|X + X| \leq K |X| \leq K \sqrt{n}$, and where $B''$ has at most $n^{1.5}$ non-$\bot$ entries. This is an instance of the problem described in the lemma statement of size $N = O((K^2 + L) n)$. Therefore, the running time is $O(N^{2.5-\epsilon}) = O((K^5 + L^{2.5}) n^{2.5-\epsilon})$.
    \end{framed}

    \item \emph{(Bad Case)} It remains to correct $C[i, j]$ for all bad pairs $(i, j)$. To this end enumerate all bad pairs $(x, y) \in R$, and all $k \in [\sqrt{n}]$. Then enumerate all $i \in [n]$ with $A[i, k] = x$ and all $j \in [n]$ with $B[k, j] = y$. We update $C[i, j] \gets \min\set{C[i, j], A[i, k] + B[k, j]}$. If $(i, j)$ is bad then by definition we will correctly compute $C[i, j]$ in this step.
    
    \begin{framed}
        \emph{Running Time:} $O(|R| \cdot \sqrt{n} \cdot \rho n \cdot \rho n) = O(K D^2 / L \cdot \sqrt{n} \cdot n/D \cdot n/D) = O(n^{2.5} K / L)$, recalling that each integer appears in at most a $\rho \leq 1/D$-fraction of all rows and columns in $A$ and $B$.
    \end{framed}
\end{itemize}
The correctness is clear from before. The total running time is $O((K^5 + L^{2.5}) n^{2.5-\epsilon} + n^{2.5} K / L)$. Recall that $K \leq n^{\epsilon/6}$ and choose the parameter $L = n^{\epsilon/3}$. The total time becomes $O(n^{2.5-\epsilon/6})$ as claimed.
\end{proof}

To conclude \cref{thm:min-plus-bd} from \cref{lem:min-plus-bd} we additionally need to argue that $A$ is row-monotone, and that $B$ is column-bounded-difference and column-monotone. It turns out that these three extra constraints are without loss of generality, based on a sequence of observations mostly from previous work; see~\cite{GuPVX21} and~\cite[Section 2]{ChiDXZ22}.

\begin{lemma}[\cite{GuPVX21,ChiDXZ22}] \label{lem:min-plus-bd-obs}
The computation of an $n \times n \times n$ min-plus product $A * B$ where $A$ is row-bounded-difference can be reduced in time $O(n^2)$ to the computation of an $n \times n \times n$ min-plus product $A' * B'$ where~$A'$ is row-bounded-difference and row-monotone and $B'$ is column-bounded-difference and column-monotone.
\end{lemma}
\begin{proof}
Let $A, B$ be the given matrices, where $A$ is row-bounded-difference, i.e., $|A[i, k] - A[i, k + 1]| \leq c$ for some constant $c$. The reduction is to apply the following four simple transformations step-by-step.
\begin{itemize}
    \item \emph{(Make $A$ Small-Universe)} When $A$ is row-bounded-difference, by definition each row in $A$ only takes integer values in some interval $\set{a, \dots, a + O(n)}$. By subtracting $a$ from the entire row we can make sure that $A$ only takes values in $\set{0, \dots, u = O(n)}$. We can recover the original min-plus product $A * B$ from this transformation by adding back the respective offset $a$ to each row.
    \item \emph{(Make $B$ Small-Universe)} Next, we also restrict the range of entries in $B$. Focus on any column in~$B$, and let $b$ denote its minimum entry. All entries larger than $b + 2u$ and all $\bot$ entries cannot matter in the min-plus product $A * B$, so we simply replace these entries by $b + 2u$. Afterwards we can remove $b$ from the entire column; the resulting matrix has integer entries in $\set{0, \dots, 2u = O(n)}$. We can recover the original min-plus product $A * B$ from this transformation by adding back the respective offset $b$ to each column.
    \item \emph{(Make $B$ Column-Bounded-Difference)} Next, to ensure that $B$ is column-bounded-difference we replace each entry $B[k, j]$ by $B'[k, j] = \min_{k'} (B[k', j] + |k' - k| \cdot c)$. We first claim that this transformation leaves the min-plus product unchanged, i.e., $A * B = A * B'$. The ``$\geq$'' direction is clear, so focus on the ``$\leq$'' direction. Suppose that $(A * B')[i, j] = A[i, k] + B'[k, j]$ and that $B'[k, j] = B[k', j] + |k' - k| \cdot c$. As $A$ is row-bounded-difference (with constant $c$), it follows that indeed
    \begin{align*}
        (A * B)[i, j]
        &\leq A[i, k'] + B[k', j] \\
        &\leq A[i, k] + |k - k'| \cdot c + B[k', j] \\
        &= A[i, k] + B'[k, j] \\
        &= (A * B')[i, j].
    \end{align*}
    Moreover, the augmented matrix $B'$ is indeed column-bounded-difference as $|B'[k, j] - B'[k+1, j]| \leq c$, and still consists of integer entries in the range $\set{0, \dots, O(n)}$.
    
    Finally, we remark that $B'$ can be computed in quadratic time from $B$. We will separately compute $B_L[k, j] = \min_{k' \leq k} (B[k', j] + (k - k') \cdot c)$ and $B_R[k, j] = \min_{k' \geq k} (B[k', j] + (k' - k) \cdot c)$, and then take their entry-wise minimum. But each column in $B_L$ can be computed in linear time by a dynamic program evaluating $B_L[k, j] = \min\set{B'[k, j], B_L[k - 1, j] + c}$, and similarly for $B_R$.
    
    In the following we will continue to write $B$ for the augmented matrix $B'$.
    
    \item \emph{(Make $A$ Row-Monotone and $B$ Column-Monotone)} We finally turn $A$ into a row-monotone matrix and~$B$ into a column-monotone matrix based on the following observation due to~\cite{GuPVX21}. Replace each entry $A[i, k]$ by $A[i, k] + kc$, and each entry $B[k, j]$ by $B[k, j] + (n - k)c$. The resulting matrices still have entries in $\set{0, \dots, O(n)}$. Each row in $A$ forms a monotonically non-decreasing sequence, recalling that before the transformation any two adjacent entries in the same row differed by at most $c$. Similarly, each column in $B$ forms a monotonically non-increasing sequence. Note that $A$ also remains row-bounded-difference (though with a larger constant,~$2c$), and similarly $B$ remains column-bounded-difference. The min-plus product $A * B$ remains unchanged, up to adding the fixed offset $n c$ to all entries. \qedhere
\end{itemize}
\end{proof}

The proof of \cref{thm:min-plus-bd} is now an immediate combination of \cref{lem:min-plus-bd} and \cref{lem:min-plus-bd-obs}.

\subsection{Min-Witness Product} \label{sec:intermediate:sec:min-witness-prod}
Finally, we consider the Min-Witness Product problem first introduced by Czumaj, Kowaluk, and Lingas~\cite{CzumajKL07}. It takes an important role in the class of intermediate-complexity problems as it is among the \emph{simplest} such problems: It can be solved in time $\tilde O(n^{2+\mu})$ by fast rectangular matrix multiplication, and is known to reduce to many known intermediate problems (via typically very simple reductions). On the flip-side this means that it is among the \emph{hardest} intermediate problems to establish tight lower bounds. To date, only non-matching lower bounds are known---an $n^{11/5-o(1)}$-time lower bound based on Strong APSP, and an $n^{7/3-o(1)}$-time lower bound based on Directed Unweighted APSP~\cite{ChanVX23}. Here we give a strengthened albeit still non-matching lower bound of $n^{7/3-o(1)}$ based on the APSP Hypothesis.

\begin{definition}[Min-Witness Product] \label{def:min-witness-prod}
The \emph{min-witness product} of two matrices $A, B \in \set{0, 1}^{n \times n}$ is the matrix $C \in ([n] \cup \set{\bot})^{n \times n}$ defined by
\begin{equation*}
    C[i, j] = \min\set{k \in [n] : A[i, k] = B[k, j] = 1}
\end{equation*}
(where we understand that the minimum is $\bot$ if there is no $k$ with $A[i, k] = B[k, j] = 1$). The \emph{Min-Witness Product} problem is to compute the min-witness product of two given matrices $A, B$.
\end{definition}

\begin{lemma}[Min-Witness Product] \label{lem:min-witness-prod}
Min-Witness Product cannot be solved in time $O(n^{7/3-\epsilon})$ (for any constant $\epsilon > 0$), unless the APSP Hypothesis fails.
\end{lemma}
\begin{proof}
Suppose that Min-Witness Product is in time $O(n^{7/3-\epsilon})$ for some $\epsilon > 0$. We show that in this case $\MinPlus(n, n^{1/3}, n \mid D \leq n^{1/3}) = O(n^{7/3-\epsilon})$ which contradicts the APSP Hypothesis by \cref{cor:min-plus-rect-unif}.

Let $A, B$ be the given matrices of respective sizes $n \times n^{1/3}$ and $n^{1/3} \times n$ (where we assume that $n^{1/3}$ is an integer), and let $X \subseteq \Int$ denote the set of entries appearing in $A$ and $B$ with size $|X| \leq n^{1/3}$. Construct the following two augmented matrices $A', B'$ indexed by $[n] \times ([n^{1/3}] \times X \times X)$ and $([n^{1/3}] \times X \times X) \times [n]$, respectively:
\begin{align*}
    A'[i, (k, x, y)] &=
    \begin{cases}
        1 &\text{if $A[i, k] = x$,} \\
        0 &\text{otherwise,}
    \end{cases} \\
    B'[(k, x, y), j] &=
    \begin{cases}
        1 &\text{if $B[k, j] = y$,} \\
        0 &\text{otherwise.}
    \end{cases}
\end{align*}
Moreover, define an arbitrary order $\leq$ on triples $[n^{1/3}] \times X \times X$ satisfying that $(k_1, x_1, y_1) \leq (k_2, x_2, y_2)$ whenever $x_1 + y_1 \leq x_2 + y_2$. Then compute the min-witness product of $A'$ and $B'$, where we order the inner dimension $[n^{1/3}] \times X \times X$ by that order $\leq$. By construction, in the resulting min-witness product each entry~$(i, j)$ is a triple $(k, x, y)$ where $(A * B)[i, j] = x + y$, and where $k$ is a witness of $(A * B)[i, j]$. In particular, we can easily read off the desired min-plus product $A * B$.
\end{proof}

\section*{Acknowledgements}
I am very grateful to Amir Abboud and Karl Bringmann for valuable discussions, to Amir Abboud and Leo Wennmann for many helpful comments on an earlier draft of this paper, and to Ernie Croot for answering some questions related to~\cite{AmirkhanyanBC18}.

\bibliographystyle{plainurl}
\bibliography{main}

\appendix
\section{Sum-Order-Preserving Hashing} \label{sec:hashing}
In this section, we recap the concept of sum-order-preserving hashing, provide a proof of \cref{thm:sum-order-preserving-quasi-poly}, and further discuss the plausibility of \cref{hypo:hashing}. Recall that a function $h : X \to \Int$ (for some integer set~$X$) is \emph{sum-order-preserving} if for all $x_1, x_2, y_1, y_2 \in X$:
\begin{equation*}
    x_1 + x_2 < y_1 + y_2 \quad\text{implies}\quad h(x_1) + h(x_2) < h(y_1) + h(y_2).
\end{equation*}
Analogously, we call $h$ \emph{order-preserving} if for all $x, y \in X$:
\begin{equation*}
    x < y \quad\text{implies}\quad h(x) < h(y).
\end{equation*}

\subsection{Quasi-Polynomial Bounds} \label{sec:hashing:sec:quasi-poly}
We will first focus on the proof of \cref{thm:sum-order-preserving-quasi-poly} (restated next) which claims that for any integer set $X$ with doubling $K$ there exists a large fraction, depending \emph{quasi-polynomially} on $K$, that can be hashed with a sum-order-preserving hash function.

\thmsumorderpreservingquasipoly*

The proof of this theorem follows by combining the work of Amirkhanyan, Bush, and Croot~\cite{AmirkhanyanBC18} with Sanders' state-of-the-art bounds~\cite{Sanders12} for the Freiman-Ruzsa theorem. To state these two ingredients, we start with some definitions. A set of the form
\begin{equation*}
    P = \set*{ \sum_{i=1}^d \ell_i a_i : |\ell_i| \leq L_i}
\end{equation*}
is called a (centered) \emph{generalized arithmetic progression}. Here, $d$ is called the \emph{dimension} of the progression. We say that $P$ is \emph{proper} if each element $p \in P$ has a unique representation of the form \smash{$\sum_{i=1}^d \ell_i a_i$} for $|\ell_i| \leq L_i$. Then Sanders' result~\cite{Sanders12} can be stated as follows; see also~\cite{Lovett15} for an exposition of this result.

\begin{theorem}[Quasi-Polynomial Freiman-Ruzsa~\cite{Sanders12}] \label{thm:freiman-quasi-poly}
Let $X$ be an integer set with $|X + X| \leq K |X|$. Then $X$ can be covered by $\exp((\log K)^{O(1)})$ translates of a proper generalized arithmetic progression with dimension $O((\log K)^{O(1)})$ and size at most $\exp((\log K)^{O(1)}) |X|$. 
\end{theorem}

The main result from~\cite{AmirkhanyanBC18} can be stated as follows.

\begin{theorem}[Condensing Lemma~\cite{AmirkhanyanBC18}] \label{thm:condensing}
Let $P$ be a generalized arithmetic progression of the form
\begin{equation*}
    P = \set*{ \sum_{i=1}^d \ell_i a_i : |\ell_i| \leq L_i}
\end{equation*}
such that
\begin{equation*}
    \set*{ \sum_{i=1}^d \ell_i a_i : |\ell_i| \leq 4 L_i}
\end{equation*}
is also proper. Then there are integers $a_1', \dots, a_d'$ such that the map $h : P \to \Int$ defined by
\begin{equation*}
    h\parens*{\sum_{i=1}^d \ell_i a_i} = \sum_{i=1}^d \ell_i a_i'
\end{equation*}
is order-preserving, and such that $|h(x)| \leq d^{O(d)} |P|$ for all $x \in P$.
\end{theorem}

\begin{proof}[Proof of \cref{thm:sum-order-preserving-quasi-poly}]
Let $X \subseteq \Int$ be arbitrary with doubling $K$. From \cref{thm:freiman-quasi-poly} it follows that~$X$ can be covered by $\exp((\log K)^{O(1)})$ translates of a proper generalized arithmetic progression $Q$ with dimension $d \leq O((\log K)^{O(1)})$ and size $\exp((\log K)^{O(1)}) |X|$. We express $Q$ as
\begin{equation*}
    Q = \set*{ \sum_{i=1}^d \ell_i x_i : |\ell_i| \leq L_i },
\end{equation*}
and let
\begin{equation*}
    P = \set*{ \sum_{i=1}^d \ell_i x_i : |\ell_i| \leq \floor*{\frac{L_i}{4}} }.
\end{equation*}
Note that $Q$ satisfies the conditions of \cref{thm:condensing}, so there are integers $x_1', \dots, x_d'$ such that the map $h' : Q \to \Int$ defined by
\begin{equation*}
    h'\parens*{\sum_{i=1}^d \ell_i x_i} = \sum_{i=1}^d \ell_i x_i'.
\end{equation*}
is order-preserving, and such that $|h'(x)| \leq d^{O(\log d)} |Q|$. Let
\begin{equation*}
    R = \set*{ \sum_{i=1}^d \ell_i x_i : |\ell_i| \leq \floor*{\frac{L_i}{8}} }.
\end{equation*}
We claim that the restriction of $h'$ to the domain $R$ is sum-order-preserving. Indeed, fix any $x_1, x_2, y_1, y_2 \in R$, and suppose that $x_1 + x_2 < y_1 + y_2$. Note that $x_1 + x_2$ is an element of $P$, and thus $h'(x_1) + h'(x_2) = h'(x_1 + x_2)$ by definition. The same holds for $y_1$ and $y_2$. Moreover, we have $h'(x_1 + x_2) < h'(y_1 + y_2)$ as $h'$ is order-preserving. Combining these facts we indeed have $h'(x_1) + h'(x_2) < h'(y_1) + h'(y_2)$. 

Next, observe that $O(1)^d = \exp((\log K)^{O(1)})$ translates of $R$ cover $P$, and in turn $\exp((\log K)^{O(1)})$ translates of $P$ cover $X$. Hence, for some translate $t$ the set $Y = (R + t) \cap X$ has size $|X| / \exp((\log K)^{O(1)})$. Let $h : Y \to \Int$ be defined by $h(x) = h'(x - t)$. The sum-order-preserving property is maintained under linear transformations, hence $h$ is sum-order-preserving. Moreover, recall that for all $x \in Y$ we have that
\begin{align*}
    |h(x)| = |h'(x - t)|
    &\leq d^{O(d)} |Q| \\
    &\leq d^{O(d)} O(1)^d |P| \\
    &\leq d^{O(d)} O(1)^d \exp((\log K)^{O(1)}) |X| \\
    &\leq \exp((\log K)^{O(1)}) |X|,
\end{align*}
and thus $h : Y \to \set{-K' |X|, \dots, K' |X|}$ for some $K' = \exp((\log K)^{O(1)})$.

A final cosmetic modification is to restrict the range of $h$ to $\set{0, \dots, |X|}$ as in the theorem statement. To this end partition $\set{-K' |X|, \dots, K' |X|}$ into $O(K')$ intervals of size at most $|X|$. Take the range with largest preimage in $Y$, and restrict $Y$ to that preimage. This preserves at least an $\exp((\log K)^{O(1)})$-fraction of the elements in $Y$ and the remaining map is clearly still sum-order-preserving.
\end{proof}

\subsection{Polynomial Bounds?} \label{sec:hashing:sec:poly}
It is an open question whether the quasi-polynomial bounds from the previous section can be improved to polynomial. This is exactly \cref{hypo:hashing}---the hypothesis that our full universe reduction for APSP is conditioned on:

\hypohashing*

In the following paragraphs we muse on a possible proof of \cref{hypo:hashing}.

It seems plausible that a proof of \cref{hypo:hashing} could work along the same lines as in the previous subsection. This would require two improvements: strengthening the Freiman-Ruzsa theorem as well as the condensing lemma to have polynomial bounds. The first is exactly the famous \emph{Polynomial Freiman-Ruzsa (PFR) Conjecture}. Its finite field analogue has recently been resolved in a celebrated paper by Gowers, Green, Manners, and Tao~\cite{GowersGMT25}. The second is not as well-studied. However, the proof of the condensing lemma relies on the same ``geometry of numbers'' technique that the state-of-the-art proofs for the Freiman-Ruzsa theorem also rely on~\cite{Sanders12}. For this reason it is reasonable to hope that a resolution of the PFR Conjecture might resolve our second issue along the way.

On a technical level, there is one more obstacle: The PFR Conjecture is known to be \emph{false} for generalized arithmetic progressions by a counterexample due to Lovett and Regev~\cite{LovettR18}. That is, we cannot expect that \cref{thm:freiman-quasi-poly} as it is stated here can be improved to polynomial bounds. Instead, the usual formulation of the PFR Conjecture involves more general objects called \emph{convex progressions}; see also~\cite{Manners17}. The consequence for us is that, in a hypothetical proof of \cref{hypo:hashing} along the lines we described before, the condensing lemma would also need to be generalized to convex progressions. Alternatively, this issue could be avoided by resorting to the following even weaker version of \cref{hypo:hashing}:

\begin{hypothesis} \label{hypo:hashing-weak}
For every $\epsilon > 0$ there is some constant $c \geq 0$ such that the following holds. For every integer set $X$ with doubling $|X + X| \leq K|X|$ there is a subset $Y \subseteq X$ of size $|Y| \geq \Omega(|X|^{1-\epsilon} / K^c)$ and a sum-order-preserving function $h : Y \to \set{0, \dots, |X|}$. Moreover, given $X$ one can compute $Y$ and $h$ in time $|X|^{3-\Omega_\epsilon(1)} K^{O_\epsilon(1)}$.
\end{hypothesis}

Indeed, when allowing a small polynomial loss of the form $|X|^\epsilon$ then Lovett and Regev's counterexample does no longer apply, and it is conceivable that the proof avoids convex progressions. \cref{hypo:hashing-weak} is still strong enough to imply all of our results implied by \cref{hypo:hashing}; see \cref{rem:hashing-weak}.

As a final note, recall that for us it is necessary to also \emph{compute} $Y$ and $h$ (as stated in \cref{hypo:hashing,hypo:hashing-weak}). This is perhaps the smallest concern---prior work on additive combinatorics in algorithm design has already ported many existential theorems to computational versions~\cite{ChanL15,BringmannW21,BringmannN21,AbboudBF23,RandolphW24,FischerJX25}, and most of these adaptations did not deviate significantly from their original proofs (a notable exception is~\cite{ChenMZ25}); relatedly, there has already been recent work on turning the finite-field PFR theorem~\cite{GowersGMT25} into an efficient algorithm~\cite{ArunachalamCDG26}.
\section{Derandomization of the Conflict-Free Covering Lemma} \label{sec:covering}
The purpose of this section is to complete the missing proof of \cref{lem:covering}. The overall idea is to derandomize the simple randomized algorithm described earlier in \cref{sec:exact-tri-low-rank} by the method of conditional expectations, however, to achieve a sufficiently efficient algorithm some technical details are necessary.

\lemcovering*
\begin{proof}
Let $\mathcal S \gets \emptyset$ be the collection of sets we are about to construct, and let $I \subseteq [n]$ denote the set of remaining items $i$ that are not covered by $\mathcal S$ (i.e., $I = \set{ i \in [n] : \nexists S \in \mathcal S : \text{$x_i \in S$ and $C_i \cap S = \emptyset$}}$). We will repeatedly construct sets $S$ to be inserted into $\mathcal S$. We describe how to construct the next such set~$S$. Initialize $S \gets \emptyset$, and consider the partition $I = B(S) \sqcup G(S) \sqcup U(S)$ into the remaining \emph{bad}, \emph{good}, and \emph{undecided} items defined by
\begin{align*}
    B(S) &= \set{i \in I : C_i \cap S \neq \emptyset}, \\
    G(S) &= \set{i \in I : \text{$C_i \cap S = \emptyset$ and $x_i \in S$}}, \\
    U(S) &= \set{i \in I : \text{$C_i \cap S = \emptyset$ and $x_i \not\in S$}}
\end{align*}
which we will maintain throughout the construction of $S$. Write
\begin{equation*}
    \phi(S) = |G(S)| - \frac{|B(S)|}{2s}
\end{equation*}
for the \emph{potential} of $S$. As we will prove shortly, we are always in one of two cases: (1) there is some $j \in [r]$ so that when inserting $j$ into $S$ the potential increases by at least $\phi(S \cup \set{j}) - \phi(S) \geq |I| / (8r)$, or (2) there are many good items, $|G(S)| \geq |I| / (16s)$. While we are in the first case we insert the respective element into $S$. When we are eventually in the second case we have completed the construction of $S$, so we update~\makebox{$\mathcal S \gets \mathcal S \cup \set{S}$}, update $I$ appropriately (i.e., remove all indices $i$ covered by $S$), and, while~\makebox{$I \neq \emptyset$}, proceed to the construction of the next set $S$.

This almost completes the description of the algorithm, but we have not described yet how to efficiently find such an element $j \in [r]$. We will now first analyze that the algorithm is correct, and later describe the missing efficient implementation.

\paragraph{Correctness}
First, it is clear that the algorithm terminates after at most $|\mathcal S| \leq O(s \log n)$ rounds. Indeed, in each round we decrease $|I|$ by at least $|G(S)| \geq |I|  / (16s)$, so after $O(s)$ rounds the size of $|I|$ has at least halved, and after $O(s \log n)$ rounds the remaining set $I$ must be empty.

It remains to prove the key claim that whenever $|G(S)| < |I| / (16s)$ we can find some $j \in [r]$ with a potential increase of at least $\phi(S \cup \set{j}) - \phi(S) \geq |I| / (8r)$. In the construction of $S$ the initial potential $\phi(\emptyset)$ is zero and it stays nonnegative throughout, thus we can bound the number of bad items by:
\begin{equation*}
    |B(S)| = 2s \cdot (|G(S)| - \phi(S)) \leq 2s \cdot |G(S)| \leq \frac{|I|}{8}.
\end{equation*}
It follows that almost all items are undecided:
\begin{equation*}
    |U(S)| = |I| - |B(S)| - |G(S)| \geq |I| - \frac{|I|}{8} - \frac{|I|}{16s} \geq \frac{3 |I|}{4}.
\end{equation*}
Now pretend that we sample an element $j \in [r]$ uniformly at random. On the one hand, for each undecided item $i \in U(S)$ the probability that $i$ becomes good in $S \cup \set{j}$ is exactly $1/r$. On the other hand, for each good item $i \in G(S) \sqcup U(S)$ the probability that $i$ becomes bad in $S \cup \set{j}$ is at most $s / r$, and thus:
\begin{equation*}
    \Ex\brackets*{|G(S \cup \set{j})| - |G(S)|} \geq \frac{|U(S)|}{r} - \frac{s |G(S)|}{r} \geq \frac{3|I|}{4r} - \frac{s |I|}{16rs} \geq \frac{5|I|}{8r}.
\end{equation*}
At the same time, the probability that any item $i$ becomes bad is at most $s / r$, hence
\begin{equation*}
    \Ex\brackets*{|B(S \cup \set{j})| - |B(S)|} \leq \frac{s |I|}{r}.
\end{equation*}
It follows that
\begin{equation*}
    \Ex[\phi(S \cup \set{j}) - \phi(S)] = \Ex\brackets*{|G(S \cup \set{j})| - |G(S)| - \frac{|B(S \cup \set{j})| - |B(S)|}{2s}} \geq \frac{5|I|}{8r} - \frac{s |I|}{2r s} = \frac{|I|}{8r}.
\end{equation*}
In particular, there exists an element $j \in [r]$ so that by adding $j$ to $S$ the potential would increase by at least $|I| / (8r)$.

\paragraph{Efficient Implementation}
We now turn to the problem of efficiently constructing the set $S$. Specifically, we give an efficient algorithm that, while $|G(S)| < |I| / (16s)$, finds some $j \in [r]$ with a potential increase of at least $\phi(S \cup \set{j}) - \phi(S) \geq |I| / (8r)$.

We rely on some preprocessing. If necessary increase $r$ to a power of $2$. Recall that a \emph{dyadic} interval~\makebox{$J \subseteq [r]$} is a set of the form $a2^\ell + \set{1, \dots, 2^\ell}$ for some $0 \leq \ell \leq \log r$ and $a \in \set{0, \dots, r / 2^\ell - 1}$. In the preprocessing phase we compute and store, for all $i \in [n]$ and all dyadic intervals $J \subseteq [r]$, the sizes $|C_i \cap J|$ and $|\set{x_i} \cap J|$. With access to this data we can, given any sets $G, U \subseteq [n]$ and any dyadic interval $J \subseteq [r]$, evaluate the following quantity in time $O(n)$:
\begin{equation*}
    E_J(G, U) = \frac{1}{|J|} \parens*{\sum_{i \in U} |\set{x_i} \cap J| - \sum_{i \in G} |C_i \cap J| - \frac{1}{2s} \sum_{i \in G \cup U} |C_i \cap J|}.
\end{equation*}
The key here is that $E_J(G(S), U(S)) = \Ex[\phi(S \cup \set{j}) - \phi(S)]$, where $j$ is sampled uniformly random from $J$. Indeed, by the same argument as before we have that
\begin{align*}
    E_{J}(G(S), U(S))
    &= \Ex\brackets*{\sum_{i \in U(S)} |\set{x_i} \cap \set{j}| - \sum_{i \in G(S)} |C_i \cap \set{j}| - \frac{1}{2s} \sum_{i \in G(S) \cup U(S)} |C_i \cap \set{j}|} \\
    &= \Ex\brackets*{|G(S \cup \set{j})| - |G(S)| - \frac{|B(S \cup \set{j})| - |B(S)|}{2s}} \\
    &= \Ex\brackets*{\phi(S \cup \set{j}) - \phi(S)}\vphantom{\brackets*{\frac{|}{s}}}.
\end{align*}

We return to the construction of $S$. If $|G(S)| < |I| / (16s)$ we have that $E_{[r]}(G(S), U(S)) \geq |I| / (8r)$ by the correctness argument from before. Moreover, note that for any dyadic intervals $J = J_1 \sqcup J_2$, $E_J(G, U)$ is exactly the average of $E_{J_1}(G, U)$ and $E_{J_2}(G, U)$. Based on this insight, we binary-search for $j$. We start with $J = [r]$. While $J$ can be split into $J = J_1 \sqcup J_2$ we update $J \gets J_1$ if $E_{J_1}(G(S), U(S)) \geq E_J(G(S), U(S))$, and update $J \gets J_2$ otherwise. Eventually we reach $J = \set{j}$, and then
\begin{equation*}
    \frac{|I|}{8r} \leq E_{\set{j}}(G(S), U(S)) = \phi(S \cup \set{j}) - \phi(S)
\end{equation*}
as claimed.

\paragraph{Running Time}
We finally analyze the running time of this algorithm. In the preprocessing step we consider the $n$ items and $O(r \log r)$ dyadic intervals; with a little care we can precompute all sizes $|C_i \cap J|$ and $|\set{x_i} \cap J|$ in time $O(n r \log r)$.

To analyze the main phase of the algorithm we first bound the size of the sets $S$. Recall that whenever we insert an element $j$ into some set $S$ then the potential $\phi(S)$ increases by at least $|I| / (8r)$. Hence, we have that
\begin{equation*}
    |G(S)| \geq \phi(S) \geq \frac{|S| \, |I|}{8r},
\end{equation*}
and thus after inserting at most $r / s$ elements into $S$ the algorithm halts the construction and proceeds to the next set $S$. That is, $|S| \leq r / s$. In particular, we query at most $O(|\mathcal S| \cdot r / s \cdot \log r)$ values $E_J(G, U)$ in total. Each query takes time $O(n)$, so the total query time is $O(|\mathcal S| \cdot r / s \cdot \log r \cdot n) = O(n r \log n \log r)$. All other bookkeeping steps run in negligible time $O(n r)$.
\end{proof}

\end{document}